\PassOptionsToPackage{hidelinks,colorlinks=true,allcolors=blue}{hyperref}
\documentclass[a4paper,twocolumn,11pt,accepted=2026-05-14]{quantumarticle}
\pdfoutput=1

\usepackage[utf8]{inputenc}
\usepackage[english]{babel}
\usepackage[T1]{fontenc}

\usepackage{graphicx} 
\usepackage{subfig}
\usepackage{orcidlink}

\usepackage{amsmath}
\usepackage{amssymb}
\usepackage{amsthm}
\usepackage{physics}
\usepackage{float}%
\usepackage[linesnumbered,ruled,vlined]{algorithm2e}
\RestyleAlgo{ruled}
\usepackage{ragged2e}
\usepackage{diagbox}
\usepackage{booktabs}
\usepackage{tikz}
\usepackage[numbers,sort&compress]{natbib}

\usepackage[numbers]{natbib}
\usepackage[hidelinks,colorlinks=true,allcolors=blue]{hyperref}
\usepackage{cleveref}

\newcommand*\circled[1]{\tikz[baseline=(char.base)]{
\node[shape=circle,draw,inner sep=2pt] (char) {#1};}}

\newtheorem{theorem}{Theorem}
\newtheorem{proposition}{Proposition}
\newtheorem{corollary}{Corollary}
\newtheorem{lemma}[theorem]{Lemma}
\theoremstyle{definition}
\newtheorem{definition}{Definition}

\DeclareMathOperator*{\argmin}{arg,min}
\DeclareMathOperator*{\argmax}{arg,max}
\def\doublesetminus{\setminus \!\! \setminus}

\newcommand{\qinc}{Quantum Innovation Centre (Q.InC), Agency for Science, Technology and Research (A*STAR), 2 Fusionopolis Way, Innovis \#08-03, Singapore 138634, Republic of Singapore\looseness=-1}

\newcommand{\sutd}{Science, Mathematics and Technology Cluster, Singapore University of Technology and Design, 8 Somapah Road, Singapore 487372, Republic of Singapore\looseness=-1}

\newcommand{\ihpc}{Institute of High Performance Computing (IHPC), Agency for Science, Technology and Research (A*STAR), 1 Fusionopolis Way, \#16-16 Connexis, Singapore 138632, Republic of Singapore\looseness=-1}

\newcommand{\cqct}{Centre for Quantum Computation and Communication Technologies (CQC2T), Department of Quantum Science and Technology, Research School of Physics, Australian National University, Acton 2601, Australia\looseness=-1}

\newcommand{\cqt}{Centre for Quantum Technologies (CQT), National University of Singapore, Singapore 117543, Republic of Singapore\looseness=-1}

\newcommand{\sit}{Engineering Cluster, Singapore Institute of Technology, 1 Punggol Coast Road, Singapore 828608, Republic of Singapore\looseness=-1}

\makeatletter
\renewcommand{\@printtitletextwithappropriatefontsize}{%
  \@titleatfontsize{\fontsize{19pt}{21pt}\selectfont}%
}
\makeatother

\begin{document}

\title{Near-Optimal Parameter Tuning of Level-1 QAOA for Ising Models}

\author{V Vijendran\,\orcidlink{0000-0003-3398-1821}}
\email{vjqntm@gmail.com}
\affiliation{\cqt}
\affiliation{\qinc}
\affiliation{\cqct}

\author{Dax Enshan Koh\,\orcidlink{0000-0002-8968-591X}}
\affiliation{\qinc}
\affiliation{\ihpc}
\affiliation{\sutd}
\affiliation{\sit}

\author{Eunok Bae\,\orcidlink{0000-0001-7531-0992}}
\affiliation{School of Computational Sciences, Korea Institute for Advanced Study (KIAS), Seoul 02455, Korea\looseness=-1}
\affiliation{Electronics and Telecommunications Research Institute (ETRI), Daejeon 34129, Korea
\looseness=-1}

\author{Hyukjoon Kwon\,\orcidlink{0000-0001-5520-0905}}
\affiliation{School of Computational Sciences, Korea Institute for Advanced Study (KIAS), Seoul 02455, Korea\looseness=-1}

\author{Ping Koy Lam\,\orcidlink{0000-0002-4421-601X}}
\affiliation{\qinc}
\affiliation{\cqt}
\affiliation{\cqct}

\author{Syed M Assad\,\orcidlink{0000-0002-5416-7098}}
\email{cqtsma@gmail.com}
\affiliation{\qinc}
\affiliation{\cqct}

\begin{abstract}
    \boldmath
    The Quantum Approximate Optimisation Algorithm (QAOA) tackles combinatorial optimisation problems by encoding their solutions into the ground state of an Ising Hamiltonian prepared by a $p$-level parameterised circuit, with the angles tuned classically. Parameter optimisation is widely regarded as a central bottleneck, even for the shallowest circuits. Focusing on QAOA at $p=1$ (QAOA$_1$), we show that tuning the two angles $(\gamma, \beta)$ for weighted Ising models is not a black-box search but a structured signal-processing problem. We prove that the QAOA$_1$ expectation value is a partial Fourier series in $\gamma$ whose frequencies are determined explicitly by the problem's couplings and fields, giving instance-wise bandwidth bounds and, via the Nyquist--Shannon theorem, the sampling resolution needed to avoid the aliasing that causes coarse-grid searches to return spurious optima. We then eliminate the mixer angle analytically, computing $\beta^*(\gamma)$ in closed form to reduce the search to one dimension, and apply a subdivision algorithm that locates the globally optimal $\gamma$ in polynomial time with a certificate of optimality when the weights are commensurable and bounded. For regular weighted graphs, we further prove the conventional wisdom that the globally optimal $\gamma^* \in \mathbb{R}^+$ concentrates near zero and coincides with the first local optimum, giving a rigorous account of the empirical success of small-angle initialisation and allowing gradient descent to replace exhaustive line searches. Validated within Recursive QAOA (RQAOA) on weighted instances of 128 and 256 qubits, our method consistently outperforms both coarsely optimised RQAOA and semidefinite programming.
\end{abstract}

\maketitle

\section{Introduction}

\begin{figure*}[!t]
    \centering
    \includegraphics[width=\textwidth]{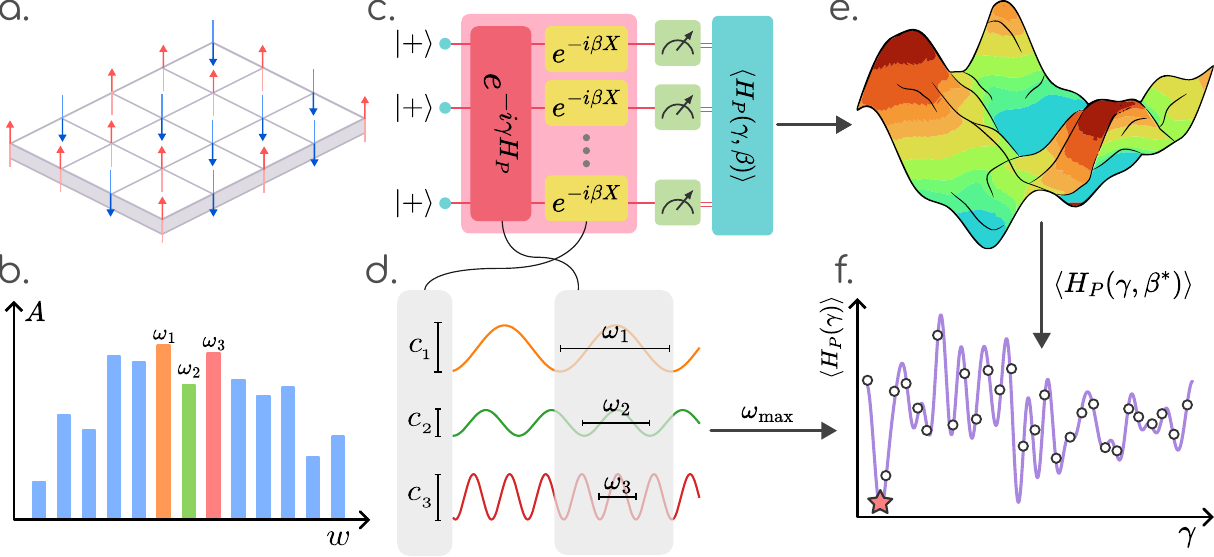}
    \caption{\textbf{Illustration of the Main Results.} \justifying This paper focuses on optimising the variational parameters of QAOA$_1$ for solving QUBO problems formulated as Ising models. (a) The Ising model is represented by the problem Hamiltonian $H_P$, which is diagonal in the Pauli-$Z$ basis. Its eigenvalues define (b) the spectrum of frequencies $\omega$ for the given problem instance, where each $\omega_i$ corresponds to the difference between eigenvalues. To solve the problem exactly, the quantum circuit must express the full spectrum of frequencies. However, (c) QAOA$_1$ is limited to representing only a small subset of the spectrum. The expectation value of the cost function in QAOA$_1$ can be expressed as (d) a truncated Fourier series whose frequencies are determined by the problem unitary, while its coefficients depend on the mixing unitary. At $p=1$, most coefficients vanish, significantly constraining the circuit’s expressivity. We provide an analytical method to compute the maximum frequency of the cost function along $\gamma$, enabling an appropriate sampling rate for optimisation. Previous studies derived analytical expressions for classically computing the expectation value $\langle H_P(\gamma, \beta) \rangle$, yielding (e) the cost landscape. In this work, we show how the two-dimensional optimisation over $(\gamma, \beta)$ can be reduced to (f) a univariate optimisation over $\gamma$, with $\beta^*$ determined analytically as a function of $\gamma$. By leveraging the maximum frequency, the univariate cost function can be efficiently sampled, enabling accurate estimation of the optimal $\gamma^*$ without overshooting the global optimum. Finally, we prove that, on average, the global optimum $\gamma^* \in \mathbb{R}^{+}$ is concentrated near 0 and that this global optimum coincides with the first local optimum.}
    \label{main_results_illustration}
\end{figure*}

Quantum computing is poised to revolutionise problem-solving in fields ranging from optimisation and machine learning to quantum chemistry and materials science. While fault-tolerant quantum computers hold the promise of implementing groundbreaking algorithms like Shor’s for factoring~\cite{shor1999polynomial} and Grover’s for unstructured search~\cite{grover1997quantum}, current quantum hardware falls short of the error correction required to implement these algorithms at scale. Instead, the era of Noisy Intermediate-Scale Quantum (NISQ) devices~\cite{preskill2018quantum,cheng2023noisy} has given rise to variational quantum algorithms (VQAs), which blend parameterised quantum circuits (PQCs) with classical optimisation to tackle complex problems~\cite{cerezo2021variational, bharti2022noisy}. By iteratively refining circuit parameters, VQAs navigate vast quantum state spaces, leveraging their adaptability to address diverse computational challenges. However, their success depends on efficient parameter optimisation, which is complicated by the intricate optimisation landscapes of PQCs~\cite{bittel2021training, mcclean2018barren, anschuetz2022quantum, wang2021noise, cerezo2021cost, ortiz2021entanglement, nemkov2024barren}.

Among VQAs, the Quantum Approximate Optimisation Algorithm (QAOA)~\cite{farhi2014quantum} has emerged as a leading contender for solving combinatorial optimisation problems. QAOA translates these problems into Hamiltonians and approximates their ground states using circuits parameterised by $2p$ variables, where $p$ represents the number of alternating layers of problem and mixing unitaries. While increasing $p$ improves solution quality, it also amplifies the computational burden of parameter optimisation and is further hindered by physical constraints such as gate noise, limited qubit connectivity, and state-preparation-and-measurement errors, effectively limiting QAOA implementations to shallow depths on NISQ devices~\cite{weidenfeller2022scaling}. Even in this regime, optimising QAOA parameters remains a significant challenge, particularly for large, dense, asymmetric, or weighted problem instances. Addressing these challenges has inspired a growing body of research focused on heuristic methods to identify near-optimal parameters efficiently, with the goal of unlocking QAOA's full potential in shallow-depth applications.

The search for effective initialisation and optimisation strategies for QAOA, aimed at minimising the number of optimisation steps required for solving combinatorial problems, remains an active area of research. Effective initialisation ensures the algorithm begins close to a local or global optimum in the parameter space, while efficient optimisation enables smooth and rapid traversal through the landscape. Existing strategies to address this challenge range from physics-based approaches to advanced machine learning techniques, including transfer learning and data-driven models. Early work focused on deriving analytical expressions for optimal parameters at $p=1$ providing performance bounds for specific graph structures such as unweighted $D$-regular triangle-free graphs and large-girth regular graphs; however, these derivations are less effective for complex, weighted graphs encountered in practical settings~\cite{hadfield2018quantum, wang2018quantum, ozaeta2022expectation, vijendran2024expressive, farhi2014quantum, basso2021quantum,ng2024analytical, fixed_angle_dreg}. Parameter transfer methods, inspired by transfer learning, leverage optimised parameters from similar problem instances, particularly for uniform-structure problems like MaxCut on unweighted $D$-regular graphs; however, they are less reliable for graphs without well-defined distributions or regular structure~\cite{brandao2018fixed, akshay2021parameter, katial2024instance, galda2021transferability, galda2023similarity, brundin2024symmetry, langfitt2023parameter, sakai2024linearly}. Quantum annealing insights also guide parameter initialisation by setting parameters based on a fraction of the optimal annealing time, with additional refinements from control theory techniques like the bang-bang protocol and Lyapunov control, yet these are primarily suited to deeper circuits and struggle with capturing optimal transitions in shallow QAOA circuits~\cite{sack2021quantum, liang2020investigating, magann2022lyapunov}. Another approach leverages an optimised level-$p$ circuit to initialise a level-$(p+1)$ circuit, using methods like linear interpolation or Fourier transformations to reduce optimisation steps, although full optimisation of the level-$p$ circuit is typically required~\cite{zhou2020quantum}. Efficiency-oriented techniques such as multilevel leapfrogging~\cite{ni2023multilevel} and depth-progressive initialisation~\cite{lee2023depth} reduce computational costs by incrementally optimising parameters, but these methods often yield suboptimal solutions and generally lack performance guarantees, particularly for real-world weighted problems where structure varies widely~\cite{shaydulin2023parameter, stkechly2023connecting}. Finally, various machine learning approaches~\cite{alam2020accelerating, xie2023quantum, liang2024graph, meng2024parameter, amosy2022iterativefreequantumapproximateoptimization, khairy2019reinforcement, montanez2024transfer, falla2024graph, cheng2024quantum} have been proposed to predict and generate near-optimal initial guesses for QAOA parameter optimisation. However, these methods often incur high computational costs and also lack performance guarantees.

This work addresses the challenge of optimising QAOA$_1$ for QUBO problems, formulated as Ising models. Despite QAOA$_1$ involving only two parameters, $(\gamma, \beta)$, the widely held belief that coarse grid searches followed by local minimisation suffice is misleading. We reveal that the optimisation landscape of QAOA$_1$ is highly oscillatory, with oscillation rates increasing with problem size, density, and weight variability. These oscillations necessitate high-resolution grid searches to avoid landscape distortion and suboptimal parameter estimates. To mitigate the computational cost of grid searches at fine resolutions, we propose an efficient parameter-tuning strategy that reduces the two-dimensional search over $(\gamma, \beta)$ to a one-dimensional search over $\gamma$, with the optimal $\beta^*$ derived analytically. We establish the minimum sampling frequency required to capture the oscillatory landscape accurately and provide an algorithm to estimate the optimal parameters in polynomial time. For sufficiently large instances of regular graphs, we rigorously prove the globally optimal $\gamma^* \in \mathbb{R}^+$ concentrate very close to 0 and that this global optimum coincides with the first local optimum. This insight enables the replacement of exhaustive line searches with gradient descent, dramatically simplifying the optimisation process. Our approach is provably near-optimal---if not optimal---problem-agnostic and effectively addresses the challenges associated with parameter tuning in QAOA$_1$ for QUBO problems.

The utility of our method extends beyond QAOA$_1$ itself, as it can be applied to quantum algorithms that use QAOA$_1$ as a subroutine for solving QUBO problems~\cite{bravyi2020obstacles, kim2024distributed, yue2023local, chen2024noise, finvzgar2024quantum, brady2023iterative, bach2024mlqaoa, zhou2023qaoa, ponce2023graph, dupont2023quantum, dupont2024extending, caha2022twisted, fischer2024role}. Additionally, the optimised parameters for $p=1$ can serve as warm-start values for deeper QAOA circuits ($p>1$). Instead of initialising all parameters randomly, one can utilise the near-optimal $p=1$ parameters and only randomly initialise the additional parameters required for higher depths. This approach provides a better starting point for the optimiser, enhancing convergence efficiency compared to traditional random initialisation methods.

We validate our strategy by applying it to the Recursive QAOA (RQAOA)~\cite{bravyi2020obstacles}, benchmarking its performance on dense and weighted QUBO instances. Our results show that, unlike coarsely optimised RQAOA$_1$, which struggles to outperform semidefinite programming (SDP), our method consistently enables RQAOA$_1$ to surpass SDP, coarsely optimised RQAOA$_1$, and QAOA$_1$ across all tested instances. These findings highlight the critical importance of parameter optimisation in improving variational quantum algorithm performance, demonstrating the potential of quantum methods to outperform state-of-the-art classical techniques in producing high-quality approximations.

The paper is structured as follows: `Preliminary Material' introduces the foundational concepts, followed by `Theoretical Results' and `Numerical Results', which present our primary findings. The `Discussion' section delves into the broader implications of these results. Comprehensive details on the benchmarked QAOA variants are available in the `Methods' section, and all proofs and additional supporting information are included in the supplementary materials.



\section{Preliminary Material}

\subsection{Ising Models}

The Ising model was introduced to describe ferromagnetic interactions in materials. In this model, binary variables, called spins, take values of $\pm 1$ to represent two possible states of a particle. Spins interact pairwise through coupling strengths and may also experience external influences, such as a magnetic field. The goal is to find the ground state configuration of spins that minimises the system's total energy.

Mathematically, the Ising model is defined by a set of spin variables $s_{j} \in \{-1 , +1\}$ for $j = 1,2, \ldots, n$, which interact via a symmetric coupling matrix $\mathbf{J} \in \mathbb{R}^{n \times n}$. External magnetic fields are represented by a vector $\mathbf{h} \in \mathbb{R}^n$. The energy of a configuration $\mathbf{s} = (s_1, s_2, \ldots, s_n)$ is given by the Hamiltonian:
\begin{equation}
\begin{aligned}
H(\mathbf{s})
&= \sum_{j=1}^n \sum_{k=1}^n J_{jk}s_js_k
+ \sum_{j=1}^n h_js_j \\
&= \mathbf{s}^\mathrm{T}\mathbf{J}\mathbf{s}
+ \mathbf{h}^\mathrm{T}\mathbf{s}.
\end{aligned}
\label{ising_eqn}
\end{equation}
The goal is to find a configuration $\mathbf{s}^* \in \{-1,+1\}^n$ that minimises $H(\mathbf{s})$. This optimisation task parallels the QUBO~\cite{punnen2022quadratic} problem, which involves minimising a quadratic function over binary variables.

At its core, a QUBO problem seeks a binary vector $\mathbf{x} \in \{0, 1\}^n$ that minimises the quadratic objective: 
\begin{equation} 
f(\mathbf{x}) = \mathbf{x}^\mathrm{T} \mathbf{A} \mathbf{x} + \mathbf{b}^\mathrm{T} \mathbf{x}, \end{equation} 
where $\mathbf{A} \in \mathbb{R}^{n \times n}$ is a symmetric matrix capturing the quadratic interactions between variables, and $\mathbf{b} \in \mathbb{R}^n$ represents the linear coefficients. For the problem to remain non-trivial, $\mathbf{A}$ must contain negative eigenvalues; otherwise, setting all binary variables to zero yields a trivial optimal solution of zero. The binary nature of the variables implies that $x_j^2 = x_j$, which allows the linear term $\mathbf{b}^\mathrm{T} \mathbf{x}$ to be absorbed into the quadratic form by adjusting the diagonal elements of $\mathbf{A}$. Specifically, the quadratic objective function $f(\mathbf{x})$ can be rewritten as
\begin{equation}
\begin{aligned}
f(\mathbf{x})
&= \sum_{j=1}^n \sum_{k=1}^n A_{jk}x_jx_k
+ \sum_{j=1}^n b_jx_j \\
&= \sum_{j=1}^n \sum_{k=1}^n A'_{jk}x_jx_k .
\end{aligned}
\end{equation}
where $A'_{jj}=A_{jj}+b_j$ for $j = 1, 2, \dots, n$, effectively incorporating the linear coefficients into the quadratic matrix.

The equivalence between QUBO and Ising models arises from a simple variable substitution, $x_j = \frac{1}{2}(s_j + 1)$, mapping QUBO variables to Ising spins. Substituting this into the QUBO formulation gives:
\begin{align}
\mathbf{J} &= \frac{1}{4}\mathbf{A}, \\
h_j &= \frac{1}{4}
\left(
2b_j + \sum_{k=1}^n (A_{jk}+A_{kj})
\right),
\end{align}
up to an insignificant additive constant. 

Linear terms in the Ising model can similarly be removed by extending the coupling matrix. Introducing an additional spin variable $s_{n+1}$ and an extended coupling matrix\footnote{In contrast to the QUBO model, where linear terms are absorbed into the matrix diagonal, the Ising model sets diagonal entries to zero, as they add only a constant to any solution. Instead, an additional spin variable is introduced, coupling to other spins with strengths proportional to the external field.} $\overline{\mathbf{J}} \in \mathbb{R}^{(n+1) \times (n+1)}$, the reformulated matrix is:
\begin{equation}
    \bar{J}_{j k}= 
    \begin{cases}J_{j k}, & \text { if } 1 \leq j, k \leq n \\
    \frac{1}{2} h_j & \text { for } k=n+1, \forall j \in [n] \\
    \frac{1}{2} h_k & \text { for } j=n+1, \forall k \in [n] \\
    0 & \text { for } j=k=n+1 .
    \end{cases}
    \label{ising_linearise_eqn}
\end{equation}
This reformulation produces an Ising model equivalent to the original but without explicit linear terms.

\subsection{Quantum Approximate Optimisation Algorithm}

The Quantum Approximate Optimisation Algorithm (QAOA) is a hybrid quantum-classical algorithm for approximately solving combinatorial optimisation problems, including QUBO, by encoding solutions into the ground state of a Hamiltonian, approximated by a PQC. QAOA relies on two primary operators: the problem unitary and the mixing unitary, derived from the problem Hamiltonian $H_P$ and the mixing Hamiltonian $H_M$, defined as:
\begin{align}
    H_P &= \sum_{j=1}^n \sum_{k=1}^n J_{j k} Z_j Z_k+\sum_{j=1}^n h_j Z_j, \label{qising_eqn} \\
    H_M &= \sum\limits^n_{j=1} X_j,
\end{align}
where $Z_j$ and $X_j$ are the Pauli-$Z$ and Pauli-$X$ operators acting on the $j$-th qubit. The problem Hamiltonian $H_P$, a quantum version of \cref{ising_eqn}, encodes the objective function to be minimised, while the mixing Hamiltonian $H_M$ facilitates transitions between states.

The problem unitary is parameterised by a real-valued angle $\gamma \in \mathbb{R}$ and is given by:
\begin{equation}
\begin{aligned}
U(H_P,\gamma)
&= e^{-i\gamma H_P} \\
&= \prod_{j,k}^n e^{-i\gamma J_{jk}Z_jZ_k}
\prod_j^n e^{-i\gamma h_jZ_j}.
\end{aligned}
\end{equation}

while the mixing unitary, parameterised by $\beta \in [0, \pi]$, is expressed as:
\begin{equation}
U (H_M, \beta) = e^{-i\beta H_M} = \prod\limits^n_{j=1}\, e^{-i \beta  X_j}.
\end{equation}
For any positive integer $p$, the QAOA ansatz generates a parameterised quantum state using $2 p$ angles, $(\boldsymbol{\gamma}, \boldsymbol{\beta})=\left(\gamma_1, \ldots, \gamma_p, \beta_1, \ldots, \beta_p\right)$, by alternating the application of $U\left(H_P, \gamma\right)$ and $U\left(H_M, \beta\right)$, resulting in the quantum state:
\begin{equation}
\ket{\boldsymbol{\gamma}, \boldsymbol{\beta}} = \prod_{i=1}^p U (H_M, \beta_i) \, U (H_P, \gamma_i)\, \ket{s},
\label{QAOA_Var_State}
\end{equation}
where $\ket{s}$ denotes the uniform superposition over all $n$-bit strings:
\begin{equation}
\ket{s} = \frac{1}{\sqrt{2^n}} \! \! \! \sum_{z\in \{0,1\}^n} \! \! \! \! \ket{z} .
\end{equation}
The goal is to optimise the parameters $\boldsymbol{\gamma}$ and $\boldsymbol{\beta}$ by minimising the expectation value of $H_P$ with respect to the generated quantum state:
\begin{equation}
\langle H_P (\boldsymbol{\gamma}, \boldsymbol{\beta}) \rangle  = \bra{\boldsymbol{\gamma}, \boldsymbol{\beta}} H_P \ket{\boldsymbol{\gamma}, \boldsymbol{\beta}},
\label{FarhiEq8}
\end{equation}
which is computed through repeated measurements of the quantum system. The performance of QAOA is often quantified by the approximation ratio, which compares the objective value achieved at the optimised parameters, $\left\langle H_P\left(\gamma^*, \beta^*\right)\right\rangle$, with the optimal value of the underlying problem.


To identify the optimal parameters, $\left(\boldsymbol{\gamma}^*, \boldsymbol{\beta}^*\right)$, a classical optimiser iteratively adjusts $\boldsymbol{\gamma}$ and $\boldsymbol{\beta}$ to minimise $\left\langle H_P\right\rangle$:
\begin{equation}
(\boldsymbol{\gamma}^*, \boldsymbol{\beta}^*)=
\argmin_{\boldsymbol{\gamma}, \boldsymbol{\beta}} \ 
\langle H_P (\boldsymbol{\gamma}, \boldsymbol{\beta}) \rangle  .
\end{equation}
However, in general finding these parameters is $\mathsf{NP}$-Hard~\cite{bittel2021training}, and practical implementations rely on heuristic methods such as gradient descent and grid search. Gradient descent can suffer from poor initialisation, local minima, and barren plateaus, while grid search becomes computationally infeasible as the number of parameters grows.


In this work, we focus on QAOA at depth $p=1$ (QAOA$_1$), with classical parameters $(\gamma, \beta)$ applied to an Ising model, where closed-form expressions are available for evaluating \cref{FarhiEq8}. Before exploring these expressions, we introduce a few definitions that will be used throughout the paper.

The Ising model is depicted by an undirected graph $G=(V, E)$, where each vertex $v \in V$ signifies a spin, and each edge $\{u, v\} \in E$ signifies an interaction between spins $u$ and $v$. These interactions are defined by the symmetric coupling matrix $\mathbf{J}$, with each edge $\{u, v\}$ assigned a coupling strength $J_{u v}$. Furthermore, the external magnetic fields are represented by node weights, where each vertex $u \in V$ carries an external field strength $h_u$, collectively denoted by the vector $\mathbf{h}$. For a vertex $w \in V$, let $\mathcal{N}(w)=\{x \in V \mid\{x, w\} \in E\}$ represent the set of neighbours of $w$, i.e., vertices adjacent to $w$. For an edge $\{u, v\} \in E$, we define:
\begin{itemize} 
    \item $e_{uv} = \mathcal{N}(v)\setminus \{u\}$ is the set of vertices other than $u$ that are connected to $v$.
    \item $d_{uv} = \mathcal{N}(u)\setminus \{v\}$ is the set of vertices other than $v$ that are connected to $u$.
    \item $F_{uv} =\mathcal{N}(u) \cap \mathcal{N}(v)$ is the set of vertices that form a triangle with the edge $\{u,v\}$. In other words, $F_{uv}$ is the set of vertices that are neighbours of both $u$ and $v$.
\end{itemize}
For compactness, we omit subscripts and use $e$, $d$, and $F$ directly when their meanings are clear from context. 

The following theorem can be used to compute the expectation value of QAOA$_1$ for an arbitrary Ising model with fields, thereby allowing us to assess its performance.

\begin{theorem}
    Consider the QAOA$_1$ state $|\gamma, \beta\rangle$ for an arbitrary Ising model with external fields. The expectation value of the Hamiltonian $H_P$, as defined in \cref{qising_eqn}, in this state is given by:
    \begin{equation}
        \langle\gamma, \beta|H_P| \gamma, \beta\rangle= \sum\limits_{\{u, v\} \in E}\left\langle C_{u v}\right\rangle + \sum\limits_{i \in V} \left\langle C_{i}\right\rangle, 
    \end{equation}
    where
    \begin{equation}
        \left\langle C_i \right\rangle = h_i \sin 2 \beta \sin \gamma'_i \prod_{k \in \mathcal{N}(i)} \cos \gamma'_{ik}
        \label{qaoa_ozaeta_thm_eqn1}
    \end{equation}
    \begin{widetext}
        \begin{equation}
            \begin{aligned}
                \left\langle C_{uv} \right\rangle &= \frac{J_{uv}}{2} \sin 4\beta \sin \gamma'_{uv} 
                \left( 
                    \cos \gamma'_v \prod_{w \in e} \cos \gamma'_{wv} 
                    + \cos \gamma'_u \prod_{w \in d} \cos \gamma'_{uw} 
                \right) \\[0.25cm]
                & \quad - \frac{J_{u v}}{2} \sin ^2 2 \beta \! \prod_{\substack{w \in e \\ w \notin F}} \! \cos \gamma_{w v}^{\prime} \! \prod_{\substack{w \in d \\ w \notin F}} \! \cos \gamma_{u w}^{\prime}\left(\sum_{\chi= \pm 1} \! \chi \cos \left(\gamma_u^{\prime}+\chi \gamma_v^{\prime}\right) \prod_{f \in F} \! \cos \left(\gamma_{u f}^{\prime}+\chi \gamma_{v f}^{\prime}\right)\!\right)
            \end{aligned}
            \label{qaoa_ozaeta_thm_eqn2}
        \end{equation}
    \end{widetext}
    and $\gamma'_{uv} = 2J_{uv}\gamma$ and $\gamma'_i = 2 h_i \gamma$.
    \label{qaoa_ozaeta_thm}
\end{theorem}
In the absence of local fields (i.e., $h_i=0$ for all $i \in V)$, the Hamiltonian $H_P$ of the Ising model reduces to
\begin{equation}
    H_P = \sum_{j=1}^n \sum_{k=1}^n J_{j k} Z_j Z_k, \label{qising_eqn2} 
\end{equation}
and the expression for the expectation value simplifies, as presented in the following corollary.
\begin{corollary}
Consider the QAOA$_1$ state $|\gamma, \beta\rangle$ for an arbitrary Ising model without external fields. The expectation value of the Hamiltonian $H_P$, as defined in \cref{qising_eqn2}, in this state is given by:
    \begin{equation}
        \langle\gamma, \beta|H_P| \gamma, \beta\rangle = \sum_{\{u, v\} \in E}\left\langle C_{u v}\right\rangle,
    \end{equation}
    where
    \begin{widetext}
    \begin{equation}
        \begin{split}
            \left\langle C_{u v}\right\rangle &= \frac{J_{uv}}{2} \sin 4 \beta \sin\gamma_{uv}'\! \left( \prod_{w \in e} \! \cos\gamma_{wv}' \! + \! \prod_{w \in d} \! \cos \gamma_{uw}' \! \right) \\[0.25cm]
            & \quad  - \frac{J_{uv}}{2} \sin^2 2 \beta \! \prod_{\substack{w \in e \\ w \notin F}} \! \cos\gamma_{wv}' 
            \! \prod_{\substack{w \in d \\ w \notin F}} \! \cos \gamma_{uw}' \Bigg(\! \prod_{f \in F} \! \cos(\gamma_{uf}' \! + \! \gamma_{vf}') \! - \! \prod_{f \in F} \! \cos(\gamma_{uf}' \! - \! \gamma_{vf}') \! \! \Bigg)
        \end{split}
        \label{qaoa_ozaeta_col_eqn}
    \end{equation}
    \end{widetext}
    and $\gamma'_{uv} = 2J_{uv} \gamma$.
    \label{qaoa_ozaeta_col}
\end{corollary}
The proof of \cref{qaoa_ozaeta_thm}, which concerns arbitrary Ising models with external fields, follows Ozaeta et al.~\cite{ozaeta2022expectation}. In the absence of external fields, related expectation-value formulas were previously derived for unweighted Ising/MaxCut-type instances by Wang et al.~\cite{wang2018quantum}, and for weighted field-free graphs by Bravyi et al.~\cite{bravyi2020obstacles} and Vijendran et al.~\cite{vijendran2024expressive}. Corollary 1 is the corresponding field-free specialisation in our notation.

From \cref{qaoa_ozaeta_thm} and \cref{qaoa_ozaeta_col}, it is clear that, for $p=1$, the expectation value $\left\langle C_{u v}\right\rangle$ of any edge in a graph depends only on the nodes and edges adjacent to it. Similarly, the single-spin term $\left\langle C_i\right\rangle$ for any node depends only on the node's external field and its adjacent edges. The total expectation value for QAOA$_1$ is obtained by summing over all nodes and edges. For a system with $n$ nodes, the analytical expressions for each edge and node can be computed in linear time, $\mathcal{O}(n)$. Since the graph has at most $\binom{n}{2}=\mathcal{O}\left(n^2\right)$ edges, the overall time complexity for computing the full expectation value for QAOA$_1$ is $\mathcal{O}\left(n^3\right)$. To obtain an approximate solution in the form of a bit string for an arbitrary problem instance, however, requires executing the QAOA circuit on a quantum computer to generate the quantum state, followed by measurements on that state.

\subsection{Fourier Analysis of VQAs}

Fourier analysis studies how general functions can be represented or approximated by sums of simpler trigonometric functions, allowing complex functions to be analysed in terms of their frequency components and providing insight into their structure. In VQAs, Fourier analysis is a powerful tool for examining the cost landscape~\cite{schuld2021effect, gil2020input, fontana2022spectral, fontana2022efficient, nemkov2023fourier}. By expressing the cost function as a Fourier series, we can analyse the frequency spectrum determined by the problem Hamiltonian and circuit parameters. This approach reveals the accessible frequencies and structural features of the optimisation landscape.

For a PQC with a cost function defined by $\mathcal{C}(\boldsymbol{\theta})=\langle 0|U^{\dagger}(\boldsymbol{\theta}) H_P U(\boldsymbol{\theta})| 0\rangle$, where $H_P$ is the problem Hamiltonian and $\boldsymbol{\theta} \in \mathbb{R}^M$ represents $M$ independent continuous parameters, the Fourier series representation is given by:
\begin{equation}
    \mathcal{C}(\boldsymbol{\theta})=\sum_{\boldsymbol{\omega} \in \boldsymbol{\Omega}} c_{\boldsymbol{\omega}} e^{i \boldsymbol{\omega} \cdot \boldsymbol{\theta}},
\end{equation}
where $c_{\boldsymbol{\omega}}$ are the Fourier coefficients corresponding to the frequencies $\boldsymbol{\omega}$, and $\boldsymbol{\omega} \cdot \boldsymbol{\theta}$ denotes their inner product. The frequency spectrum $\boldsymbol{\Omega} \subset \mathbb{R}^M$ is determined entirely by the eigenvalues of the problem Hamiltonian, while the circuit design influences the coefficients $c_{\boldsymbol{\omega}}$ that the quantum model can realise. 

By representing quantum models as partial Fourier series---where only a subset of Fourier coefficients are non-zero---we can characterise the functional families that a quantum model can learn based on two key properties. The first is the frequency spectrum $\boldsymbol{\Omega}$, which specifies the set of Fourier basis functions $e^{i \boldsymbol{\omega} \cdot \boldsymbol{\theta}}$ that may appear in the model. The second is the expressivity of the coefficients $\left\{c_{\boldsymbol{\omega}}\right\}$, which determines which linear combinations of these basis functions can be realised. Together, $\boldsymbol{\Omega}$ and $\left\{c_{\boldsymbol{\omega}}\right\}$ delineate the functional families that the quantum model can approximate and provide insight into cost-landscape features such as smoothness and the distribution of local minima.


\section{Theoretical Results}

\subsection{QAOA as Partial Fourier Series}

We show that the QAOA expectation value can be expressed as a partial Fourier series, revealing the frequency components accessible to the circuit. This decomposition links circuit depth to optimisation performance and highlights the expressivity limits of shallow QAOA (as discussed in \cref{qaoa_fourier_app}). Identifying the maximum frequency also determines the grid resolution needed to characterise the QAOA$_1$ cost landscape effectively.

Consider the QAOA$_1$ with the quantum state
\begin{equation}
|\psi(\gamma, \beta)\rangle=e^{-i \beta H_M} e^{-i \gamma H_P}|s\rangle
\end{equation}
for an arbitrary Ising model described by the Hamiltonian $H_P$. The expectation value of $H_P$ with respect to the state $|\psi(\gamma, \beta)\rangle$ is given by:
\begin{equation}
f(\gamma, \beta)=\left\langle\psi(\gamma, \beta)\left|H_P\right| \psi(\gamma, \beta)\right\rangle.
\label{qaoa1_exp_eqn}
\end{equation}
Our goal is to express $f(\gamma, \beta)$ as a partial Fourier series of the form:
\begin{equation}
f(\gamma, \beta)=\sum_{n \in \Omega} c_n(\beta) e^{i n \gamma}
\label{fourier_series_eqn}
\end{equation}
with integer-valued frequencies (if $\Omega=\{-K, \ldots, K\}$, then we call \cref{fourier_series_eqn} a truncated Fourier series). To proceed, let us consider the eigenstates $|j\rangle$ of the cost Hamiltonian $H_P$, with corresponding eigenvalues $\lambda_j$ such that:
\begin{equation}
H_P|j\rangle=\lambda_j|j\rangle .
\end{equation}
The unitary evolution under $H_P$, given by $e^{-i \gamma H_P}$, can be expanded in the eigenbasis of $H_P$ as:
\begin{equation}
e^{-i \gamma H_P}=\sum_{j = 1}^n e^{-i \gamma \lambda_j}|j\rangle\langle j|.
\end{equation}
Applying the problem unitary $e^{-i \gamma H_P}$ to $\ket{s}$ yields:
\begin{equation}
e^{-i \gamma H_P}|s\rangle=\sum_{j = 1}^n e^{-i \gamma \lambda_j}\langle j |s\rangle|j\rangle .
\end{equation}
Thus, the expectation value $f(\gamma, \beta)$ becomes:
\begin{equation}
f(\gamma, \beta) = \sum_{j = 1}^n \sum_{k = 1}^n c_{jk}(\beta) e^{i \gamma (\lambda_j - \lambda_k)},
\label{qaoa_fourier_exp}
\end{equation}
where $c_{jk}(\beta)$ are the Fourier coefficients encapsulating the contributions from the overlaps of the eigenstates of $H_P$ after evolution under the mixer Hamiltonian $H_M$. These coefficients are given by:
\begin{equation}
\begin{aligned}
c_{jk}(\beta)
&= \langle s|j\rangle \langle j| e^{i\beta H_M} H_P e^{-i\beta H_M} |k\rangle \langle k|s\rangle \\
&= 2^{-n}\,\langle j| e^{i\beta H_M} H_P e^{-i\beta H_M} |k\rangle .
\end{aligned}
\end{equation}
Since the problem Hamiltonian $H_P$ is diagonal in the Pauli-$Z$ basis, with diagonal entries representing cost values for each bitstring $z \in\{1,-1\}^n$, QAOA aims to identify the eigenvector corresponding to the minimum eigenvalue to solve the optimisation problem. The largest possible frequency of the expectation value $f(\gamma, \beta)$ is determined by the difference between the costs of the best and worst solutions, $\lambda_{\max}-\lambda_{\min}$. To fully capture this frequency range, all coefficients in \cref{qaoa_fourier_exp} must be non-zero for each pair of eigenvalues in $H_P$.

At shallow depths, QAOA is restricted in the frequencies it can access because many Fourier coefficients $c_{j k}(\beta)$ vanish, particularly when $\left|\lambda_j-\lambda_k\right|$ is large. This limitation arises from the dependence of accessible frequencies on the circuit depth $p$, which governs the range of Fourier terms representable in \cref{qaoa_fourier_exp}. As $p$ increases, more frequencies become available, enhancing algorithm performance. Identifying the largest non-zero frequency in QAOA$_1$ provides valuable insights and helps determine the optimal resolution for grid searches over $(\gamma, \beta)$. This avoids under- or overestimating the resolution, which could compromise search accuracy or efficiency. Although finding these non-zero coefficients can be computationally challenging, we derive the maximum frequency for QAOA$_1$ on Ising models analytically using \cref{qaoa_ozaeta_thm} and \cref{qaoa_ozaeta_col}. The following subsection examines the results of this derivation, shedding light on the expressivity of QAOA$_1$ and its optimisation landscape.

\subsection{Maximum Frequency of Level-1 QAOA} \label{max_freq_subsub}

To understand the maximum frequency for QAOA$_1$ on Ising models, it is useful to reinterpret the Fourier series in \cref{qaoa_fourier_exp}. This representation frames the expectation value of QAOA$_1$, as described in \cref{qaoa_ozaeta_thm} and \cref{qaoa_ozaeta_col}, as a bandlimited two-dimensional signal. For clarity in the subsequent analysis, we begin by formally defining the concept of a bandlimited function.
\begin{definition} \label{bandlimited_f_def}
Let $T > 0$ be the period and $K \in \mathbb{N}$ a positive integer. A periodic function $f(t)$ with period $T$ is called \emph{bandlimited} to bandwidth $K \omega_0$, where $\omega_0 = 2\pi/T$, if it can be expressed as
\begin{equation} \label{bw_per_fun_eqn}
    f(t) = \sum_{k=-K}^{K} c_k e^{i k \omega_0 t},
\end{equation}
where $c_k = 0$ for all integers $k$ with $|k| > K$.
\end{definition}

For any fixed value of $\beta$, we treat the $\beta$-dependent terms in \cref{qaoa_ozaeta_thm,qaoa_ozaeta_col} as prefactors and focus on the $\gamma$-dependence. This is because the instance-dependent oscillatory structure enters through $\gamma$, which is weighted by the couplings and external fields in the closed-form expressions, and equivalently appears through the phase factors $e^{i \gamma\left(\lambda_j-\lambda_k\right)}$ in the Fourier representation of \cref{qaoa_fourier_exp}. The signal along the $\gamma$ dimension---representing the cost landscape---is periodic only if the weights of the problem instance are commensurable\footnote{Two values are said to be commensurable if their ratio is a rational number, meaning they share a common measure or can be expressed as integer multiples of a common unit.}. 


For a bandlimited function $f(t)$ as defined in \cref{bandlimited_f_def}, we denote its maximum angular frequency by $\omega_{\max }[f]$, where $\omega = 2\pi \nu$ relates angular frequency to the ordinary frequency $\nu$. The following theorem leverages the closed-form expressions in \cref{qaoa_ozaeta_thm} to compute the maximum angular frequency of the expectation value in QAOA$_1$, specifically for the $\gamma$-dependent terms in Ising models with fields.
\begin{theorem}\label{qaoa_2_field_max_freq_thm}
    Consider the QAOA$_1$ state $|\gamma, \beta\rangle$ for an arbitrary Ising model with fields. For any fixed $\beta$, the maximum angular frequency of the $\gamma$-dependent terms in the expectation value of the Hamiltonian $H_P$, as defined in \cref{qising_eqn}, is given by:
    \begin{multline}
  \omega_{\max}\!\left[\langle H_P \rangle_\gamma\right] = \max\Big\{ 
    \omega_{\max}\!\left[\langle C_i \rangle_\gamma\right], \\
    \omega_{\max}\!\left[\langle C_{uv} \rangle_\gamma\right] 
    \mid i \in V, \{u,v\} \in E \Big\}
\end{multline}
    where
    \begin{equation}
        \omega_{\max} \left[ \langle C_i \rangle_\gamma \right] = 2 \times \left( |h_i| + \sum_{k \in \mathcal{N}(i)} |J_{ik}| \right)
    \end{equation}
    and
    \begin{widetext}
    \begin{equation}
        \begin{split}
            \omega_{\max} \left[ \langle C_{uv} \rangle_\gamma \right] &= 2 \times  \max \left\{ |J_{uv}| + \max \left\{ |h_v| + \sum_{w \in e} |J_{wv}|, |h_u| + \sum_{w \in d} |J_{uw}| \right\}, \right. \\[0.25cm]
        & \qquad\left. \sum_{\substack{w \in e \\ w \notin F}} |J_{wv}| + \sum_{\substack{w \in d \\ w \notin F}} |J_{uw}| + \max \left\{ |h_u \pm h_v| + \sum_{f \in F} |J_{uf} \pm J_{vf}| \right\} \right\}. 
        \end{split}
    \end{equation}
    \end{widetext}
\end{theorem}

The proof of this theorem is provided in \cref{qaoa_2_field_max_freq_proof}. Notably, this theorem simplifies in the absence of external fields. The following corollary presents the maximum angular frequency of the expectation value for QAOA$_1$ in Ising models without fields, which can also be derived directly from \cref{qaoa_ozaeta_col}.

\begin{corollary}\label{qaoa_2_max_freq_col}
     Consider the QAOA$_1$ state $|\gamma, \beta\rangle$ for an arbitrary Ising model without fields. For any fixed $\beta$, the maximum angular frequency of the $\gamma$-dependent terms in the expectation value of the Hamiltonian $H_P$, as defined in \cref{qising_eqn2}, is given by:
    \begin{equation}
    \omega_{\max} \left[ \langle H_P \rangle_\gamma \right] \! = \! \max \left\{ \omega_{\max} \left[ \langle C_{uv} \rangle_\gamma \right] \Big| \{u,v\} \! \in \! E \right\} \!,
    \end{equation}
    where
    \begin{align}
    \omega_{\max}\!\left[\langle C_{uv} \rangle_\gamma\right]
    &= 2 \! \times \! \max\left\{ |J_{uv}| \! + \! \max\left\{
        \sum_{w \in e}|J_{wv}|,\right.\right. \notag\\[0.25cm]
    & \left.\left. \sum_{w \in d}|J_{uw}|
      \right\}, \sum_{\substack{w \in e \\ w \notin F}}|J_{wv}|
      + \sum_{\substack{w \in d \\ w \notin F}}|J_{uw}| \right.\notag\\[0.25cm]
    & \left. + \max\left\{
        \sum_{f \in F}|J_{uf} \pm J_{vf}|
      \right\} \right\}.
  \label{qaoa_2_max_freq_col_eq2}
\end{align}
\end{corollary}
\begin{proof}
This result follows directly from \cref{qaoa_2_field_max_freq_thm} by setting $h_i = 0$ for all $i \in V$.
\end{proof}

These results characterise the maximum frequency of QAOA$_1$ for arbitrary Ising models, including cases with external fields, allowing us to efficiently analyse the cost landscape and identify optimal parameters.

\subsection{Characterising the Energy Landscape of Level-1 QAOA} \label{char_ener_subsub}

Building on the frequency characterisation of QAOA$_1$, we employ the Nyquist-Shannon Sampling Theorem~\cite{shannon1949communication} to determine the maximum permissible sampling period required to efficiently reconstruct its cost landscape. Unlike prior studies that have utilised signal processing techniques to analyse the optimisation landscapes of PQCs~\cite{fontana2022efficient, hao2023enabling}, our method is specifically tailored to QAOA$_1$ applied to Ising models. It relies exclusively on analytical expressions, eliminating the need to execute the quantum algorithm on actual hardware. Moreover, our approach captures the entire QAOA$_1$ landscape, including challenging features such as barren plateaus and narrow gorges, which earlier methods do not address. 

Before proceeding, we formally state the Nyquist-Shannon Sampling Theorem~\cite{shannon1949communication}:
\begin{theorem}[Nyquist-Shannon Sampling Theorem] \label{nyquist_shannon_thm}
Let $x(t)$ be a bandlimited function with maximum frequency $B$. If $x(t)$ is sampled at uniform intervals $t_n = nT_s$ for all integers $n \in \mathbb{Z}$, and the sampling period $T_s$ satisfies
\begin{equation}
    T_s \leq \frac{1}{2B + 1},
\end{equation}
then $x(t)$ is uniquely determined by its samples $\{ x(t_n) \}$, and can be exactly reconstructed from these samples.
\end{theorem}

The Nyquist-Shannon Sampling Theorem ensures that a bandlimited function can be perfectly reconstructed provided the sampling rate is at least twice its highest frequency. Furthermore, for a periodic bandlimited function with period $T$, it suffices to sample the function within the interval $[0, T]$ at this minimum rate. In the context of QAOA$_1$, the theorem establishes the maximum allowable distance between consecutive samples (i.e., the maximum permissible sampling period) necessary to accurately reconstruct the optimisation landscape along $\gamma$. This result is formalised in the following theorem.
\begin{theorem}\label{qaoa_min_samples_thm}
Let $|\gamma, \beta \rangle$ denote the QAOA$_1$ state applied to an arbitrary Ising model. Suppose the expectation value $\langle C(\gamma) \rangle = \langle \gamma, \beta^* | H_P | \gamma, \beta^* \rangle$ is a bandlimited function of $\gamma$ with maximum frequency $\nu_{\max}\left[ \langle C(\gamma) \rangle \right]$. Then if $\langle C(\gamma) \rangle$ is sampled at uniformly spaced points with spacing $\Delta_\gamma$ satisfying
\begin{equation}
    \Delta_\gamma \leq \frac{1}{2 \nu_{\max}\left[ \langle C(\gamma) \rangle \right] + 1},
\end{equation}
then $\langle C(\gamma) \rangle$ can be uniquely determined from these samples.
\end{theorem}
\begin{proof}
    The proof follows from \cref{nyquist_shannon_thm} by setting $x(t) = \langle C(\gamma) \rangle$ and $B = \mu_{\max}/2\pi$, where $\mu_{\max}$ is determined from \cref{qaoa_2_field_max_freq_thm} and \cref{qaoa_2_max_freq_col}.
\end{proof}
If the problem Hamiltonian $H_P$ consists of incommensurable weights, the expectation value of QAOA $_1$ along $\gamma$ becomes non-periodic, necessitating sampling at uniform intervals $n \Delta_\gamma$ for all integers $n \in \mathbb{Z}$. However, if the optimal $\gamma^*$ is known \emph{a priori} to lie within a specific region, it suffices to sample only within that interval (which we will discuss in \cref{first_local_subsubsec}). In contrast, when $H_P$ has commensurable weights, the expectation value of QAOA$_1$ along $\gamma$ exhibits periodicity with period $T_\gamma$, making it sufficient to sample within the interval $\left[0, T_\gamma\right]$. Specifically, if $H_P$ contains only integer eigenvalues, the maximum possible period $T_\gamma$ is $\pi$. Therefore, for Hamiltonians with fractional weights, it is advantageous to convert them to integer weights by multiplying by an appropriate constant factor, allowing sampling within $[0, \pi]$ and enabling exact reconstruction with evaluations at $\left\lceil\pi / \Delta_\gamma\right\rceil$ equally spaced points.

Beyond providing sampling requirements, \cref{qaoa_min_samples_thm} highlights that the shape of the QAOA$_1$ optimisation landscape is closely tied to graph-dependent parameters governing the maximum frequency. In particular, higher frequencies in $\langle C (\gamma) \rangle$ can lead to more rapid changes or \emph{sharper} features across $\gamma$, whereas lower-frequency regimes yield \emph{smoother} variations. Empirical observations in, for example, Figure 9 of Vijendran et al.~\cite{vijendran2024expressive} and related works~\cite{stkechly2023connecting,mussig2024connecting} illustrate how increasing graph degree or density compresses regions of large gradients while expanding flatter regions of the landscape. In the remainder of this section, we apply \cref{qaoa_min_samples_thm} to two graph families to elucidate how these graph properties influence the QAOA$_1$ optimisation landscape.

We begin with the Sherrington-Kirkpatrick (SK) model, a paradigmatic example from condensed matter physics characterised by its fully connected nature. The following corollary provides an upper bound on the sampling period required for QAOA$_1$ applied to this model, ensuring exact reconstruction of the expectation value.
\begin{corollary}\label{sk_gregion_cor}
Consider the QAOA$_1$ state $|\gamma, \beta^* \rangle$ applied to the Sherrington-Kirkpatrick model with $n$ spins, where the coupling and field strengths are $\pm 1$. If the samples are spaced by $\Delta_{\gamma}$ satisfying
\begin{equation}
    \Delta_\gamma \leq \frac{\pi}{\mathcal{O}(n) + \pi},
\end{equation}
then the expectation value $\langle \gamma, \beta^* | H_P | \gamma, \beta^* \rangle$ can be exactly reconstructed for all $\gamma \in [0, \pi]$.
\end{corollary}
The proof of this corollary is provided in \cref{sk_narrow_gorge_proof}. It demonstrates that the sampling interval $\Delta_\gamma$ must decrease inversely with the system size $n$ to maintain exact reconstruction. This behaviour is intuitive, as the degree of the SK model scales with the number of spins, $D = n - 1$, and the complexity of interactions increases with $n$. Consequently, the optimisation landscape exhibits increasingly rapid variations as $n$ grows, necessitating finer sampling to accurately capture its features.

Next, we consider triangle-free $D$-regular graphs, which are graphs where each vertex has exactly $D$ neighbours and no three vertices form a triangle. The following corollary establishes a precise upper bound on the sampling period for QAOA$_1$ applied to such graphs with $\pm 1$ edge weights.
\begin{corollary}\label{d_reg_tri_free_conc}
    Let $|\gamma, \beta^* \rangle$ denote the QAOA$_1$ state applied to a triangle-free $D$-regular graph with edge weights $\pm 1$. If the samples are spaced by $\Delta_{\gamma}$ satisfying
    \begin{equation}
        \Delta_\gamma \leq \frac{\pi}{2D + \pi},
    \end{equation}
    then the expectation value $\langle \gamma, \beta^* | H_P | \gamma, \beta^* \rangle$ can be exactly reconstructed for all $\gamma \in [0, \pi]$. For triangle-free graphs with bounded degrees, this bound holds with $D$ replaced by the maximum degree $D_{\max}$.
\end{corollary}
The proof of this corollary is presented in \cref{unw_trif_dreg_proof}. Similar to the SK model, the required sampling interval $\Delta_\gamma$ decreases inversely with the graph's degree $D$. In triangle-free regular graphs, the absence of triangles simplifies interactions, yet higher degrees still induce sharper features in the optimisation landscape. For triangle-free graphs with bounded degrees, $\Delta_\gamma$ inversely depends on the maximum degree $D_{\max}$. Furthermore, even in generally sparse graphs with localised dense regions, the maximum frequency—and consequently the required sampling interval—remains high, making the landscape highly oscillatory. Therefore, maintaining exact reconstruction necessitates finer sampling as the degree increases, ensuring accurate capture of the optimisation landscape and preventing the loss of critical information due to under-sampling.

By applying the Nyquist-Shannon Sampling Theorem, we derived the maximum permissible spacing between consecutive samples required to accurately reconstruct the optimisation landscape along $\gamma$. These results indicate that the necessary sampling density increases inversely with the size, weights, and density of the problem. While the $\beta$ frequencies are independent of the problem Hamiltonian $H_P$, the $\gamma$ frequencies depend directly on it. Consequently, performing a full grid search over $(\beta, \gamma)$ becomes computationally expensive as the problem’s size and complexity grow, particularly for weighted instances. To address this challenge, the following section introduces a method to streamline the process by reducing the grid search to a line search over $\gamma$ while solving analytically for $\beta^*$.

\subsection{Univariate Representation of the Level-1 QAOA's Objective Function}\label{linearizing_qaoa_sec}

\begin{figure*}[!t]
    \centering
    \subfloat[Optimal Path for an Unweighted Graph]{
        \includegraphics[height=0.2425\textheight]{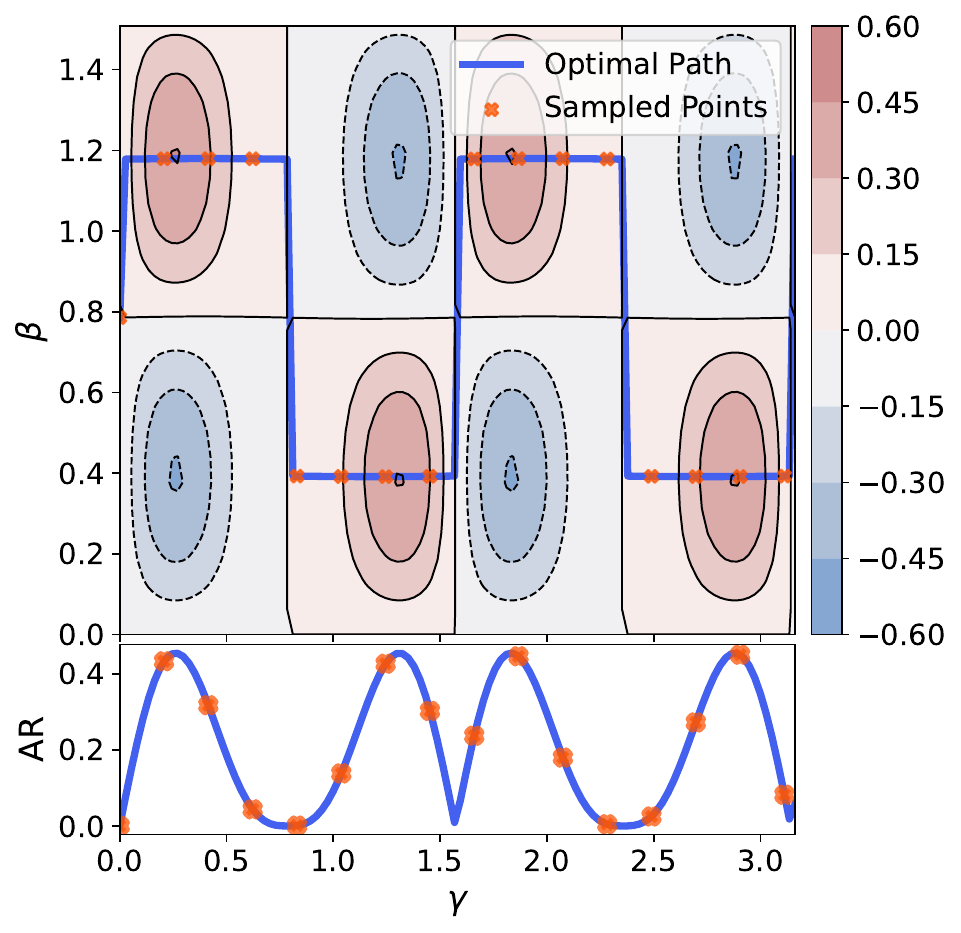}
        \label{optimal_path_2local_diagram_a}
    }
    \subfloat[Optimal Path for a Weighted Graph]{
        \includegraphics[height=0.2425\textheight]{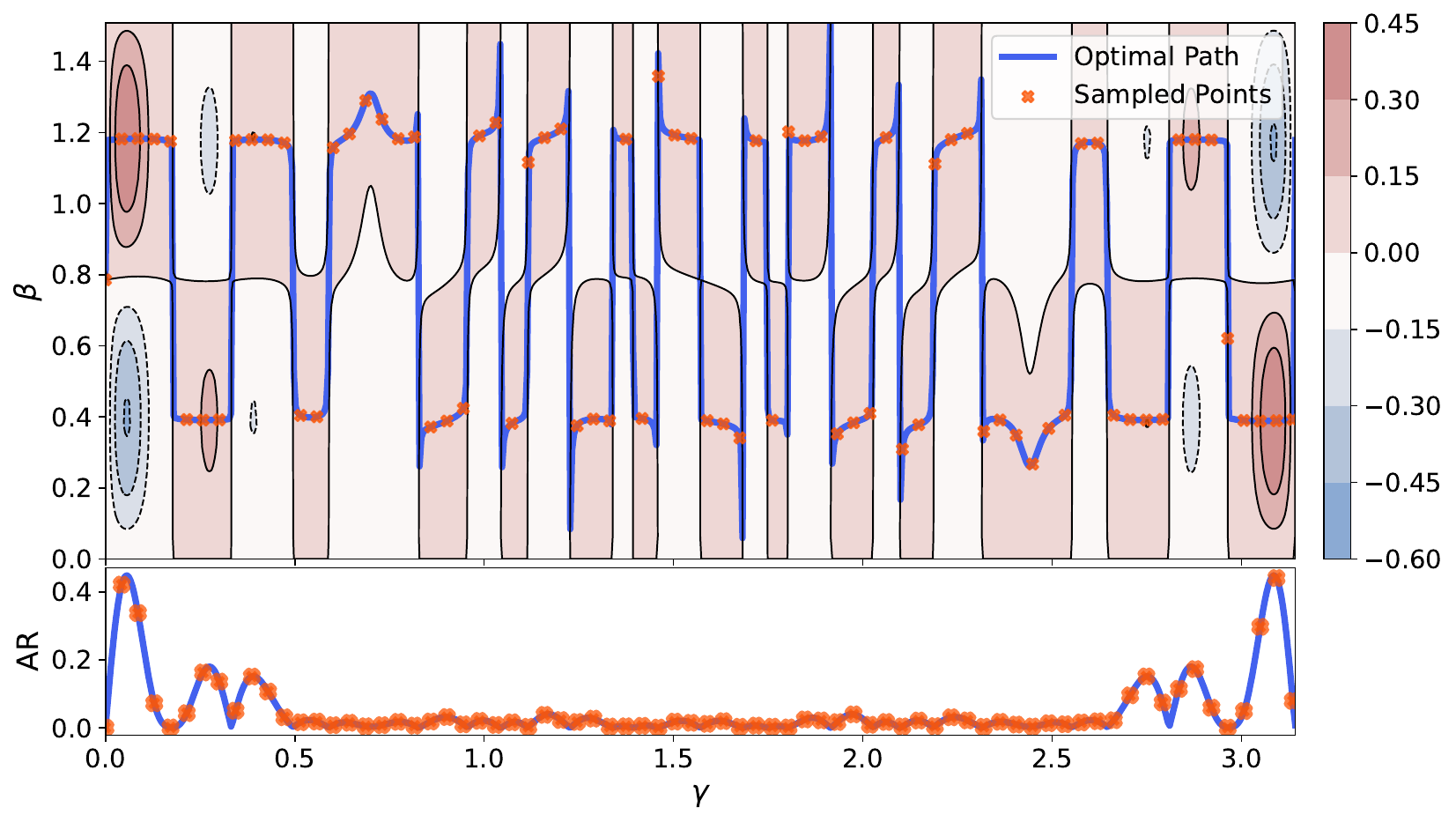}
        \label{optimal_path_2local_diagram_b}
    }
    \caption{\textbf{Optimal $(\gamma, \beta)$ Paths for QAOA$_1$ on a 4-Regular Graph.} \justifying These plots depict the optimal $(\gamma, \beta)$ paths for QAOA$_1$ on a 4-regular graph with 128 vertices, calculated using \cref{qaoa_2_local_opt_params_thm}. The top contour plots show the cost landscape, with colours indicating approximation ratios (higher is better). The line plots below display the approximation ratios computed along $\gamma$ using \cref{QAOA_Gamma_Linear_Eqn}, where the approximation ratio is defined as the objective value obtained by the algorithm divided by the optimal objective value for the same instance.  
    Solid blue lines trace the optimal trajectory, $(\gamma, \beta^*_{\gamma})$, where $\beta^*_{\gamma}$ is the optimal $\beta$ for each fixed $\gamma$. Orange `x' markers indicate sampled points, spaced according to \cref{qaoa_min_samples_thm}. The left plot (unweighted graph) reveals symmetric extrema in the cost landscape, while the right plot (weighted graph with positive integer edge weights drawn from a Gaussian distribution with mean 5 and variance 1) 
    shows extrema clustered near $\gamma \approx 0$ and $\pi$. Discontinuities in the optimal path, evident in both cases, result from abrupt shifts in $\beta$ for specific $\gamma$ values---an effect more pronounced in weighted graphs.}
    \label{optimal_path_2local_diagrams}
\end{figure*}

We now present one of our main results: the closed-form expressions from \cref{qaoa_ozaeta_thm} and \cref{qaoa_ozaeta_col} allow the expectation value of QAOA$_1$ on Ising models to be expressed solely in terms of the $\gamma$ parameter, with the optimal $\beta^*$ determined analytically for any $\gamma \in \mathbb{R}$. We begin by considering models without external fields, where the following theorem demonstrates that this reduction holds, enabling an analytical solution for $\beta^*$ based on the optimal $\gamma^*$.

\begin{theorem}[Ising Model without Fields]\label{qaoa_2_local_opt_params_thm}
    Consider the QAOA$_1$ state $|\gamma, \beta\rangle$ for an arbitrary Ising model without fields. The expectation value of the Hamiltonian $H_P$, as defined in \cref{qising_eqn2}, is a function of $\gamma$ given by:
    \begin{equation}
       \langle\gamma|H_P| \gamma\rangle = -  \sqrt{A^2(\gamma)+ \frac{B^2(\gamma)}{4}} - \frac{B(\gamma)}{2}.
       \label{QAOA_Gamma_Linear_Eqn}
    \end{equation}
    The optimal angles $(\gamma^*, \beta^*)$ that minimise $\langle\gamma, \beta|H_P| \gamma, \beta\rangle$ are given by:
    \begin{align}
        \gamma^* &= \argmin_{\gamma \in \mathbb{R}} \ \langle\gamma|H_P| \gamma\rangle, \\
        \beta^* &= \frac{1}{4}\left[ \arctan \Big( 2A(\gamma^*),B(\gamma^*) \Big) + \pi \right]. \label{qaoa_2_local_opt_beta}
    \end{align}
    The functions $A(\gamma)$ and $B(\gamma)$, defined over the graph $G = (V, E)$ corresponding to the Hamiltonian $H_P$, are given by:
    \begin{align*}
    A(\gamma) &=  \hspace{-0.4cm}\sum\limits_{\{u, v\} \in E} \hspace{-0.25cm} \frac{J_{uv}}{2} \! \sin \! \gamma_{uv}' \hspace{-0.15cm} \left(\prod_{w \in e} \! \cos\gamma_{wv}' \! + \! \prod_{w \in d} \! \cos \gamma_{uw}'\! \right)\!, \\[0.35cm]
    B(\gamma) &= \hspace{-0.4cm}\sum\limits_{\{u, v\} \in E} \hspace{-0.25cm} \frac{J_{uv}}{2} \! \prod_{\substack{w \in e \\ w \notin F}} \! \cos\gamma_{wv}' \! \prod_{\substack{w \in d \\ w \notin F}} \!\cos \gamma_{uw}' \\[0.25cm]
    &  \times \! \Bigg(\! \prod_{f \in F} \! \cos(\gamma_{uf}' \! + \! \gamma_{vf}')\! -\! \prod_{f \in F}\! \cos(\gamma_{uf}' \! - \! \gamma_{vf}') \! \Bigg),
\end{align*}
with $\gamma_{uv}^{\prime} = 2J_{uv}\gamma$.
\end{theorem}

The proof of this theorem is provided in \cref{qaoa_2_local_fields_opt_params_proof}. This approach reduces the grid search over $(\gamma, \beta)$ into an efficient line search on $\gamma$ using \cref{QAOA_Gamma_Linear_Eqn}, with sampled points spaced $\Delta_\gamma$ apart, and $\beta^*$ determined using \cref{qaoa_2_local_opt_beta}. The effectiveness of \cref{qaoa_2_local_opt_params_thm} is demonstrated in \cref{optimal_path_2local_diagrams}, which presents contour plots and optimal paths on the cost landscape for both unweighted and lightly weighted 4-regular graphs with 128 vertices.

In \cref{optimal_path_2local_diagram_a}, the unweighted graph exhibits clear symmetry, with large intervals along $\gamma$ containing extrema. In contrast, the lightly weighted graph in \cref{optimal_path_2local_diagram_b} presents a more intricate landscape. Although some symmetry remains due to the graph's regularity and small weight variance, the extrema clusters near $\gamma \approx 0$ and $\pi$, leaving the intermediate regions filled with spurious local minima and barren plateaus. A naive grid search in such scenarios would be computationally expensive, requiring extensive sampling to navigate these local traps. While exploiting symmetry can reduce computational effort, this strategy is only viable for highly structured problems, which are rare in practical applications. The method in \cref{qaoa_2_local_opt_params_thm} overcomes these challenges by leveraging the line search on $\gamma$ to compute $\beta^*$ directly. This efficiency is evident in \cref{optimal_path_2local_diagrams}, where the paths traced by \cref{QAOA_Gamma_Linear_Eqn} show $\beta$ alternating between lower and upper regions to maximise the approximation ratio, with the line plot capturing the optimal trajectory along $\gamma$.

\begin{figure*}[htpb]
    \centering
    \subfloat[Grid Resolution of $20 \times 20$]{
        \includegraphics[height=0.335\textheight]{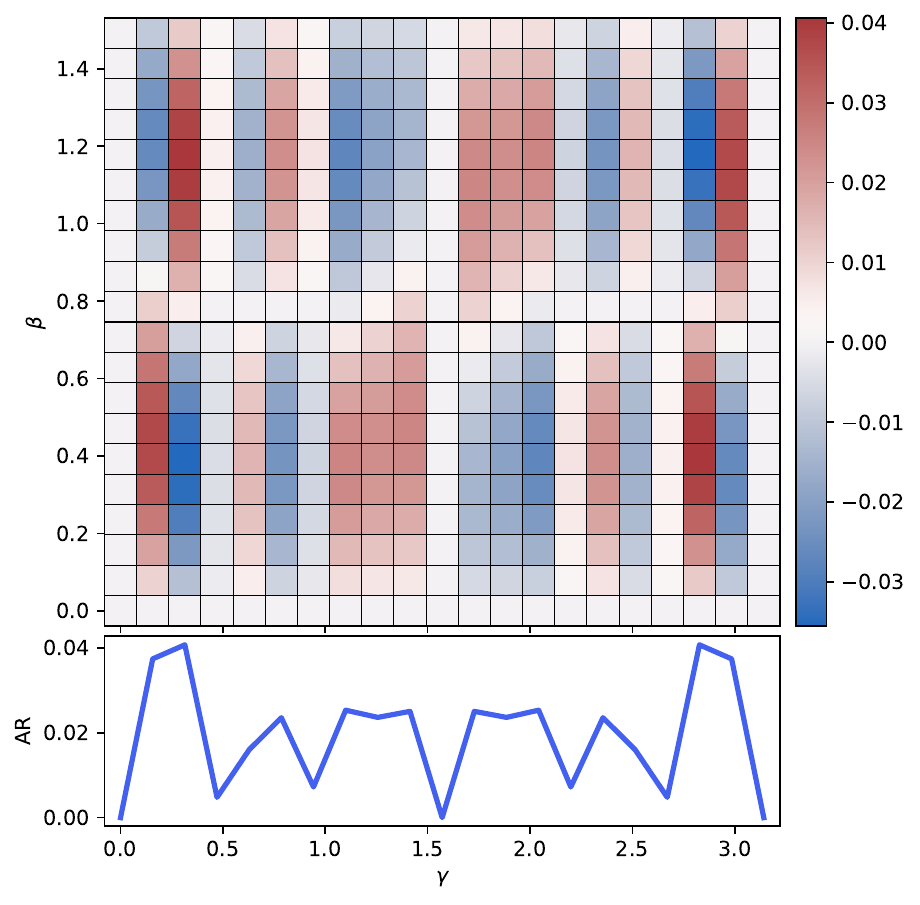}
        \label{Conc_2Local_Plot_a}        
    }
    \subfloat[Grid Resolution of $40 \times 20$]{
        \includegraphics[height=0.335\textheight]{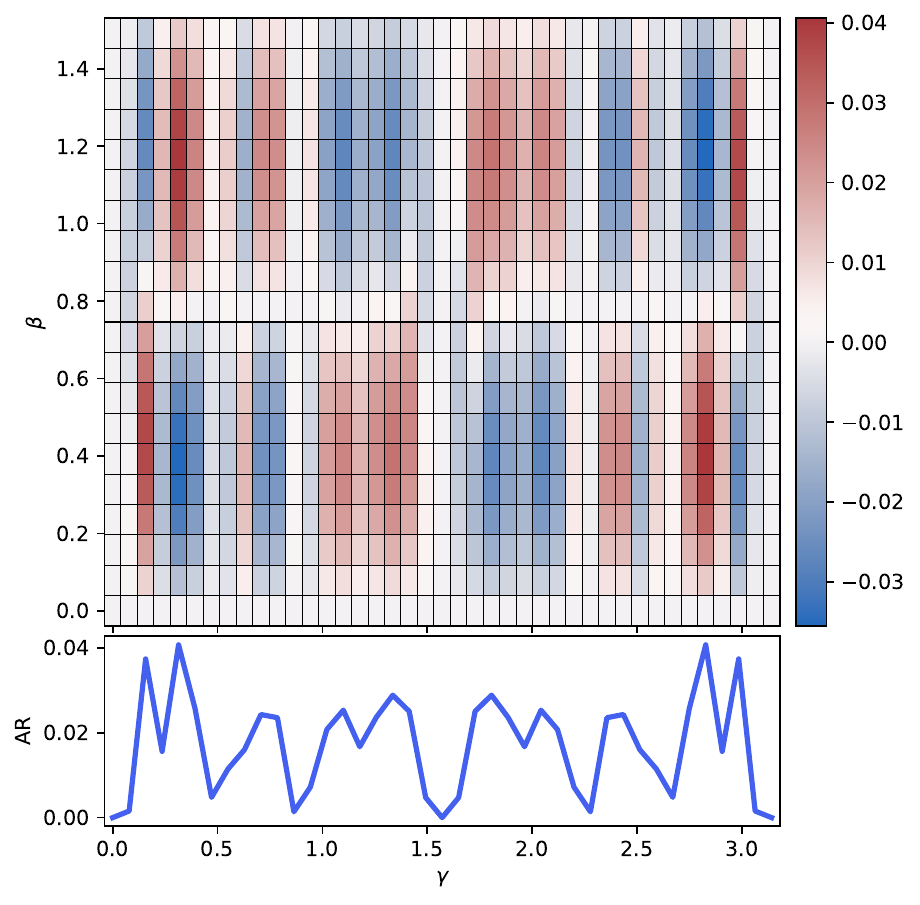}
        \label{Conc_2Local_Plot_b}
    }\\
     \subfloat[Grid Resolution of $60 \times 20$]{
        \includegraphics[height=0.335\textheight]{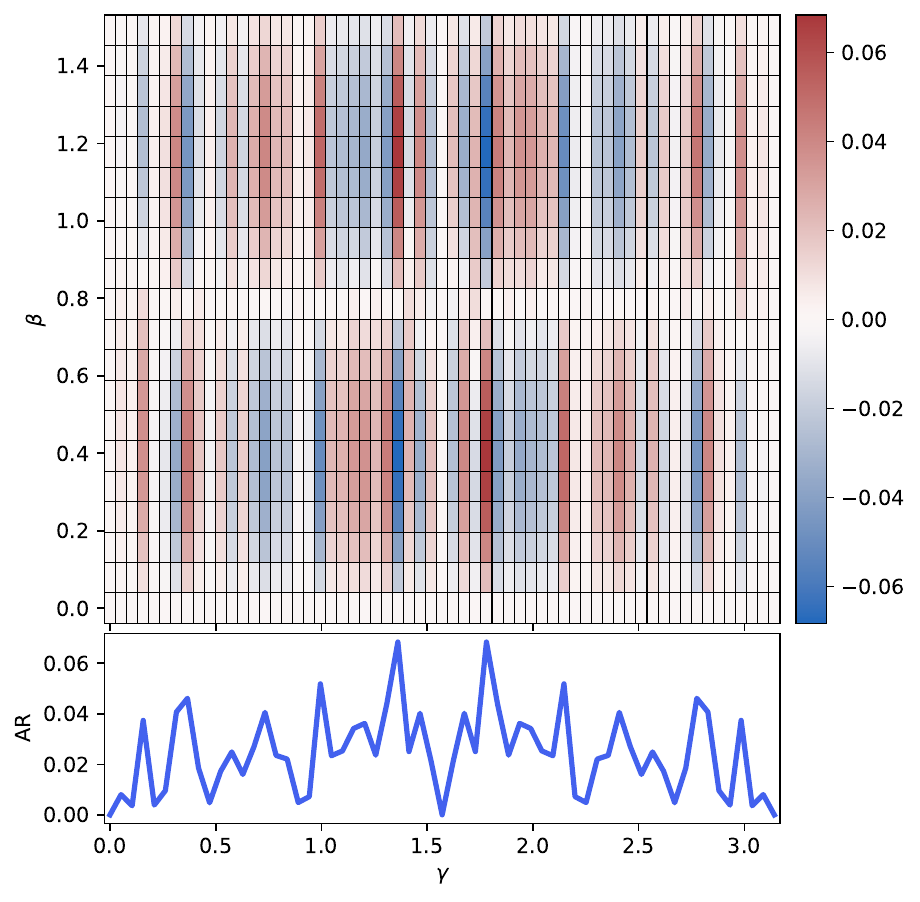}
        \label{Conc_2Local_Plot_c}
    }
    \subfloat[Grid Resolution of $560 \times 20$]{
        \includegraphics[height=0.335\textheight]{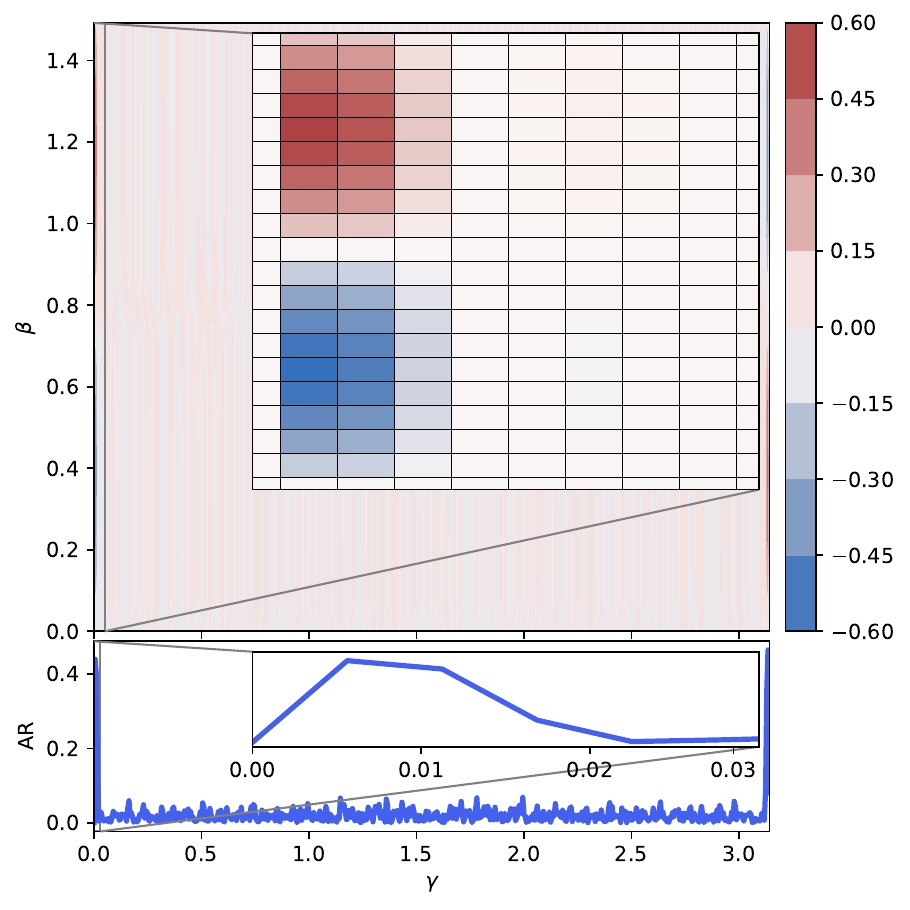}
        \label{Conc_2Local_Plot_d}
    }
    \caption{\textbf{QAOA$_1$ Cost Landscape for a Weighted Regular Graph with Varying $\gamma$ Resolutions.} \justifying These plots illustrate the cost landscape of QAOA$_1$ for a weighted $4$-regular graph with $64$ vertices, where edge weights are sampled from a Gaussian distribution with a mean and variance of $25$ and rounded to the nearest integer. The $\gamma$ dimension is sampled at resolutions of $20$, $40$, $60$, and $560$, while the $\beta$ dimension is fixed at 20 samples. The final plot's resolution of $560$ was determined using \cref{qaoa_min_samples_thm}. Below each mesh grid, a line plot displays the cost landscape along $\gamma$, calculated using the expression in \cref{QAOA_Gamma_Linear_Eqn}. In plots (a) and (b), the maximum approximation ratio is $0.04$, with the first optimum around $\beta < \pi/4$ and $\gamma \approx 0.15$. Plot (c) shows a higher ratio of $0.06$ with $\gamma$ near $1.4$. In plot (d), sampling $\gamma$ at the appropriate resolution yields a ratio of $0.48$, with the optimum at $\beta > \pi/4$ and $\gamma \approx 0.01$. These results highlight the critical importance of appropriately sampling the $\gamma$ dimension to avoid aliasing and distortion, which can lead to misleading local optima.}
    \label{Conc_2Local_Plots}
\end{figure*}

While \cref{qaoa_2_local_opt_params_thm} reduces optimisation to a line search, determining an appropriate sampling resolution for the $\gamma$ dimension is crucial for balancing efficiency and accuracy when optimising the univariate cost function. Excessively fine resolutions increase computational cost, while overly coarse resolutions risk failure due to local minima and barren plateaus. To illustrate, we use a weighted 4-regular graph with 64 vertices. \Cref{Conc_2Local_Plots} compares the two-dimensional cost landscape to the one-dimensional univariate landscape at different $\gamma$ resolutions. Unlike traditional contour plots, these visualisations (excluding \cref{Conc_2Local_Plot_d} without its inset) are generated using a mesh grid to avoid interpolation, which might further distort the already inaccurate landscape. As shown in \cref{Conc_2Local_Plot_a}, \cref{Conc_2Local_Plot_b}, and \cref{Conc_2Local_Plot_c}, both grid and line searches achieve only a marginal approximation ratio of 0.04 to 0.06, barely better than random guessing. In contrast, \cref{Conc_2Local_Plot_d}, with $\gamma$ sampled at the rate prescribed by \cref{qaoa_min_samples_thm}, identifies the global extremum, achieving the maximum possible approximation ratio of approximately 0.48.

Building on \cref{qaoa_2_local_opt_params_thm}, we extend the univariate representation of the QAOA$_1$'s objective function to Ising models with external fields. The presence of external fields introduces additional complexity to the optimisation landscape, requiring a more general treatment. Nevertheless, the following theorem demonstrates that the QAOA$_1$'s objective function can still be simplified to a univariate form, enabling the reduction of the two-dimensional grid search over $(\gamma, \beta)$ to an efficient line search over $\gamma$.

\begin{theorem}[Ising Model with Fields]\label{qaoa_2_local_fields_opt_params_thm}
    Consider the QAOA$_1$ state $|\gamma, \beta\rangle$ for an arbitrary Ising model with fields, defined by the Hamiltonian $H_P$ as in \eqref{qising_eqn}. The optimal angles $\gamma^*$ and $\beta^*$ that minimise the expectation value $\langle\gamma, \beta| H_P|\gamma, \beta\rangle$ are given by:
    \begin{equation}\label{qaoa_field_opt_par_linear_eqn}
        \begin{aligned}
        \left(\gamma^*, \beta^*\right) &= \argmin_{(\gamma,\beta)\in \mathbb{R} \times \mathcal{B}_\gamma} \Big[ \sin (2 \beta) A(\gamma)\\ &  +\sin (4 \beta) B(\gamma) 
        +\sin ^2(2 \beta) C(\gamma) \Big],
        \end{aligned}
    \end{equation}
    where the set $\mathcal{B}_\gamma$ is defined as:
    \begin{equation}
    \mathcal{B}_\gamma=\left\{\left.\beta= \pm \frac{1}{2} \arccos \left(x_i\right) \right\rvert\, i=1,2,3,4\right\}.
    \end{equation}
    and $x_i$ are the real roots (within $[-1,1]$) of the quartic polynomial:
    \begin{equation}\label{qaoa_quartic_eqn}
    \begin{aligned}
        p(x) &=  \left[16 B^2(\gamma) + 4 C^2(\gamma)\right] x^4 \\
        & \quad + 8 A(\gamma) B(\gamma) x^3 \\ 
        & \quad + \left[A^2(\gamma) - 16 B^2(\gamma) - 4 C^2(\gamma)\right] x^2 \\
        & \quad - 4 A(\gamma) B(\gamma) x + 4B^2(\gamma).
    \end{aligned}
    \end{equation}
    The functions $A(\gamma)$, $B(\gamma)$, and $C(\gamma)$, defined over the graph $G = (V, E)$ corresponding to the Hamiltonian $H_P$, are as follows:
    \begin{align*}
            A(\gamma) &= \sum_{i \in V} h_i \sin \gamma'_i \prod_{k \in \mathcal{N}(i)} \cos \gamma'_{ik},\\[0.25cm]
            B(\gamma)\! &= \hspace{-0.275cm} \sum\limits_{\{u, v\} \in E} \! \frac{J_{uv}}{2} \! \sin \! \gamma'_{uv} \! 
            \left( \! 
                \cos \! \gamma'_v \! \prod_{w \in e} \! \cos \! \gamma'_{wv} \! \right. \\
                & \left. \hspace{2.85cm} + \! \cos \! \gamma'_u \! \prod_{w \in d} \! \cos \! \gamma'_{uw} \! 
            \right) \!,\\[0.25cm]
            C(\gamma) &= - \hspace{-0.4cm} \sum\limits_{\{u,v\}\in E}\hspace{-0.15cm}\frac{J_{uv}}{2}\!\prod_{\substack{w\in e\\w\notin F}}\!\cos\gamma_{wv}^{\prime}\!\prod_{\substack{w \in d \\ w\notin F}}\!\cos\gamma_{uw}^{\prime}\\[-0.2cm]
            & \hspace{-0.1cm} \times \hspace{-0.1cm}  \left(\sum_{\chi=\pm1}\!\chi\cos\!\left(\gamma_u^{\prime}\!+\!\chi\gamma_v^{\prime}\right)\hspace{-0.15cm}\prod_{f\in F}\!\cos\!\left(\gamma_{uf}^{\prime}\!+\!\chi\gamma_{vf}^{\prime}\right)\!\right)\hspace{-0.15cm},
    \end{align*}
where $\gamma_{uv}^{\prime} = 2J_{uv}\gamma$ and $\gamma_i^{\prime} = 2h_i\gamma$.
\end{theorem}

The proof of this theorem is provided in \cref{qaoa_2_local_fields_opt_params_proof}. Although the line search method in \cref{qaoa_2_local_fields_opt_params_thm} is slightly more computationally intensive than the approach in \cref{qaoa_2_local_opt_params_thm}, as it requires solving for $\beta^*$ at each $\gamma \in \mathbb{R}$, it remains significantly more efficient than a full grid search. Specifically, $\beta^*$ is determined by solving a quartic polynomial, which can be computed either using known analytical solutions~\cite{chavez2022complete} or efficiently with numerical solvers. Since the resulting optimal path and cost landscape closely resemble those shown in \cref{optimal_path_2local_diagrams} and \cref{Conc_2Local_Plots}, additional visualisations are omitted for brevity.

We note that any Ising model with fields can be transformed into an equivalent model without fields using the transformation in \cref{ising_linearise_eqn}, enabling the use of \cref{qaoa_2_local_opt_params_thm} for simpler optimisation. While this transformation is straightforward, it is not always efficient. In certain cases---especially for problems where the graph representation has a bounded degree, with the maximum degree much smaller than the number of vertices---transforming the model can actually increase the complexity of finding the optimal parameters. The following proposition illustrates this with an example.
\begin{proposition}\label{edge_only_is_bad_cor} 
Let $G$ be a triangle-free graph with $n$ vertices and a maximum degree $D_{\max }$, where $D_{\max } \ll n$. Assume that all edge and node weights are $\pm 1$. When the corresponding Ising model with external fields is transformed into an equivalent model without external fields, the sampling period $\Delta_\gamma$ for QAOA$_1$ is reduced by a factor of $\frac{n}{D_{\max }+1}$.
\end{proposition}
The proof of this proposition is provided in \cref{edge_only_is_bad_proof}. \Cref{edge_only_is_bad_cor} shows that for bounded-degree graphs with $D_{\max } \ll n$, eliminating external fields can substantially worsen the sampling requirements for QAOA $_1$. In contrast, for models such as the Sherrington-Kirkpatrick model, the effect on optimisation complexity is minor because the maximum degree increases by only 1. Nevertheless, even Hamiltonian-equivalent formulations can yield different approximate $\mathrm{QAOA}_1$ solutions due to the algorithm's locality constraints.


\subsection{Optimising the Univariate Function} \label{opt_univar_subsubsec}

In this section, we focus on optimising the variable $\gamma$ for the univariate functions defined in \cref{qaoa_2_local_opt_params_thm} and \cref{qaoa_2_local_fields_opt_params_thm}. Consider a univariate trigonometric polynomial of degree $K$ defined in \cref{bw_per_fun_eqn}. Polynomial-time algorithms exist for computing all roots of such trigonometric polynomials~\cite{boyd2006computing, boyd2013comparison, kalantari2022characterization}, enabling the efficient determination of their optimal values. However, these methods are not directly applicable to our problem for two main reasons. First, \cref{QAOA_Gamma_Linear_Eqn} includes a square root, which prevents it from being expressed in the standard trigonometric polynomial form of \cref{bw_per_fun_eqn}. Second, \cref{qaoa_field_opt_par_linear_eqn} involves two variables, $(\gamma, \beta)$, where the optimal $\beta^*_{\gamma}$ depends on $\gamma$. This dependency prevents us from applying univariate optimisation methods that assume a fixed parameter.

To address these challenges, we perform a grid search over the parameter space $(\gamma, \beta)$, focusing exclusively on pairs of the form $(\gamma, \beta^*_{\gamma})$. Specifically, for each $\gamma$, we consider only the optimal $\beta^*_{\gamma}$ that minimises the objective function. This approach effectively reduces the problem to a univariate optimisation task, enabling the application of established algorithms such as Hansen's subdivision algorithm based on interval analysis~\cite{hansen1979global, hansen2003global}.

The subdivision algorithm iteratively partitions the domain into progressively smaller intervals. In each iteration, the algorithm evaluates the midpoint of every interval to identify the presence of global minima using a rejection criterion. Intervals that are unlikely to contain the global minima are discarded, while those that may still contain them are further subdivided. This process continues until the size of each interval falls below a predefined tolerance value. At this point, the midpoints of all remaining intervals are computed, and the midpoint with the maximum value is returned. For trigonometric polynomials, analogous subdivision methods have been developed to estimate their maximum modulus in both univariate and multivariate settings, utilising Steckin's lemma as the rejection criterion~\cite{green1999calculating, de2009finding}.

\begin{algorithm}[!t]
\SetAlgoLined
\SetArgSty{}
\SetKwInput{KwData}{Input}
\SetKwInput{KwResult}{Output}

\KwData{Initial Interval Spacing $\Delta_{\gamma}^{(0)}$, Maximum Frequency $\nu_{\max}$, Tolerance $\epsilon$}
\KwResult{Near-Optimal Angle $\gamma^*$}

$\Delta_{\gamma} \gets \Delta_{\gamma}^{(0)}$\;
Initialise list of intervals: 
$$\mathcal{I} \gets \{ [0, \Delta_{\gamma}], [\Delta_{\gamma}, 2\Delta_{\gamma}], \ldots, [\pi - \Delta_{\gamma}, \pi] \} ;$$

\Repeat{$\tilde{q} \left( \sec\left( \nu_{\max} \Delta_\gamma \right) - 1 \right) < \epsilon$}{
    \tcp{Initialise empty list for current subdivisions}
    $\mathcal{I}_{\text{current}} \gets \emptyset$\;
    $\tilde{q} \gets -\infty$\;

    \tcp{Evaluate the objective function at midpoints of all active intervals}
    \ForEach{$I \in \mathcal{I}$}{
        $t_I \gets \text{midpoint of } I$\;
        Compute $q(t_I)$\;
        \If{$q(t_I) > \tilde{q}$}{
            $\tilde{q} \gets q(t_I)$\;
        }
    }

    \tcp{Prune intervals not meeting the threshold and prepare for subdivision}
    \ForEach{$I \in \mathcal{I}$}{
        $t_I \gets \text{midpoint of } I$\;
        \If{$q(t_I) \geq \tilde{q} \cos \left( \nu_{\max} \Delta_\gamma \right)$}{
            Split $I$ into $I_{\text{left}}$ and $I_{\text{right}}$\;
            Add $I_{\text{left}}$ and $I_{\text{right}}$ to $\mathcal{I}_{\text{current}}$\;
        }
    }

    \tcp{Update intervals and spacing for the next iteration}
    $\mathcal{I} \gets \mathcal{I}_{\text{current}}$\;
    $\Delta_{\gamma} \gets \dfrac{\Delta_{\gamma}}{2}$\;
}

$\gamma^* \gets \displaystyle\argmax_{I \in \mathcal{I}} q(t_I)$\;

\Return{$\gamma^*$}

\caption{Subdivision Algorithm for Estimating the Optimal Angle $\gamma^*$}
\label{subdivision_algorithm}
\end{algorithm}

By analytically determining $\beta^*_{\gamma}$ for each $\gamma$, we effectively transform the original multivariate optimisation problem into a univariate one. This reduction allows us to apply efficient subdivision techniques tailored for univariate trigonometric functions. The specific subdivision algorithm used for computing the optimum modulus of univariate trigonometric functions, as attributed to Green~\cite{green1999calculating}, is detailed in algorithm \ref{subdivision_algorithm}. For a detailed explanation on using Steckin's lemma as a rejection criterion, refer to \cref{steck_sec_app}.

To apply algorithm \ref{subdivision_algorithm} to our expressions \cref{QAOA_Gamma_Linear_Eqn} and \cref{qaoa_field_opt_par_linear_eqn}, we square the expressions. Specifically, the function $q$ in algorithm \ref{subdivision_algorithm} is defined as $\langle\gamma, \beta^*_{\gamma}| H_P|\gamma, \beta^*_{\gamma}\rangle^2$. For the case of \cref{qaoa_field_opt_par_linear_eqn}, it is crucial that $\beta^*_{\gamma} \in \mathcal{B}_\gamma$ is the value that minimises the objective function before squaring it. Additionally, since all $\gamma$ terms have a constant prefactor of 2, the domain that the subdivision algorithm searches is effectively reduced from $[0, 2\pi]$ to $[0, \pi]$. Algorithm~\ref{subdivision_algorithm} maintains a worst-case time complexity of $\mathcal{O}(N/\sqrt{\varepsilon})$ for estimating the maximum modulus and consequently the optimum of \cref{bw_per_fun_eqn} (see \cref{subdiv_linear_thm}). 

We note that there exist other polynomial-time algorithms for computing the provably optimal $(\gamma^*, \beta^*)$ angles for QAOA$_1$ on Ising models.  Specifically, to determine the optimal $(\gamma^*, \beta^*)$ that minimise $\langle H_P (\gamma, \beta) \rangle$, we solve the system of equations:
\begin{equation}
    \frac{\partial \langle H_P (\gamma, \beta) \rangle}{\partial \gamma} = \frac{\partial \langle H_P (\gamma, \beta) \rangle}{\partial \beta} = 0. \label{bivar_poly_eqn}
\end{equation}
This involves expressing \cref{bivar_poly_eqn} as standard polynomials and applying algorithms designed to find all roots of bivariate polynomials in polynomial time~\cite{emeliyanenko2012complexity, plestenjak2016roots}. Although these methods are more computationally intensive than Algorithm~\ref{subdivision_algorithm}, they illustrate that optimising QAOA$_1$ parameters need not be $\mathsf{NP}$-Hard under certain conditions, such as when $\gamma$ is effectively periodic in $[0, \pi]$ and the distribution of problem weights does not scale with the problem size.


This stands in contrast to the results of Bittel et al.~\cite{bittel2021training}, who showed that optimising the parameters of VQAs, and QAOA in particular, is $\mathsf{NP}$-Hard. This hardness, however, is not a property of arbitrary Ising Hamiltonians; it is established through explicit reductions from MaxCut, in which carefully tailored problem instances are constructed so that an efficient parameter optimiser would necessarily solve an $\mathsf{NP}$-Hard problem. For QAOA$_1$, such instances fall into one of two regimes. The first is when the problem weights are
incommensurable, rendering the expectation value along $\gamma$ non-periodic, so that no finite period can be exploited to confine the search to a bounded interval. The second is when the weights grow with the problem size: their single-layer reduction employs a spectrum spanning exponentially many orders of magnitude, which forces the maximum
frequency of the cost landscape---and hence the sampling resolution required to reconstruct it---to scale exponentially, demanding exponential precision to resolve the global optimum. A related result shows that QAOA$_n$ remains $\mathsf{NP}$-Hard even for bounded integer weights, but only at the expense of a circuit depth that grows with the
number of variables. Crucially, none of these worst-case regimes apply to the setting considered here. For QAOA$_1$ on Ising models with commensurable weights that do not scale exponentially with the system size, the expectation value along $\gamma$ is periodic on $[0,\pi]$ and bandlimited with a polynomially bounded maximum frequency, so that the optimal angles $(\gamma^*,\beta^*)$ can be determined in polynomial time. Consequently, these instances of QAOA$_1$ do not fall under the $\mathsf{NP}$-Hard regime identified by Bittel et al.~\cite{bittel2021training}.

\subsection{Globally Optimal Parameters for Regular Graphs} \label{first_local_subsubsec}

We present our final theoretical result: in large, dense, and heavily weighted Ising model instances, the globally optimal parameter $\gamma^* \in \mathbb{R}^+$ tends to concentrate near zero, often coinciding with the first local optimum. Consequently, we propose simplifying the optimisation process by performing gradient descent in the vicinity of zero to determine $\gamma^*$. This approach is motivated by empirical observations showing that the first local optimum near zero typically yields globally optimal parameters~\cite{shaydulin2023parameter, boulebnane2023peptide, brandhofer2022benchmarking}. 

We rigorously prove that for the family of regular graphs—including triangle-free regular graphs, complete graphs, and all intermediate cases—the optimal $\gamma^*$ in the infinite size limit sharply concentrates near zero. This concentration depends on the graph's degree, the second moments of the weight distribution, and the variance of the weights. The following theorem formalises this result. By assuming that $\gamma \propto D^{-1/2}$, we derive closed-form expressions for numerically evaluating the scaled expectation cost of QAOA$_1$ and obtain expressions for the optimal $\gamma^*$ based on the leading-order terms of the expectation value. 
\begin{theorem} \label{fst_local_opt_thm}
  Consider a weighted Ising model defined on a $(D+1)$-regular graph $G = (V, E)$. Let the field strengths $\{h_i\}_{i \in V}$ and coupling strengths $\{J_{uv}\}_{\{u,v\} \in E}$ be i.i.d.\ random variables with finite second moments. For any distinct vertices $u, v \in V$, let $\left|\overline{F}_{u v}\right|=a D^\lambda$ and $\left|F_{u v}\right|=b D^\mu$, where $a, b \geq 0$ and $0 \leq$ $\lambda, \mu \leq 1$. Here, $\overline{F}_{u v}$ consists of neighbours of $\{u, v\}$ that do not form triangles, while $F_{u v}$ consists of those that do. Then, the scaled expected QAOA$_1$ cost function with respect to parameters $\beta \in \mathbb{R}$ and $\gamma = \alpha/\sqrt{D} \in \mathbb{R}^{+}$ can be expressed as follows:
  \begin{widetext}
      \begin{equation}
    \resizebox{0.935\textwidth}{!}{$
    \dfrac{\mathbb{E} \left[\xi_G(\alpha, \beta) \right]}{|E|} = 
        \begin{cases}
            \mathcal{C}_1(\alpha, \beta) + \mathcal{O}\left(D^{-1}\right), & b = 0, \, a = \lambda = 1, \\
            \mathcal{C}_1(\alpha, \beta) + \mathcal{C}_2(\alpha, \beta, 1, 1) + \mathcal{O}\left(D^{-1}\right), & a = 0, \, b = \mu = 1, \\
            \mathcal{C}_1(\alpha, \beta) + \mathcal{C}_2(\alpha, \beta, 1, b) + \mathcal{O}\left(D^{-1}\right), & a, b > 0, \, \lambda = \mu = 1, \\
            \mathcal{C}_1(\alpha, \beta) + \mathcal{C}_2(\alpha, \beta, b, b) + \mathcal{O}\left(D^{-1}\right), & a, b > 0, \, \mu = 1, \, \lambda < \frac{1}{2}, \\
            \mathcal{C}_1(\alpha, \beta) + \mathcal{C}_2(\alpha, \beta, b, b) \left(1 - 4a\alpha^2 \mathbb{E}[J^2] \frac{1}{\sqrt{D}}\right) + \mathcal{O}\left(D^{\max(2\lambda-2, -1)}\right), & a, b > 0, \, \mu = 1, \, \lambda = \frac{1}{2}, \\
            \mathcal{C}_2(\alpha, \beta, b, b) \left(1 - 4a\alpha^2 \mathbb{E}[J^2] D^{\lambda-1} \right) + \mathcal{O}\left(\frac{1}{\sqrt{D}}\right), & a, b > 0, \, \mu = 1, \, \lambda > \frac{1}{2}, \\
            4b\alpha^2 \sin^2(2\beta) \mathbb{E}[J]^3 e^{-4a\alpha^2 \mathbb{E}[J^2]} D^{\mu-1} + \mathcal{O}\left(D^{\max(-\frac{1}{2}, 2\mu-2)}\right), & a, b > 0, \, \lambda = 1, \, \mu > \frac{1}{2}, \\
            \mathcal{C}_1(\alpha, \beta) + 4b\alpha^2 \sin^2(2\beta) \mathbb{E}[J]^3 e^{-4a\alpha^2 \mathbb{E}[J^2]} \frac{1}{\sqrt{D}} + \mathcal{O}\left(D^{-1}\right), & a, b > 0, \, \lambda = 1, \, \mu = \frac{1}{2}, \\
            \mathcal{C}_1(\alpha, \beta) + \mathcal{O}\left(D^{\mu-1}\right), & a, b > 0, \, \lambda = 1, \, \mu < \frac{1}{2},
        \end{cases}$}
    \label{big_equation}
\end{equation}
  \end{widetext}
where
\begin{equation}
    \mathcal{C}_1(\alpha, \beta) \! = \! 2\alpha\sin(4\beta) \mathbb{E} \! \left[J^2\right] \! e^{-2 \alpha^2 \mathbb{E} \left[J^2\right]} \! \frac{1}{\sqrt{D}},
\end{equation}
\begin{equation}
    \begin{aligned}
    \mathcal{C}_2(\alpha, \beta, \theta_1, \theta_2) &= \sin^2(2 \beta) \mathbb{E}\left[J\right] e^{-4 \theta_1 \alpha^2 \mathbb{E}\left[J^2\right]} \\
    & \quad \times \sinh\left(4 \theta_2 \alpha^2 \mathbb{E}\left[J\right]^2  \right).
\end{aligned}
\end{equation}
For leading-order terms, the optimal value of $\gamma^* \in \mathbb{R}^+$ is:
\begin{equation}
    \gamma^* = 
        \begin{cases}
            \eta(4), & b = 0, \, a = \lambda = 1, \\
            \zeta(1, 1), & a = 0, \, b = \mu = 1, \\
            \zeta(1, b), & a, b > 0, \, \lambda = \mu = 1, \\
            \zeta(b, b), & a, b > 0, \, \mu = 1, \, \lambda < 1, \\
            \eta(2a), & a, b > 0, \, \lambda = 1, \, \mu > \frac{1}{2}, \\
            \eta(4), & a, b > 0, \, \lambda = 1, \, \mu < \frac{1}{2},
        \end{cases}
\end{equation}
where
\begin{align}
    \eta(\theta) & = \frac{1}{\sqrt{\theta D \mathbb{E}\left[J^2\right]}}, \\
    \zeta(\theta_1, \theta_2) &= \hspace{-0.15cm} \sqrt{\frac{1}{8\theta_2 D \mathbb{E}\left[J\right]^2}  \ln \hspace{-0.15cm} \left( \frac{\theta_1 \mathbb{E}\left[J^2\right] + \theta_2 \mathbb{E}\left[J\right]^2}{\theta_1 \mathbb{E}\left[J^2\right] - \theta_2 \mathbb{E}\left[J\right]^2} \right)}.
\end{align}
Finally, the globally optimal parameter $\gamma^*$ also corresponds to the first local optimum for $\gamma \in \mathbb{R}^+$.
\end{theorem}
The proof of this theorem is provided in \cref{fst_local_opt_proof}. The assumption that $\gamma \in \Theta(D^{-1/2})$ in \cref{fst_local_opt_thm} is motivated by numerical observations showing that the optimal $\gamma$ for the unweighted MaxCut problem scales as $\Theta(D^{-1/2})$ (see Figure 1b of Boulebnane et al.~\cite{boulebnane2021predicting}) and by analytical results demonstrating that, for finite-size triangle-free regular graphs with weights drawn from an exponential distribution, $\gamma$ also scales as $\Theta(D^{-1/2})$ (see theorem 1 of Sureshbabu et al.~\cite{sureshbabu2024parameter}). 

\begin{figure*}[t]
    \centering
    \subfloat[Ising Model without Fields]{
        \includegraphics[height=0.33\textheight]{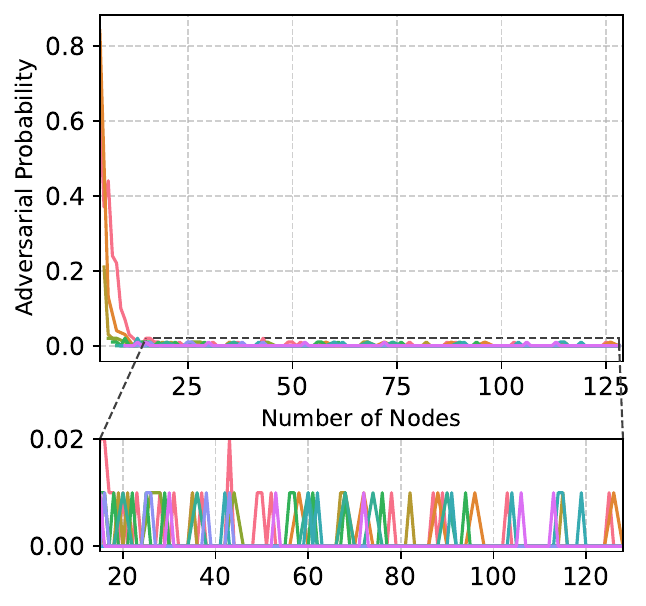}
        \label{prob_non_first_global_opt_plot_a}
    }
    \subfloat[Ising Model with Fields]{
    \hspace{-0.6cm}
        \includegraphics[height=0.33\textheight]{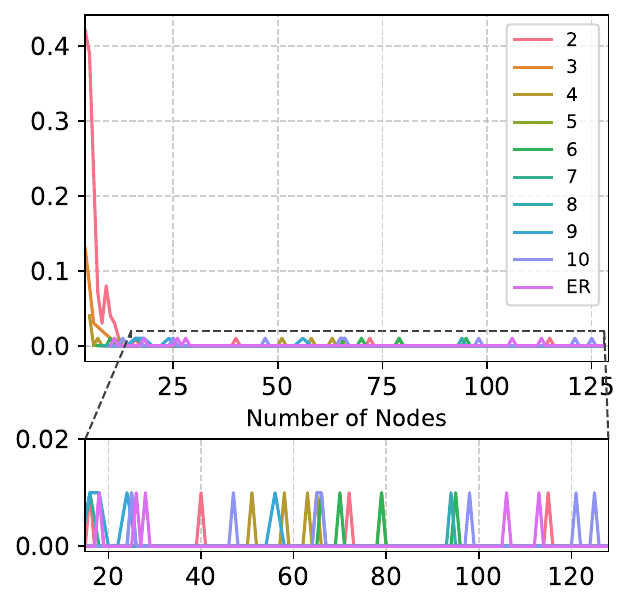}
        \label{prob_non_first_global_opt_plot_b}
    }
    \caption[Probability of Adversarial Graphs in QAOA$_1$]{\textbf{Probability of Adversarial Graphs in QAOA$_1$.} \justifying These plots illustrate the probability of encountering adversarial instances in Ising models for $D$-regular graphs ($2 \leq D \leq 10$) and Erdős-Rényi (ER) graphs with edge probability $1/2$. (a) presents results for Ising models without external fields, whereas (b) illustrates results for models with external fields. Graph sizes range from 4 to 128 nodes, with edge weights sampled as integers from a Gaussian distribution (mean 0, variance 100). For each configuration, 100 random graph instances were evaluated using two methods: (1) a full line search over $\gamma \in [0, \pi]$ using algorithm \ref{subdivision_algorithm}, and (2) gradient descent initialised near $\gamma = \Delta_{\gamma}/2$. Adversarial instances are defined as those where the optima obtained from the two methods differ. The probability of adversarial instances stabilises around 0.01 for graphs with $n > 14$, implying that only 1 in 100 graphs exhibits adversarial behaviour at larger sizes. Smaller and sparser graphs show a higher occurrence of adversarial instances, reflecting the increased likelihood of misalignment between local and global optima in these cases.}
    \label{prob_non_first_global_opt_plot}
\end{figure*}

We note that Sureshbabu et al.~\cite{sureshbabu2024parameter} recently proved a special case of our result for triangle-free regular graphs without external fields. The increased complexity of \cref{fst_local_opt_thm} compared to their result arises because, in general, the exact number of triangles an edge belongs to is not known, except in specific cases such as triangle-free graphs and complete graphs. To address this, we introduce parameters to quantify both the number of triangles an edge participates in and the number of its neighbouring edges that are not part of any triangle. This parametrisation naturally recovers triangle-free and complete graphs as limiting cases. The first two cases in \cref{big_equation} correspond to these two extremes, while the remaining cases handle intermediate scenarios where the graph is neither triangle-free nor complete. Notably, when edge weights are drawn from a distribution with zero mean, \cref{big_equation} simplifies significantly. Specifically, for $\mathbb{E}[J] = 0$, it reduces to $\mathcal{C}_1(\alpha, \beta)$ with the globally optimal $\gamma^* = \eta(4)$. Finally, since \cref{fst_local_opt_thm} does not assume the sampled weights are commensurable, our finding that the first local optimum often coincides with the global optimum remains valid even when the $\gamma$ landscape is non-periodic.

While \cref{fst_local_opt_thm} establishes that, in the infinite-size limit of regular graphs, the global optimum coincides with the first local optimum, empirical evidence from finite-size problems~\cite{shaydulin2023parameter, boulebnane2023peptide, brandhofer2022benchmarking} suggests that this alignment frequently holds in practice as well. We observe this phenomenon for finite-size $D$-regular graphs as well, as shown in \cref{optimal_path_2local_diagrams} and \cref{Conc_2Local_Plot_d}. This behaviour can be intuitively explained using the closed-form expressions in \cref{qaoa_2_local_opt_params_thm} and \cref{qaoa_2_local_fields_opt_params_thm}. When signals like $\cos(\alpha x)$ and $\cos(\beta x)$ with different frequencies interfere, they align constructively when in phase and destructively when out of phase, with the interference points depending on their oscillation rates. Similarly, the univariate cost function for QAOA$_1$ along $\gamma$ comprises cosine terms with frequencies determined by the problem instance. As the instance size, density, or weights increase, the range of frequencies broadens, creating a superposition of oscillating signals. At arbitrary $\gamma$, these signals are generally out of phase, leading to destructive interference and a reduced cost function amplitude. Near $\gamma = 0$, however, all cosine terms converge to $\cos(0) = 1$, aligning perfectly in phase and maximising constructive interference, which produces the highest amplitude. As $\gamma$ deviates from zero, the signals fall out of phase, causing destructive interference and diminishing the amplitude.

Nevertheless, it is important to recognise that the first local optimum need not always correspond to the global optimum, as adversarial graphs can be constructed where this alignment does not hold (see Figure 2 of~\cite{shaydulin2023parameter}). However, if the majority of finite instances exhibit the behaviour where the first local optimum aligns with the global optimum, the optimisation process over $\gamma \in [0, \pi]$ with interval size $\Delta_\gamma$ can be significantly simplified. To evaluate the prevalence of adversarial instances, we tested weighted $D$-regular graphs ($2 \leq D \leq 10$) and Erdős-Rényi graphs (with edge probability $1/2$) for graph sizes ranging from 4 to 128, generating 100 random instances per graph type. Edge weights were sampled from a Gaussian distribution with mean 0 and variance 100, rounded to the nearest integer. Two optimisation methods were employed: a line search over $\gamma \in [0, \pi]$ using algorithm \ref{subdivision_algorithm} with $\Delta_\gamma$ determined by \cref{qaoa_min_samples_thm}, and gradient descent initialised at $\gamma = \Delta_\gamma / 2$ to provide an initial guess that avoids overshooting the first local optimum. By comparing the $\gamma^*$ values obtained from both methods, we found that adversarial instances are rare. Their occurrence is concentrated in small, sparse graphs and decreases exponentially with increasing graph size or density, as illustrated in \cref{prob_non_first_global_opt_plot}.

These empirical findings lay the groundwork for streamlining the optimisation process of QAOA$_1$ on large Ising models. Specifically, the comprehensive line search over $\gamma \in [0,\pi]$ can be replaced with a gradient descent approach initialised near $\gamma \approx 0$, guided by the maximum sampling interval $\Delta_\gamma$ from \cref{qaoa_min_samples_thm}. This parameter serves a dual purpose: it provides an initial guess, $\gamma < \Delta_\gamma$, placing the algorithm close to the anticipated global optimum while preventing overshooting, and it establishes an upper bound on the gradient descent step size. By ensuring that updates remain smaller than $\Delta_\gamma$, the algorithm operates within a region where the cost function behaves predictably, thereby minimising the risk of overshooting the first local optima.

\section{Numerical Results} \label{num_subsec}

\begin{figure*}[htpb]
    \centering
    \subfloat[Benchmark Results for 128 Vertex Graphs]{
        \includegraphics[width=0.5\textwidth]{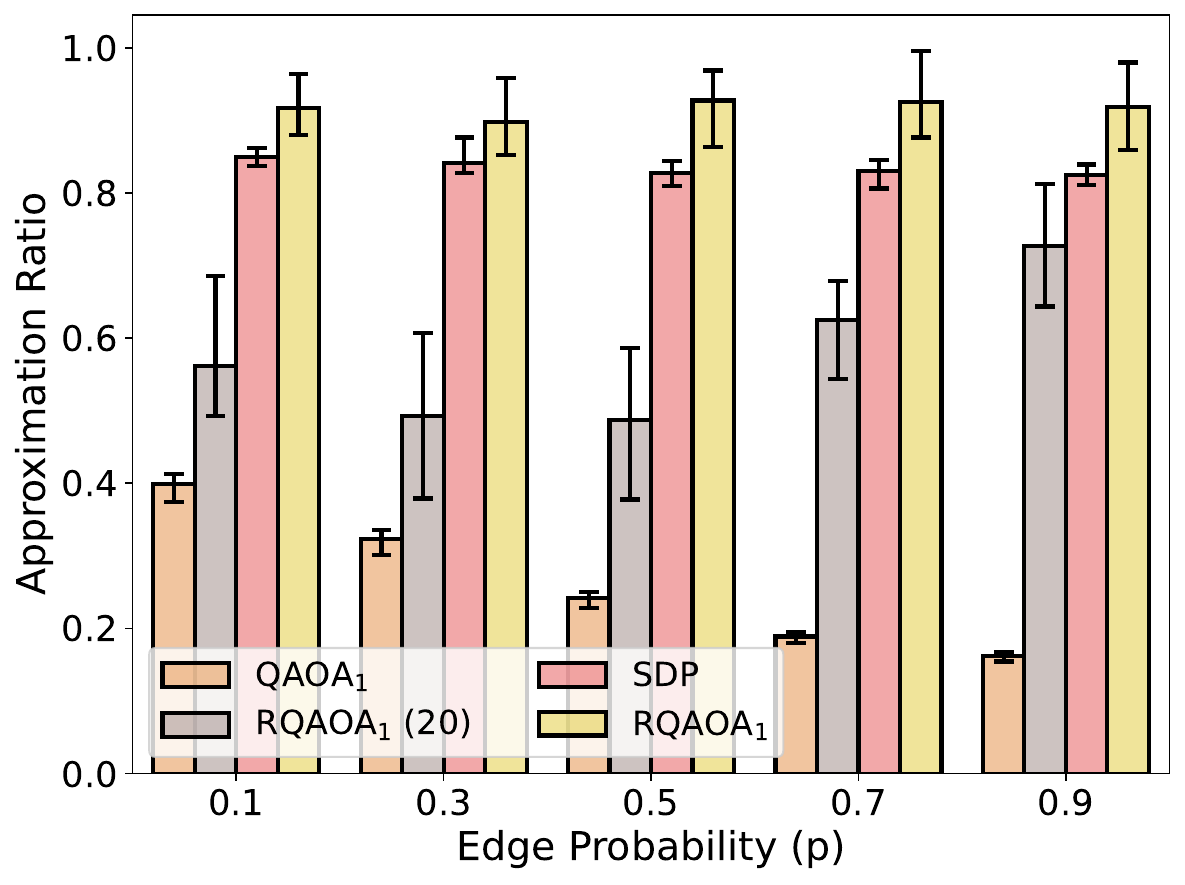}
        \label{rqaoa_benchmark_er_plot_a}
    }
    \subfloat[Benchmark Results for 256 Vertex Graphs]{
        \includegraphics[width=0.5\textwidth]{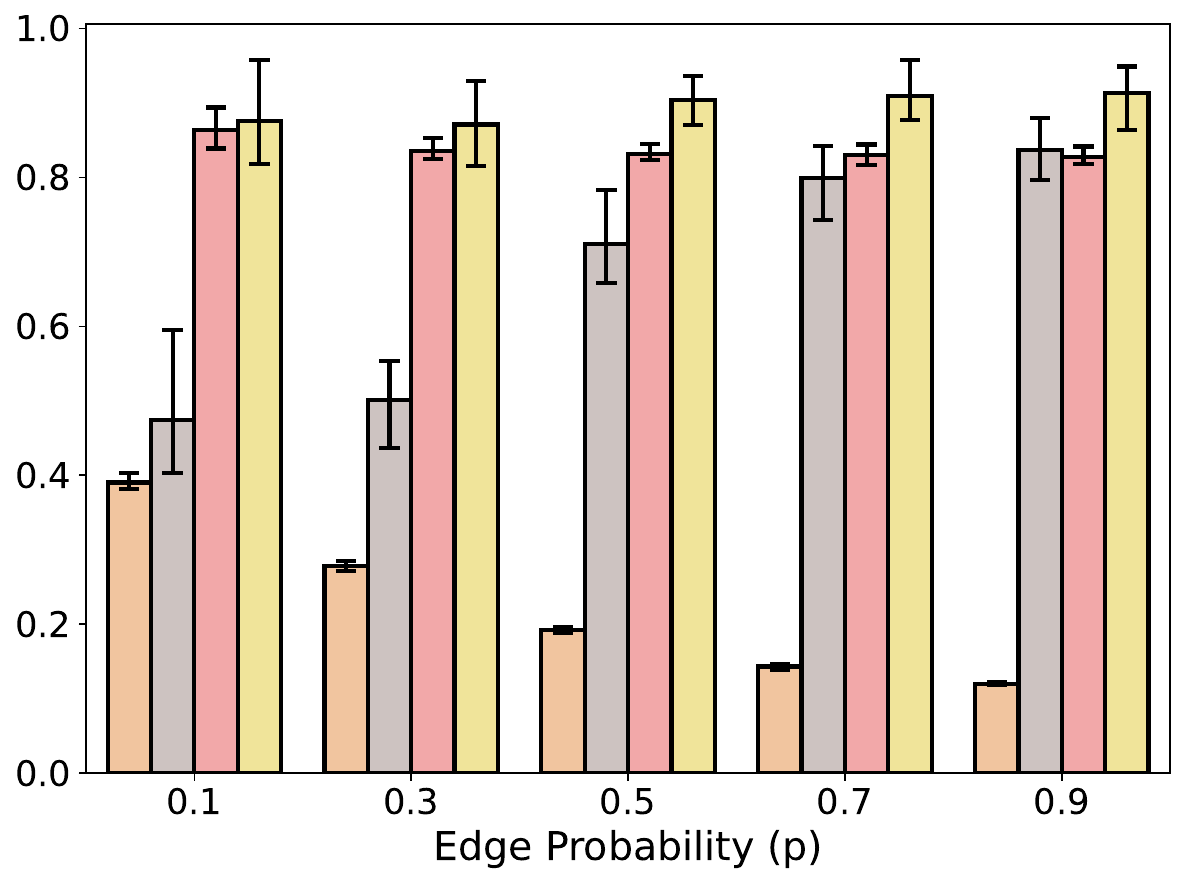}
        \label{rqaoa_benchmark_er_plot_b}
    }
    \caption{\textbf{Benchmark Results for Ising Models Without External Fields.} 
\justifying
These plots compare four algorithms---QAOA\textsubscript{1}, RQAOA\textsubscript{1}, RQAOA\textsubscript{1}(20), and a classical SDP---for approximating the ground state of the Ising model without external fields. We generate weighted Erdős--Rényi graphs with varying edge probabilities $p$ for (a) 128-vertex and (b) 256-vertex instances. Each bar indicates the \emph{mean} approximation ratio over 20 random graphs at a given $p$; the whiskers do not denote standard deviations, but instead span the \emph{minimum} and \emph{maximum ratios} observed across those instances.
QAOA\textsubscript{1} is optimised via gradient descent initialised near $\gamma \approx 0$; RQAOA\textsubscript{1}(20) employs a coarse line search of 20 samples in $[0,\pi]$ followed by gradient descent on the best found point; RQAOA\textsubscript{1} uses gradient descent from $\gamma \approx 0$ with $\eta=120$ recursive steps for 128-vertex graphs and $\eta=248$ for 256-vertex graphs. The SDP adapts the Goemans--Williamson relaxation by replacing the MaxCut cost function with the Ising Hamiltonian, and each solution is rounded from 1024 random hyperplanes.}

    \label{rqaoa_benchmark_er_plot}
\end{figure*}

To validate our parameter tuning strategy for QAOA$_1$ on Ising models, we applied it to Recursive QAOA$_1$ (RQAOA$_1$)~\cite{bravyi2020obstacles} (refer to \cref{RQAOA_Subsec} for an explanation of RQAOA). We focused on large, dense, weighted problem instances to assess the effectiveness of our approach. Since any QUBO problem can be mapped to finding the ground state of an arbitrary Ising model---represented as an undirected graph with nodes as spins and edges as interaction terms---demonstrating that RQAOA$_1$ performs well across a broad distribution of graphs implies its applicability to most QUBO problems. Given the impracticality of testing RQAOA$_1$ on all possible graph structures, we utilised Erdős-Rényi graphs~\cite{erdds1959random, gilbert1959random} as a representative approximation. Problem instances were generated using random Erdős-Rényi graphs with 128 and 256 vertices, varying the edge probability $p \in \{0.1, 0.3, 0.5, 0.7, 0.9\}$. The edge weights were sampled from a Gaussian distribution with a mean of 50 and a variance of 30, then rounded to the nearest integer. For Ising models that include external fields, the external field values were similarly sampled from a Gaussian distribution (mean 40, variance 20) and rounded to integers. For each combination of graph size and edge probability, we generated 20 random instances for both Ising models with and without external fields.

The numerical simulations of QAOA$_1$ and RQAOA$_1$ leveraged analytical solutions derived from \cref{qaoa_2_local_opt_params_thm} and \cref{qaoa_2_local_fields_opt_params_thm}. To assess performance, we computed approximation ratios by comparing the objective value achieved by each algorithm considered in our benchmark with the corresponding reference value returned by the GUROBI solver~\cite{gurobi}, an industry-standard optimisation tool renowned for its efficiency and reliability. Each problem instance was solved with GUROBI using a fixed wall-clock budget. In many cases, the incumbent objective stabilised before termination, with no further improvement in the best-found solution observed while the MIP gap\footnote{The MIP Gap refers to the relative difference between the best integer solution and the best possible bound on the optimal solution in mixed-integer programming. Smaller MIP Gaps indicate solutions closer to optimality.} continued to decrease slowly. When the final MIP gap was zero, the returned solution was certified optimal; otherwise, we used the best incumbent solution at termination and report the associated final MIP gap. Thus, the approximation ratios discussed in this section are computed relative to GUROBI reference values, with exact optimality certified whenever the final MIP gap is zero. In the remainder of this section, we present detailed benchmark results for the Ising models, highlighting the performance and effectiveness of our parameter tuning strategy in optimising QAOA$_1$ and RQAOA$_1$.


\subsection{Ising Models without Fields} \label{rqaoa_ising_benchmark_subsec}

\begin{figure*}[!t]
    \centering
    \subfloat[Benchmark Results for 128 Vertex Graphs]{
        \includegraphics[width=\textwidth]{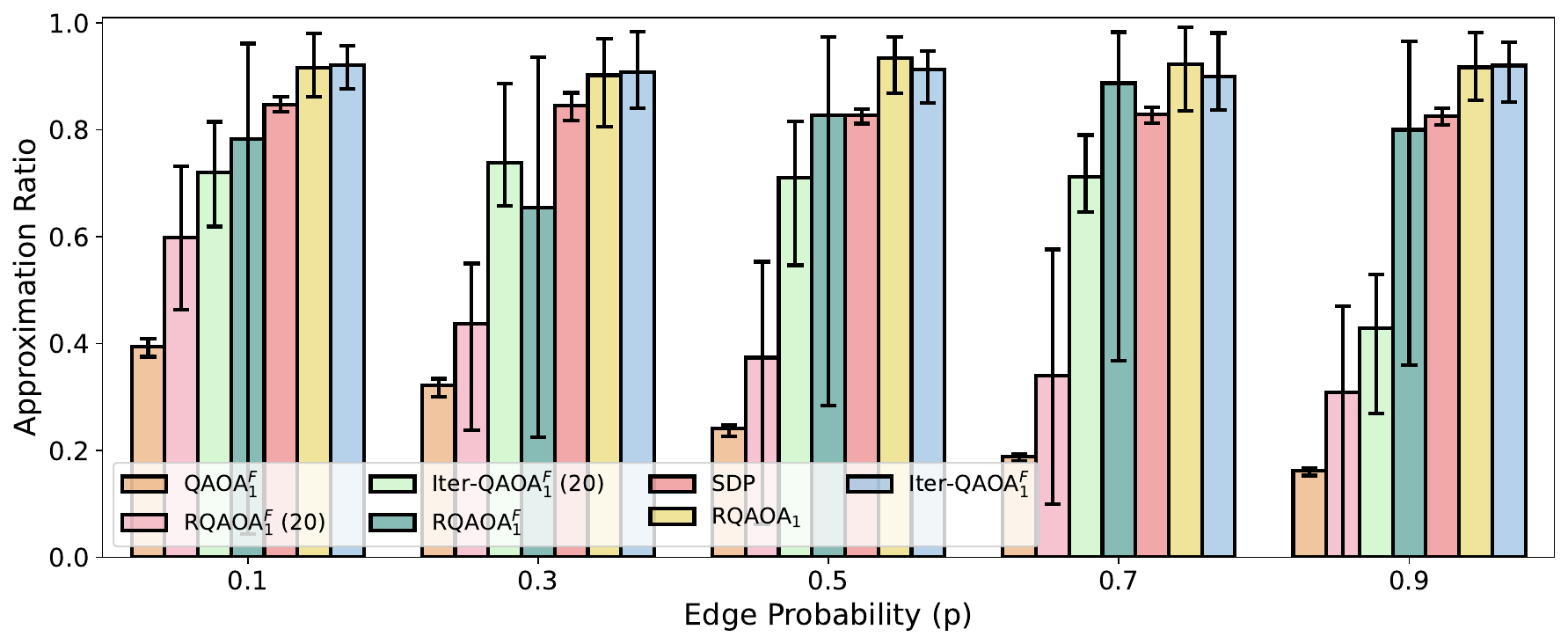}
        \label{fields_benchmark_plot_a}
    }\\
    \subfloat[Benchmark Results for 256 Vertex Graphs]{
        \includegraphics[width=\textwidth]{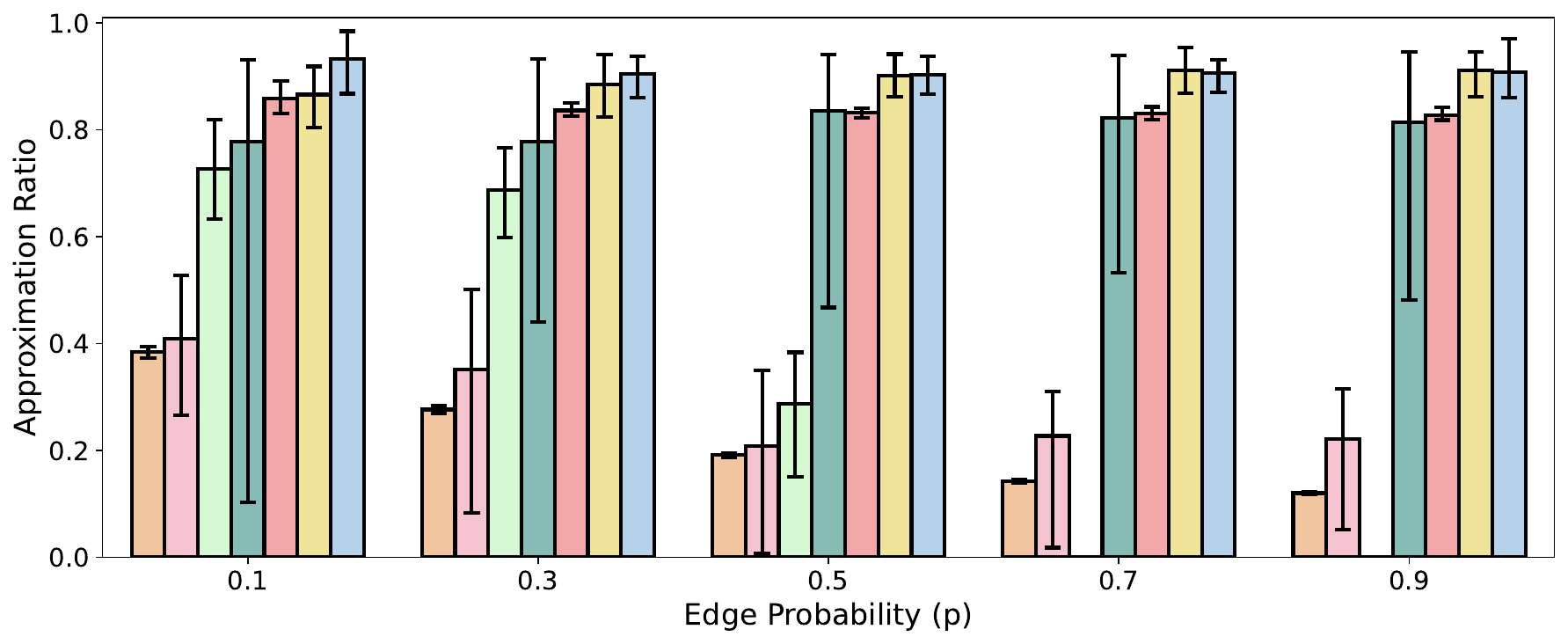}
        \label{fields_benchmark_plot_b}
    }
\caption{\textbf{Benchmark Results for Ising Models with External Fields.}
\justifying
These plots compare the performance of seven algorithms for approximating the ground state of the Ising model with external fields. We generate weighted Erdős–Rényi graphs with varying edge probabilities $p$, shown for (a) 128-vertex and (b) 256-vertex instances. Each bar indicates the \emph{mean} approximation ratio over 20 random graphs at a given $p$; the whiskers do not denote standard deviations, but instead span the \emph{minimum} and \emph{maximum ratios} observed across those instances.
Algorithms labelled with a superscript $F$ (e.g., QAOA\textsubscript{1}\textsuperscript{F}) employ parameter expressions tailored to Ising models with fields; those without $F$ (e.g., RQAOA\textsubscript{1}) convert the model to a field‐free form before parameter optimisation. 
QAOA\textsubscript{1}\textsuperscript{F} is initialised near $\gamma \approx 0$ for gradient descent; 
RQAOA\textsubscript{1}\textsuperscript{F}(20) and Iter-QAOA\textsubscript{1}\textsuperscript{F}(20) each use a coarse line search of 20 samples in $[0,\pi]$ prior to gradient‐descent refinement. RQAOA and Iter-QAOA variants are run for $\eta=120$ steps (128‐vertex) or $\eta=248$ steps (256‐vertex). 
The classical SDP adapts the Goemans–Williamson relaxation by replacing the MaxCut cost function with the Ising Hamiltonian, and each SDP solution is rounded from 1024 random hyperplanes.}

    \label{fields_benchmark_plots}
\end{figure*}

To benchmark RQAOA$_1$ on Ising models without fields, we compared its performance against four other algorithms: QAOA$_1$, two variants of RQAOA$_1$, and a classical semidefinite program (SDP). QAOA$_1$ parameters were optimised using gradient descent initialised near $\gamma \approx 0$ on \cref{QAOA_Gamma_Linear_Eqn} and is labelled as QAOA$_1$ in the results. RQAOA$_1$ was tested in two configurations: RQAOA$_1$ (20), which used a coarse line search of 20 samples for $\gamma \in [0, \pi]$ followed by gradient descent on the best found point, and RQAOA$_1$, which employed gradient descent initialised near $\gamma \approx 0$ on \cref{QAOA_Gamma_Linear_Eqn}. Both variants were tested with $\eta = 120$ and $\eta = 248$ recursive steps for 128- and 256-vertex graphs, respectively. The final comparison was with a classical SDP based on the Goemans--Williamson algorithm~\cite{goemans1995improved}, adapted for Ising models by replacing the MaxCut cost function with the Ising model Hamiltonian. After solving the relaxed SDP, the solution was rounded using 1024 random hyperplanes.

The benchmark results are summarised in \cref{rqaoa_benchmark_er_plot}. QAOA$_1$ consistently yields the lowest approximation ratios among the tested algorithms. Both RQAOA$_1$ (20) and RQAOA$_1$ outperform QAOA$_1$ across all instances, while the SDP generally outperforms RQAOA$_1$ (20), often by a significant margin. One exception occurs on 256-node graphs with $p = 0.9$, where RQAOA$_1$ (20) slightly exceeds the performance of the SDP (\cref{rqaoa_benchmark_er_plot_b}). In contrast, RQAOA$_1$ consistently outperforms the SDP across all values of $p$ for both graph sizes. These results underscore the effectiveness of the parameter tuning strategy outlined in \cref{qaoa_2_local_opt_params_thm} and show that, with optimal parameters, RQAOA$_1$ can outperform classical SDPs, highlighting the critical role of parameter optimisation in VQAs.

\subsection{Ising Models with Fields}

Benchmarking RQAOA$_1$ for Ising models with external fields revealed inconsistent performance, the underlying cause of which remains uncertain. To address this, we introduced a simplified iterative rounding procedure that demonstrated greater robustness and consistently produced better results. This modified approach, named Iterative-QAOA (Iter-QAOA), is detailed in \cref{IterQAOA_Subsec}.

We benchmarked RQAOA$_1$ and Iter-QAOA$_1$ on Ising models with external fields, comparing their performance to seven algorithms: QAOA$_1$, three RQAOA$_1$ variants, two Iter-QAOA$_1$ variants, and a classical SDP. Algorithms labelled with a superscript $F$ used the parameter expressions from \cref{qaoa_2_local_fields_opt_params_thm} for Ising models with fields, while those without $F$ used \cref{qaoa_2_local_opt_params_thm} for models without fields. For QAOA$_1^F$, parameters were optimised via gradient descent initialised near $\gamma \approx 0$ on \cref{qaoa_field_opt_par_linear_eqn}. RQAOA$_1^F$ was tested in two configurations: RQAOA$_1^F$ (20), which employed a coarse line search with 20 samples for $\gamma \in [0, \pi]$ followed by gradient descent on the best found point, and RQAOA$_1^F$, which directly used gradient descent initialised near $\gamma \approx 0$. RQAOA$_1$, by contrast, converted the Ising model with fields to one without fields and optimised parameters using QAOA$_1$. Similarly, Iter-QAOA$_1^F$ was tested in two configurations: Iter-QAOA$_1^F$ (20), which followed the same coarse line search strategy as RQAOA$_1^F$ (20), and Iter-QAOA$_1^F$, which used gradient descent initialised near $\gamma \approx 0$. All RQAOA and Iter-QAOA variants were run for $\eta = 120$ and $\eta = 248$ recursive steps on 128- and 256-vertex graphs, respectively. Finally, the SDP was benchmarked on equivalent Ising models without fields. The results, shown in \cref{fields_benchmark_plots}, highlight the relative performance of each algorithm.

As shown in \cref{fields_benchmark_plots}, Iter-QAOA$_1^F$ consistently outperforms RQAOA$_1^F$ across all instances. This is reflected not only in the higher mean approximation ratios, but also in the generally narrower min-max whiskers, which indicate more stable performance across all the random instances and are consistent with reduced sensitivity to parameter tuning errors, particularly on sparser graphs. 
In contrast, RQAOA$_1^F$ shows higher variability, with worst-case results occasionally falling below both its coarsely optimised counterpart and QAOA$_1^F$. Iter-QAOA$_1^F$ also outperforms the SDP by a significant margin and is competitive with RQAOA$_1$, occasionally surpassing it. Similar trends are observed for the coarsely optimised variants: Iter-QAOA$_1^F$ (20) generally performs better than RQAOA$_1^F$ (20), except in 256-vertex cases with edge probabilities $p \in \{0.7, 0.9\}$. In those instances, Iter-QAOA$_1^F$ (20) produces negative approximation ratios---worse than random guessing---likely due to the use of partial information in these cases.

These results highlight limitations in RQAOA’s performance on Ising models with fields, raising questions about its applicability to more general pseudo-Boolean optimisation problems with higher-degree terms. This motivates further exploration of whether RQAOA and related quantum algorithms~\cite{bravyi2020obstacles, bravyi2022hybrid, finvzgar2024quantum, dupont2023quantum, patel2024reinforcement} can effectively generalise beyond QUBO. Nonetheless, Iter-QAOA$_1^F$ consistently outperformed the SDP across all benchmarks, underscoring the robustness and efficacy of our parameter-setting strategy outlined in \cref{qaoa_2_local_fields_opt_params_thm}, which was the central focus of this particular study.

\section{Discussion}

In this work, we focused on optimising QAOA$_1$ for arbitrary Ising models, specifically addressing the challenges and misconceptions surrounding its parameter optimisation. Although QAOA$_1$ is governed by only two parameters, $(\gamma, \beta)$, it is commonly assumed that a simple coarse grid search followed by local minimisation is sufficient for finding optimal values. Contrary to this belief, our analysis revealed that the expectation value landscape is highly oscillatory. By expressing it as a partial Fourier series, we demonstrated that the oscillations along $\gamma$ are heavily dependent on the specific problem instance. We derived the maximum frequency of these oscillations for Ising models both with and without external fields and employed the Nyquist-Shannon sampling theorem to determine the minimum sampling resolution necessary for accurately reconstructing the landscape.

To overcome the inefficiency of traditional grid searches over $(\gamma, \beta)$, we streamlined the two-dimensional optimisation to a one-dimensional line search over $\gamma$ and analytically determined the optimal $\beta^*$, thereby simplifying the search process. We then introduced Green's subdivision algorithm~\cite{green1999calculating} for computing the maximum modulus of univariate trigonometric polynomials, which aids in estimating the optimal $\gamma^*$ angle. Furthermore, we proved that in the infinite-size limit for Ising models with fields on regular graphs, the global optimum for $\gamma \in \mathbb{R}^+$ tends to concentrate near zero and coincides with the first local optimum. Supporting this theoretical result, our empirical evidence demonstrates that the global optimum typically aligns with the first local optimum along $\gamma \in \mathbb{R}^+$ even for finite-size instances.  Leveraging this insight, we refined the optimisation process further by initiating gradient descent near zero, thereby eliminating the need for a comprehensive line search over $\gamma$.

We validated our parameter tuning strategies by applying them to RQAOA$_1$ and benchmarking its performance on large, dense, and weighted problem instances—scenarios where naive grid searches and traditional heuristics often falter. Our benchmarks encompassed Erdős-Rényi graphs with 128 and 256 vertices, varying densities, and weight configurations. For Ising models without external fields, RQAOA$_1$ optimised using a naive coarse grid search followed by local minimisation did not surpass the performance of the SDP. However, when optimised using our proposed parameter tuning strategy, RQAOA$_1$ consistently outperformed both the coarse search approach and the SDP across all tested graph instances. In contrast, for Ising models with external fields, RQAOA$_1$ exhibited inconsistent performance, the underlying causes of which remain unclear\footnote{This instability may not be unique to the present setting: recent results on the Binary Paint Shop Problem likewise reported that RQAOA$_1$, despite using provably optimal QAOA$_1$ parameters at each recursive step, showed diminishing performance as the problem size increased~\cite{vijendran2025classical}.}. To overcome this limitation, we developed Iter-QAOA, which demonstrated greater robustness and consistently outperformed the SDP for Ising models with fields. Although the diminished performance of RQAOA$_1$ in the presence of fields raises questions about its applicability to more general binary optimisation problems, addressing these concerns is beyond the scope of the current work and will be explored in future research.

In summary, our work presents near-optimal---if not optimal---parameter tuning strategies for QAOA$_1$ applicable to arbitrary Ising models, achievable in polynomial time. We have rigorously validated these strategies on large, dense, and heavily weighted problem instances, demonstrating their effectiveness and robustness. Unlike previously proposed heuristics, our approach does not rely on assumptions about the problem's structure, symmetry, or weight distribution. Moreover, it achieves near-optimal performance with minimal computational overhead. Consequently, we assert that our work effectively resolves the challenge of efficiently finding globally near-optimal QAOA$_1$ parameters for arbitrary QUBO problems.

\section{Methods}

\subsection{Recursive QAOA (RQAOA)} \label{RQAOA_Subsec}

Recursive QAOA (RQAOA) is a non-local variant of QAOA that uses shallow-depth QAOA as a subroutine to recursively reduce the problem size until it becomes trivial to solve via brute force. RQAOA$_1$, which specifically employs QAOA$_1$ as its subroutine, has demonstrated superior performance compared to standalone QAOA$_1$ and is competitive with state-of-the-art classical algorithms~\cite{bravyi2020obstacles, bravyi2022hybrid, bravyi2021classical}. Additionally, RQAOA$_1$ provides performance guarantees for specific graph structures~\cite{bravyi2020obstacles, bae2024recursive, bae2024improvedr} and has been explored in various contexts~\cite{moussa2022unsupervised, fischer2024role, thelen2024approximating, patel2024reinforcement, egger2021warm}.

\begin{algorithm}[!t]
\fontsize{10.25pt}{12pt}\selectfont
\SetAlgoLined
\SetArgSty{}
\SetKwInput{KwData}{Input}
\SetKwInput{KwResult}{Output}

\KwData{Problem Hamiltonian $C_0$ with $n$ variables $V_0 = \{1, \ldots, n\}$, recursion depth $\eta$}
\KwResult{Approximate ground state of $C_0$}

\For{$t \gets 0$ \KwTo $n - \eta$}{
    Prepare QAOA state $\left|\psi_t(\boldsymbol{\gamma}, \boldsymbol{\beta})\right\rangle$ for $C_t$\;
    
    Optimise parameters:
    \vspace{-0.1cm}
    \[
    (\boldsymbol{\gamma}^*, \boldsymbol{\beta}^*) \gets 
    \argmin_{\boldsymbol{\gamma}, \boldsymbol{\beta}}
    \langle\psi_t (\boldsymbol{\gamma}, \boldsymbol{\beta})| C_t | \psi_t (\boldsymbol{\gamma}, \boldsymbol{\beta})\rangle;
    \]
    \vspace{-0.25cm}
    
    Compute $M_i$ for $i \in V_t$:
    $$
    M_{i} \gets \langle\psi_t (\boldsymbol{\gamma}^*, \boldsymbol{\beta}^*)| Z_i | \psi_t (\boldsymbol{\gamma}^*, \boldsymbol{\beta}^*)\rangle;
    $$
    
    Identify the variable $u \gets \displaystyle \argmax_{i \in V_t} \left|M_i\right|$\;
    
    Compute $M_{ij}$ for all $\{i, j\} \in V_t$:
    $$
    M_{ij} \gets \langle\psi_t (\boldsymbol{\gamma}^*, \boldsymbol{\beta}^*)| Z_i Z_j | \psi_t (\boldsymbol{\gamma}^*, \boldsymbol{\beta}^*)\rangle;
    $$
    
    Identify the pair $\{u, v\} \gets \displaystyle\argmax_{\{i,j\} \in V_t}
    \left|M_{ij}\right|$\;
    
    \eIf{$\left|M_u\right| \geq \left|M_{uv}\right|$}{
        Assign Value $u: Z_{u} = \operatorname{sign}(M_u)$\;
        
        Update Variable Set: $V_{t+1} \gets V_t \setminus \{u\}$\;

        Define $h'_i \gets h_i + \operatorname{sign}(M_{u}) J_{u i}$\;

        Update Hamiltonian:
        \vspace{-0.15cm}
        \begin{equation*}
           \hspace{-0.4cm}  C_{t+1} \gets \hspace{-0.5cm} \sum_{\{i, j\} \in V_{t+1}} \hspace{-0.3cm} J_{ij} Z_i Z_j + \hspace{-0.2cm} \sum_{i \in V_{t+1}} \hspace{-0.2cm}  h'_i Z_i + \operatorname{sign}(M_{u}) h_{u};
        \end{equation*}
        \vspace{-0.3cm}
    }{
        Impose Constraint: 
        $$
        Z_{u} = \operatorname{sign}(M_{uv}) Z_{v};
        $$
        
        Update Variable Set: $V_{t+1} \gets V_t \setminus \{u\}$\;

        Define $J_{u j}^{\prime} \gets J_{u j} + \operatorname{sign}(M_{u v}) J_{v j}$\;

        Define $h_{u}^{\prime \prime} \gets h_{u} + \operatorname{sign}(M_{u v}) h_{v}$\;

        Update Hamiltonian:
        \vspace{-0.2cm}
        \begin{align*}
            C_{t+1} &\gets \hspace{-0.3cm} \sum_{\{i, j\} \in V_{t+1}} \hspace{-0.3cm} J_{i j} Z_i Z_j+\! \sum_{j \in V_{t+1}} \! J_{u j}^{\prime} Z_{u} Z_j\\
            &  + \hspace{-0.65em} \!\sum_{i \in V_{t+1}} \! h_i Z_i+h_{u}^{\prime \prime} Z_{u}+ \operatorname{sign}(M_{u v}) J_{u v};
        \end{align*}
        \vspace{-0.5cm}
    }
}

Solve $C_t$ via brute-force when $|V_t| \leq \eta$\;

Solve $C_0$ by backtracking\;

\Return{Approximate ground state of $C_0$}

\caption{RQAOA for Ising Models}
\label{rqaoa_procedure}
\end{algorithm}

At the core of RQAOA is correlation rounding, a technique that identifies the qubit or pair of qubits with the largest absolute expectation or correlation to impose constraints (see algorithm \ref{rqaoa_procedure}). This approach shares similarities with iterative rounding~\cite{williamson2011design, lau2011iterative}, a method commonly used in integer linear programming (ILP). In iterative rounding, the linear relaxation of an ILP is solved, fractional variables are rounded, and the problem is updated iteratively. Typically, the variable with the largest absolute fractional value is selected, rounded to $\pm 1$, and the reduced problem is solved.

In contrast, correlation rounding shifts attention from single variables to the correlations between variable pairs, which are derived from the QAOA state. At each iteration, the algorithm identifies the variable $u^*$ with the largest absolute single-qubit expectation $\left|M_{u^*}\right|$ and the pair $\left\{u^*, v^*\right\}$ with the largest absolute pairwise correlation $\left|M_{u^* v^*}\right|$. If the single-variable expectation $\left|M_{u^*}\right|$ is greater, $u^*$ is fixed to $\operatorname{sign}\left(M_{u^*}\right)$. Otherwise, a correlation constraint is imposed: $v^*$ is set relative to $u^*$ based on the sign of $M_{u^* v^*}$. This process ensures that decisions are guided by the strongest available information, whether from individual variables or their correlations. Unlike iterative rounding, where variables are deterministically set immediately, correlation rounding delays assigning specific values to some variables until the reduced problem is small enough to solve via brute force. By focusing on pairwise interactions, it captures the problem's structure more effectively, especially in cases where interactions play a dominant role. This distinction is crucial for problems with strong quadratic terms, where correlations between variables carry more information than individual biases.

\subsection{Iterative QAOA (Iter-QAOA)} \label{IterQAOA_Subsec}

The key distinction between Iter-QAOA and RQAOA lies in their treatment of linear and quadratic terms within the optimisation process. As outlined in algorithm~\ref{rqaoa_procedure}, RQAOA evaluates the expectation values of both linear terms ($\langle Z_u \rangle$) and quadratic terms ($\langle Z_u Z_v \rangle$) at each recursive step for every node $u$ and edge $\{u, v\}$. It selects the term with the largest absolute value and either assigns $Z_u$ a value of $\pm 1$ or enforces a correlation or anti-correlation constraint between $Z_u$ and $Z_v$. In contrast, Iter-QAOA streamlines this process by utilising only the linear terms. It assigns values based solely on the signs of these linear terms, thereby eliminating the need to impose additional constraints.

At first glance, Iter-QAOA appears counterintuitive and should, in theory, underperform compared to RQAOA since it relies on partial information for decision-making. For instance, consider a scenario involving a node $u$ and an edge $\{u, v\}$ where the magnitudes of the linear term $\langle Z_u \rangle$ and the quadratic term $\langle Z_u Z_v \rangle$ are the largest, with the condition $0 < h_u < J_{uv}$. Here, the field strength at node $u$ is weaker than the coupling between nodes $u$ and $v$. Ideally, the optimal angle should satisfy $|\langle Z_u \rangle| < |\langle Z_u Z_v \rangle|$. In such cases, RQAOA would impose the anti-correlation constraint $Z_u = -Z_v$, thereby reducing the energy by $-J_{uv}$. Conversely, Iter-QAOA, relying solely on the linear term, would assign $Z_u = -1$, achieving an energy reduction of $-h_u$ and missing the more substantial energy decrease from the coupling term. 

However, our benchmark results indicate that Iter-QAOA$_1$ consistently outperforms RQAOA$_1$ across most instances where most nodes $u$ and edges $\{u, v\}$ satisfy $h_u < J_{uv}$. In our tests, edge weights were sampled from a Gaussian distribution with a mean of 50 and a variance of 30, while node weights were sampled from a Gaussian distribution with a mean of 40 and a variance of 20, both rounded to the nearest integer. It is important to note that these problem instances were generated \emph{a priori} and not specifically tailored for this benchmark. However, when $h_u > J_{uv}$, both RQAOA and Iter-QAOA would correctly prioritise the linear term, resulting in appropriate rounding decisions.

Finally, we note that Iter-QAOA closely resembles the iterative algorithms described in~\cite{finvzgar2024quantum, brady2023iterative}, which are themselves extensions of RQAOA. We do not claim novelty for Iter-QAOA over these prior works; rather, it was developed to address performance issues and to validate our optimal parameter-setting strategy.

\section*{Code and Data Availability}

The code used to generate the numerical results, together with the datasets, is publicly available in the accompanying \href{https://github.com/vijeycreative/NOpt_QAOA1_Tuning}{GitHub repository}.

\section*{Acknowledgements}

V.V.\ is thankful to Ahad N Zehmakan and Joshua Dai for the stimulating discussions and insightful comments. This research is supported by A*STAR Start Up Fund (KIMR220701aIMRSEF) and Q.InC Strategic Research and Translational Thrust. D.E.K.\ is supported by the National Research Foundation, Singapore, and the Agency for Science, Technology and Research (A*STAR), Singapore, under its Quantum Engineering Programme (NRF2021-QEP2-02-P03); A*STAR C230917003; and A*STAR under the Central Research Fund (CRF) Award for Use-Inspired Basic Research (UIBR) and the Quantum Innovation Centre (Q.InC) Strategic Research and Translational Thrust. E.B.\ is supported by the National Research Foundation of Korea (NRF) of Korea grant funded by the Ministry of Science and ICT (MSIT) (Grant No.\ NRF-2022R1C1C2006396). H.K.\ is supported by KIAS individual grant (No.\ CG085302) at the Korea Institute for Advanced Study.
E.B.\ and H.K.\ acknowledge support from the NRF of Korea (Grants No. 2023M3K5A109480511 and No. 2023M3K5A1094813, which have since been consolidated under Grant No. RS-2023-NR119931). V.V.\ is supported by the National Research Foundation, Singapore, through the National Quantum Office, hosted in A*STAR, under its National Quantum Scholarship Scheme (NQSS) Funding Initiative (PhD).

\section*{Author Contributions}

All authors conceived the project. V.V., D.E.K., and S.M.A.\ contributed to the theoretical results. V.V.\ contributed to the numerical results. V.V.\ wrote the manuscript with input from all authors. All authors contributed to discussions regarding the results in this paper. D.E.K., P.K.L., and S.M.A. supervised the project.

\section*{Competing Interests}
All authors declare no financial or non-financial competing interests.

\bibliographystyle{quantum}
\bibliography{refs}

\clearpage
\onecolumn
\appendix

\begin{center}
{\Large{\sc{Supplementary Material}}} 
\end{center}

\section{Fourier Insights into the Expressivity and Landscape Limitations of \texorpdfstring{QAOA$_1$}{QAOA1}} \label{qaoa_fourier_app}

In \cref{max_freq_subsub}, we derived analytical expressions to compute the maximum frequencies achievable by QAOA$_1$ for arbitrary Ising models. These expressions were utilised in \cref{char_ener_subsub} to characterise the energy landscape and estimate the minimum sampling resolution required for accurately reconstructing the optimisation landscape. In this section, we compare the maximum frequencies of the Hamiltonian $H_P$ with those accessible via QAOA$_1$ for two sets of field-free Ising models: (i) 4-regular graphs with node counts ranging from 5 to 20 and (ii) $D$-regular graphs with 20 nodes, where $D$ varies from 2 to 19. For each case, 100 random graphs were generated with edge weights drawn from a Gaussian distribution (mean 50, variance 25). To facilitate this comparison, we employed specific computational methods to determine the maximum frequencies of both QAOA$_1$ and $H_P$.

The maximum frequency of QAOA$_1$ was determined analytically using \cref{qaoa_2_max_freq_col}. In contrast, the maximum frequency of $H_P$ was obtained by enumerating all possible solutions and calculating the difference between the highest and lowest costs. This is because constructing the matrix representation of $H_P$ for $n$ variables (qubits) becomes impractical for large $n$ due to its exponential scaling as $2^n \times 2^n$. Specifically, for $n = 20$, we leverage the diagonal structure of $H_P$ in the computational basis, where the diagonal entries correspond to the costs of all possible configurations in $\{0,1\}^n$. Consequently, we enumerate all configurations in $\{0,1\}^{20}$, compute their respective costs, and determine the maximum frequency as the difference between the minimum and maximum costs.

\begin{figure}[hptb]
    \centering
    \subfloat[Max Frequency vs. \# of Nodes]{
        \includegraphics[height=0.285\textheight]{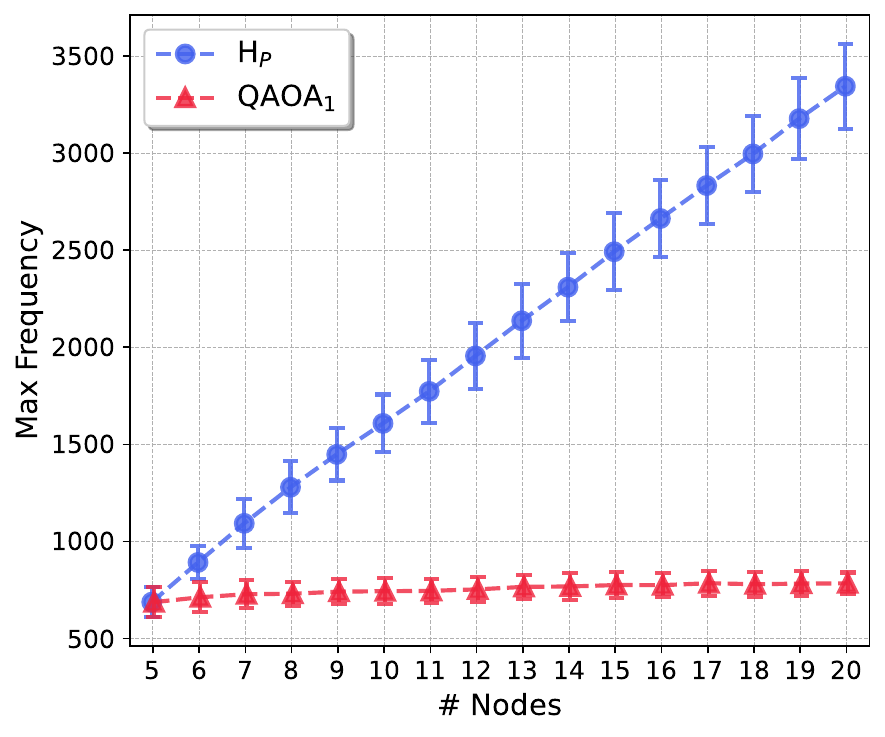}
        \label{max_freq_comp_fig_plot_a}
    }
    \subfloat[Max Frequency vs. Degree ($D$)]{
        \includegraphics[height=0.285\textheight]{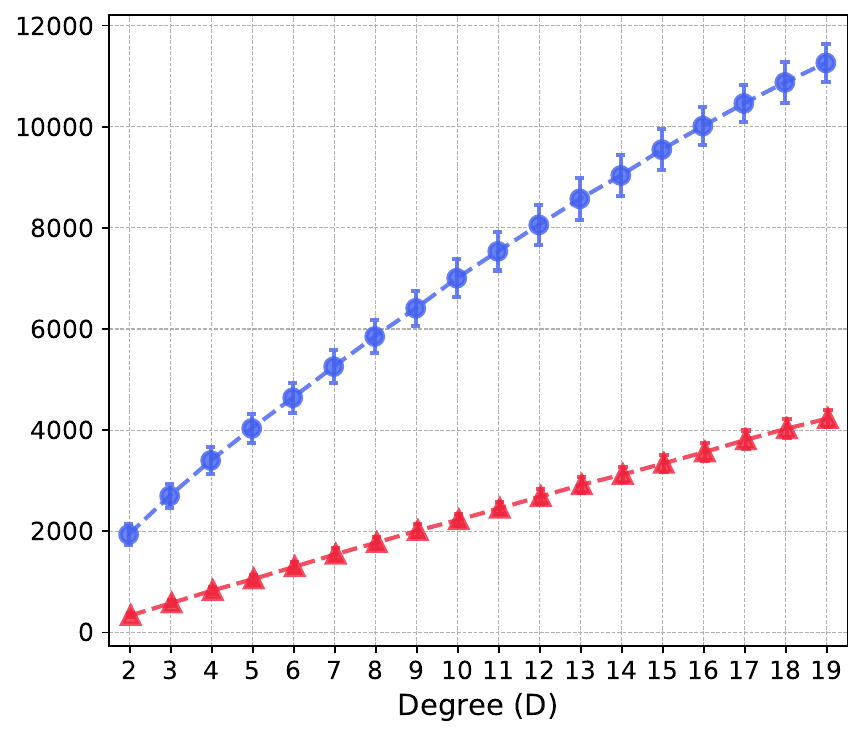}
        \label{max_freq_comp_fig_plot_b}
    }
    \caption[Maximum Frequency Comparison between Problem Hamiltonian and QAOA$_1$]{\textbf{Maximum Frequency Comparison between Problem Hamiltonian and QAOA$_1$.} \justifying The plots show the average maximum frequencies for 100 random graphs. \textbf{(a)} Maximum frequencies for 4-regular graphs with varying node counts from 5 to 20. \textbf{(b)} Maximum frequencies for 20-node $D$-regular graphs with degrees $D$ ranging from 2 to 19. Error bars represent the standard deviation across the random instances.}
    \label{max_freq_comp_fig}
\end{figure}

As shown in \cref{max_freq_comp_fig}, the maximum frequency attainable by QAOA$_1$ is significantly lower than that of $H_P$, primarily due to the locality constraints inherent in QAOA. In a $p$-level implementation, each qubit interacts only with others within $p$ edges; for $p = 1$, this means interactions are limited to immediate neighbours. This limitation is evident in \cref{max_freq_comp_fig_plot_a}, where the maximum reachable frequency for 4-regular graphs remains nearly constant as the number of nodes increases. In smaller graphs, qubits can effectively access the entire graph, but in larger graphs with a fixed degree $d = 4$, qubits become increasingly isolated, restricting their influence. Conversely, \cref{max_freq_comp_fig_plot_b} demonstrates that for 20-node $D$-regular graphs, the maximum frequency increases linearly with $D$. As the degree grows, each qubit interacts with more neighbours, increasing the range of local information incorporated into the QAOA$_1$ expectation value and thereby yielding higher maximum frequencies. 

Despite these improvements, even in complete graphs ($d = 19$), where every qubit is fully connected to every other qubit, a gap persists between the maximum frequencies achievable by QAOA$_1$ and those of $H_P$. This gap arises not only from the locality constraints but also from fundamental limitations in QAOA's Fourier series representation at level 1, where the finite number of terms restricts the circuit's ability to capture complex, high-frequency features of the energy landscape. Additionally, concepts such as the No Low Trivial Energy States (NLTS)~\cite{freedman2013quantum} provide further insights into this persistent frequency gap. It is important to note that the performance limitations of QAOA have been extensively explored in prior studies, which have focused on factors like locality constraints, the Overlap Gap Property (OGP)~\cite{farhi2020quantum1, farhi2020quantum2, basso2022performance, anshu2023concentration, barak2021classical, goh2024overlap}, reachability deficits related to constraint-to-variable ratios~\cite{akshay2020reachability, akshay2021reachability, chen2023local}, and symmetry protection~\cite{bravyi2020obstacles, chou2021limitations, marwaha2022bounds}. The novelty of this analysis, as illustrated in \cref{max_freq_comp_fig}, lies in interpreting QAOA’s expressivity and optimisation-landscape limitations through the lens of Fourier theory.

\section{Some Useful Lemmas}

Before deriving the analytical formulas, we first state and prove some key identities that will be used in \cref{qaoa_2_field_max_freq_proof}. In particular, for functions that involve a product of cosine terms, each with its own fundamental frequency, the following lemma provides a method to compute the maximum frequency of the entire function.
\begin{lemma}
    Let $h(\mathbf{w})$ be a product of cosine functions defined by
    \begin{equation}
    h(\mathbf{w}) = \prod_{i = 1}^n \cos (w_i \theta),
    \end{equation}
    where $\mathbf{w}=\left(w_1, \ldots, w_n\right)$ is a vector of real coefficients. The maximum angular frequency of $h(\mathbf{w})$ is given by
    \begin{equation}
    \omega_{\max} \left[ h(\mathbf{w}) \right] = \sum_{i = 1}^{n} |w_i|.
    \end{equation}
    \label{cos_freq_lemma}
\end{lemma}
\begin{proof}
We start by expressing the product of cosine functions $h(\mathbf{w})=\prod_{i=1}^n \cos \left(w_i \theta\right)$ using the product-to-sum identity for cosine functions. The product-to-sum identity for $n$ cosine terms is given by:
\begin{equation}
    \prod_{k=1}^n \cos \theta_k=\frac{1}{2^n} \sum_{e \in\{-1,1\}^n} \cos \left(\sum_{i=1}^n e_i \theta_i\right),
\end{equation}
where $\mathbf{e}=\left(e_1, \ldots, e_n\right)$ is an $n$-tuple with each $e_i$ taking values in $\{-1,1\}$.
Applying this identity to $h(\mathbf{w})$, we have:
\begin{equation}
    h(\mathbf{w})=\prod_{i=1}^n \cos \left(w_i \theta\right)=\frac{1}{2^n} \sum_{e \in\{-1,1\}^n} \cos \left(\sum_{i=1}^n e_i w_i \theta\right) .
\end{equation}
The maximum angular frequency component of $h(\mathbf{w})$ corresponds to the term in the sum where the argument of the cosine function is maximised. This maximum occurs when each $e_i$ is chosen to match the sign of the corresponding $w_i$, i.e., $e_i=\operatorname{sgn}\left(w_i\right)$. Therefore, the term with the maximum angular frequency is:
\begin{equation}
    \cos \left(\sum_{i=1}^n \operatorname{sgn}\left(w_i\right) w_i \theta\right)=\cos \left(\sum_{i=1}^n\left|w_i\right| \theta\right) .
\end{equation}
Thus, the maximum angular frequency of $h(\mathbf{w})$ is given by $\sum_{i=1}^n\left|w_i\right|$.
\end{proof}
The following lemma extends lemma \ref{cos_freq_lemma} to compute the maximum angular frequency of the product of a single sine function and $n$ cosine functions, each with its own fundamental angular frequency.
\begin{lemma}
    Let $g(\mathbf{w})$ be a product of one sine function and $n$ cosine functions defined by
    \begin{equation}
    g(\mathbf{w}) = \sin(w_0 \theta)\prod_{i = 1}^n \cos (w_i \theta),
    \end{equation}
    where $\mathbf{w}=\left(w_0,w_1, \ldots, w_n\right)$ is a vector of real coefficients. The maximum angular frequency of $g(\mathbf{w})$ is given by
    \begin{equation}
    \omega_{\max} \left[ g(\mathbf{w}) \right] = |w_0| + \sum_{i = 1}^{n} |w_i|.
    \end{equation}
    \label{sincos_freq_lemma}
\end{lemma}
\begin{proof}
    From lemma \ref{cos_freq_lemma}, we know that the maximum angular frequency of the product $\prod_{i=1}^n \cos \left(w_i \theta\right)$ is $\sum_{i=1}^n\left|w_i\right|$, corresponding to the term $\cos \left(\left(\sum_{i=1}^n\left|w_i\right|\right) \theta\right)$.

    Next, consider the product of this cosine term with $\sin \left(w_0 \theta\right)$. To analyse the frequency content of this product, we use the product-to-sum identity for sine and cosine functions, which states:
    \begin{equation}
        \sin \theta \cos \varphi=\frac{1}{2}[\sin (\theta+\varphi)+\sin (\theta-\varphi)] .
    \end{equation}
    Applying this identity with $\theta=w_0 \theta$ and $\varphi=\sum_{i=1}^n\left|w_i\right| \theta$, we obtain:
    \begin{equation}
        \sin(w_0 \theta) \cos \! \left(\! \left(\! \sum_{i = 1}^n |w_i|\! \right)\! \theta\! \right)\! = \!\frac{1}{2} \! \left[\! \sin \! \left(\! \left( \! w_0 + \sum_{i = 1}^n |w_i| \! \right) \! \theta \!\right)\! + \sin \! \left( \! \left( \! w_0 - \sum_{i = 1}^n |w_i| \! \right) \! \theta \! \right) \! \right].
    \end{equation}
    From the above expression, we observe that the maximum angular frequency of the product $g(\mathbf{w})=\sin \left(w_0 \theta\right) \prod_{i=1}^n \cos \left(w_i \theta\right)$ depends on the sign of $w_0$. If $w_0>0$, the first sine term on the RHS dominates with the maximum angular frequency, while for $w_0<0$, the second sine term attains this maximum angular frequency. In both cases, the maximum angular frequency of $g(\mathbf{w})$ is given by $\left|w_0\right|+$ $\sum_{i=1}^n\left|w_i\right|$.
\end{proof}

\section{Maximum Angular Frequency of Level-1 QAOA for Ising Models} \label{qaoa_2_field_max_freq_proof}

We now present the proof of \cref{qaoa_2_field_max_freq_thm}, which follows directly from the application of lemmas \ref{cos_freq_lemma} and \ref{sincos_freq_lemma} to the analytical expressions in \cref{qaoa_ozaeta_thm}. This proof demonstrates how to compute the maximum angular frequency exhibited by QAOA$_1$ for an arbitrary Ising model with external fields.
\begin{proof}[Proof of \Cref{qaoa_2_field_max_freq_thm}]
For the single spin variable $i$, the terms involving $\gamma$ are given by \cref{qaoa_ozaeta_thm_eqn1}. From lemma~\ref{sincos_freq_lemma}, the maximum frequency is:
\begin{equation}
\omega_{\max} \left[ \left\langle C_i \right\rangle_\gamma \right] = 2 \left(|h_i| + \sum\limits_{k \in \mathcal{N}(i)} |J_{ik}| \right).
\end{equation}
For the two spin variables $u$ and $v$, the terms involving $\gamma$ are:
\begin{equation}
\left\langle C_{uv} \right\rangle_\gamma = \frac{J_{uv}}{2} \sin(4\beta) f_{uv}(\mathbf{J}, \mathbf{h}, \gamma) - \frac{J_{uv}}{2} \sin^2(2\beta) g_{uv}(\mathbf{J}, \mathbf{h}, \gamma),
\end{equation}
where
\begin{align}
f_{uv}(\mathbf{J}, \mathbf{h}, \gamma) &= \sin (2 J_{uv} \gamma)\Biggl( \underbrace{\cos (2 h_v \gamma) \prod_{w \in e} \cos (2 J_{wv} \gamma)}_{\circled{1}} + \underbrace{\cos (2 h_u \gamma) \prod_{w \in d} \cos (2 J_{uw} \gamma)}_{\circled{2}} \Biggr) , \\
g_{uv}(\mathbf{J}, \mathbf{h}, \gamma) &=  \underbrace{\prod_{\substack{w \in e \\ w \notin F}} \cos (2 J_{wv} \gamma) \prod_{\substack{w \in d \\ w \notin F}} \cos (2 J_{uw} \gamma)}_{\circled{3}} \Biggl( \underbrace{\cos \left(2\gamma(h_u + h_v)\right) \prod_{f \in F} \cos \left(2\gamma (J_{uf} + J_{vf})\right)}_{\circled{4}}  \nonumber \\
& \quad  - \underbrace{\cos \left(2\gamma (h_u - h_v)\right) \prod_{f \in F} \cos \left(2\gamma (J_{uf} - J_{vf})\right)}_{\circled{5}} \Biggr).
\end{align}
Using lemma~\ref{sincos_freq_lemma}, we have:
\begin{align}
    \omega_{\max} \left[\sin (2 J_{uv} \gamma) \times \circled{1} \right] &= 2 \left(|J_{uv}| + |h_v| + \sum_{w \in e} |J_{wv}| \right), \\
    \omega_{\max} \left[\sin (2 J_{uv} \gamma) \times \circled{2} \right] &= 2 \left(|J_{uv}| + |h_u| + \sum_{w \in d} |J_{uw}| \right).
\end{align}
Using lemma~\ref{cos_freq_lemma}, we have:
\begin{align}
    \omega_{\max} \left[ \circled{3} \times \circled{4} \right] &= 2 \left( \sum_{\substack{w \in e \\ w \notin F}} |J_{wv}| + \sum_{\substack{w \in d \\ w \notin F}} |J_{uw}| + |h_u + h_v| + \sum_{f \in F} |J_{uf} + J_{vf}| \right),\\
    \omega_{\max} \left[ \circled{3} \times \circled{5} \right] &= 2 \left( \sum_{\substack{w \in e \\ w \notin F}} |J_{wv}| + \sum_{\substack{w \in d \\ w \notin F}} |J_{uw}| + |h_u - h_v| + \sum_{f \in F} |J_{uf} - J_{vf}| \right).
\end{align}
Thus, the maximum frequencies of $f_{uv}(\mathbf{h}, \gamma)$ and $g_{uv}(\mathbf{h}, \gamma)$ are:
\begin{align}
     \omega_{\max} \left[ f_{uv}(\mathbf{J}, \mathbf{h}, \gamma) \right] &= 2 \left( |J_{uv}| + \max \left\{ |h_v| + \sum_{w \in e} |J_{wv}| , |h_u| + \sum_{w \in d} |J_{uw}| \right\} \right), \\
     \omega_{\max} \left[ g_{uv}(\mathbf{J}, \mathbf{h}, \gamma) \right]\! &= \! 2 \! \left(\! \sum_{\substack{w \in e \\ w \notin F}} \! |J_{wv}| \!+\! \sum_{\substack{w \in d \\ w \notin F}} \! |J_{uw}| \! + \! \max \! \left\{ \! |h_u \! \pm \! h_v| \! + \! \sum_{f \in F} \! |J_{uf}  \! \pm \! J_{vf}| \!\right\}\! \right).
\end{align}
Therefore, the maximum angular frequency of $\left\langle C_{uv} \right\rangle_\gamma$ is the maximum of these two frequencies.
\end{proof}

\section{Upper Bounds on the Sampling Periods}

In this section, we establish upper bounds on the sampling period necessary to reconstruct the optimisation landscape of QAOA$_1$ for the SK model, $D$-regular triangle-free graphs, and triangle-free graphs with bounded maximum degrees. These bounds are proven by integrating \cref{qaoa_min_samples_thm} and \cref{qaoa_2_field_max_freq_thm} and leveraging the specific properties of each case.
\subsection{Sherrington-Kirkpatrick Model} \label{sk_narrow_gorge_proof}

\begin{proof}[Proof of \Cref{sk_gregion_cor}]
    Consider the case without external fields. The maximum angular frequency of the $\gamma$ terms in the expectation value is
    \begin{equation}
        \omega_{\max} \! \left[\! \left\langle C_{uv}  \right\rangle_\gamma \right] \! = \! 2 \times \max \! \left\{\! |J_{uv}| \! + \! \max \! \left\{ \! \sum_{w \in e}\! |J_{wv}| ,\! \sum_{w \in d}\! |J_{uw}|\!  \right\}\! ,\!  \max\! \left\{\! \sum_{f \in F}\! |J_{uf}\! \pm \! J_{vf}| \!\right\} \! \right\}.
    \end{equation}
    Since the graph is complete, the number of neighbours excluding vertices $u$ and $v$ is $|e| = |d| = n-2$, and the number of triangles is $|F| = n-2$. Therefore, we have
    \begin{align}
        \max \left\{ \sum_{w \in e} |J_{wv}| , \sum_{w \in d} |J_{uw}|  \right\} &= n - 2, \\
        \max \left\{ \sum_{f \in F} |J_{uf} \pm J_{vf}|\right\} & = 2n -4 \label{sk_tri_max}.
    \end{align}
    The maximum in the summation of the triangle terms in \cref{sk_tri_max} is achieved when all edges forming a triangle with edge $\{u, v\}$ have the same weights. Specifically, this occurs when $\operatorname{sgn}(J_{uf}) = -\operatorname{sgn}(J_{vf})$ for all $f \in F$. Under this condition, the maximum angular frequency of the $\gamma$ terms scales as $\mathcal{O}(n)$.

    This scaling remains valid for models with external fields of strength $\pm 1$, where the maximum angular frequency may differ by a constant factor compared to the case without external fields. Consequently, the maximum frequency is given by 
    \begin{equation} \nu_{\max} = \frac{\mathcal{O}(n)}{2\pi}. 
    \end{equation} 
    Since all weights are $\pm1$ and each $\gamma$ term is prefactored by 2, the period of the optimisation function is $T_{\gamma} = \pi$. Utilising the maximum frequency and the period of the function, sampling period is determined from \cref{qaoa_min_samples_thm}.
\end{proof}

\subsection{Triangle-Free Graphs} \label{unw_trif_dreg_proof}

\begin{proof}[Proof of \Cref{d_reg_tri_free_conc}]
    For a triangle-free graph, the maximum angular frequency of the $\gamma$ terms in the expectation value is given by:
    \begin{equation}
        \omega_{\max} \left[ \left\langle C_{uv} \right\rangle_\gamma \right] = 2  |J_{uv}| + 2\times \max \left\{ \sum_{w \in e} |J_{wv}| , \sum_{w \in d} |J_{uw}|  \right\} .
        \label{triangle_free_max_freq}
    \end{equation}
    Since the graph is $D$-regular, each vertex has $D-1$ neighbours excluding $u$ and $v$, so $|e| = |d| = D-1$. Consequently,
    \begin{equation}
        \max \left\{ \sum_{w \in e} |J_{wv}| , \sum_{w \in d} |J_{uw}|  \right\} = D-1,
    \end{equation}
    which gives the maximum angular frequency as 
    \begin{equation}
       \omega_{\max} \left[ \left\langle C_{uv} \right\rangle_\gamma \right] = 2 + 2(D -1) = 2D.
    \end{equation}
    The corresponding maximum frequency is $\omega_{\max}/2\pi = D /\pi$. Since each $\gamma$ term is prefactored by 2 and the weights are $\pm 1$, the optimisation function has a period of $T_{\gamma} = \pi$. Using this period and the maximum frequency, the sampling period is derived from \cref{qaoa_min_samples_thm}.

For triangle-free graphs with varying degrees, the maximum angular frequency is determined by the edge ${u, v}$ where at least one vertex has the maximum degree $D_{\max}$.
\end{proof}

\section{Univariate Representation of Level-1 QAOA's Objective Function}

In this section, we present the proofs of the linearised closed-form expressions for computing the expectation value of QAOA$_1$, specifically for Ising models with and without external fields, as described in \cref{qaoa_2_local_opt_params_thm} and \cref{qaoa_2_field_max_freq_thm}. These proofs rely on identities involving linear combinations of trigonometric functions and simple algebraic techniques.

\subsection{Ising Models without Fields} \label{qaoa_2_local_opt_params_proof}

\begin{proof}[Proof of \Cref{qaoa_2_local_opt_params_thm}]
    Using the definitions of the coefficients $A(\gamma)$ and $B(\gamma)$, the expectation value for an arbitrary Ising model without external fields can be rewritten as:
    \begin{align}
        \langle\gamma, \beta|H_P| \gamma, \beta\rangle &= A(\gamma) \sin 4\beta - B(\gamma) \sin^2 2\beta \\
        &= A(\gamma) \sin 4\beta - B(\gamma) \left(\frac{1 - \cos 4\beta}{2} \right) \\
        &= A(\gamma) \sin 4\beta + \frac{B(\gamma)}{2} \cos 4\beta - \frac{B(\gamma)}{2}.
    \end{align}
    Using the trigonometric identity for a linear combination of sine and cosine, specifically:
    \begin{equation}
        R \cos (x-\alpha)=A \sin x+B \cos x,
    \end{equation}
    the expression simplifies to:
    \begin{equation}
        \langle\gamma, \beta|H_P| \gamma, \beta\rangle=R(\gamma) \cos (4 \beta-\alpha(\gamma))-\frac{B(\gamma)}{2},
    \end{equation}
    where
    \begin{equation}
        R(\gamma)= \mathrm{sgn} \left[B(\gamma) \right]\sqrt{A^2(\gamma)+\frac{B^2(\gamma)}{4}}, \quad \alpha(\gamma)=\arctan \left(2 A(\gamma), B(\gamma)\right) .
    \end{equation}
    The minimum value of $\langle\gamma, \beta|H_P| \gamma, \beta\rangle$ is thus:
    \begin{equation}
    \underset{\gamma \in \mathbb{R}}{\min} \left(-\sqrt{A^2(\gamma)+\frac{B^2(\gamma)}{4}}-\frac{B(\gamma)}{2} \right),
    \end{equation}
    which is achieved by setting $\beta^*=\dfrac{\alpha(\gamma^*)+\pi}{4}$.
\end{proof}

\subsection{Ising Models with Fields} \label{qaoa_2_local_fields_opt_params_proof}

\begin{proof}[Proof of \Cref{qaoa_2_local_fields_opt_params_thm}]
Using the definitions of the coefficients $A(\gamma), B(\gamma)$, and $C(\gamma)$, the expectation value for an arbitrary 2-local Ising model with external fields can be rewritten as:
\begin{equation}
    \langle\gamma, \beta|H_P| \gamma, \beta\rangle = \sin (2 \beta) A(\gamma)+\sin (4 \beta) B(\gamma)+\sin ^2(2 \beta) C(\gamma). 
\end{equation}
To minimise this expression with respect to $\beta$ for a fixed $\gamma$, consider:
\begin{equation}
    \left\langle\beta|H_\gamma| \beta\right\rangle = A\sin2 \beta + B \sin4\beta + C \sin^2 2\beta.
\end{equation}
This function is bounded with endpoints $\left\langle 0\left|H_\gamma\right| 0\right\rangle=\left\langle\pi\left|H_\gamma\right| \pi\right\rangle=0$ and has a period of $\pi$. To find the global minimum in $(0, \pi)$, let $b=2 \beta$:
\begin{equation}
    \left\langle\beta|H_\gamma| \beta\right\rangle =  A\sin b + B \sin 2b + C \sin^2 b.
\end{equation}
Taking the derivative with respect to $\beta$:
\begin{align}
\frac{d \left\langle\beta|H_\gamma| \beta\right\rangle}{d \beta} &= 2 \left[ A \cos b + 2B \cos 2b + 2C \sin b \cos b \right] \\
&= 2 \left[ A \cos b + 2B (2 \cos^2 b - 1) + 2C \sin b \cos b \right] \\
&= 2 \left[ A \cos b + 4B \cos^2 b - 2B + 2C \sin b \cos b \right].
\end{align}
Setting $\dfrac{d \left\langle\beta|H_\gamma| \beta\right\rangle}{d \beta} = 0$ to find the minima:
\begin{align}
    0 &= A\cos b + 4B\cos^2 b - 2B + 2 C \sin b \cos b \\
    4C^2 \sin^2 b\cos^2 b &= (A\cos b + 4B\cos^2 b - 2B)^2\\
    4C^2(1-\cos^2 b )\cos^2 b &= (A\cos b + 4B\cos^2 b - 2B)^2.
\end{align}
Letting $x=\cos b$, we get:
\begin{equation}
4 C^2\left(1-x^2\right) x^2=\left(A x+4 B x^2-2 B\right)^2,
\end{equation}
which simplifies to the quartic equation:
\begin{equation}
    0 = p(x) = 4 B^2 - 4 A B x + (A^2 - 16 B^2 - 4 C^2) x^2 + 
 8 A B x^3 + (16 B^2 + 4 C^2) x^4.
 \label{eq:quartic_roots}
\end{equation}
Let $x_1, x_2, x_3, x_4$ be the roots of $p(x)$. The minimum of $\left\langle\beta\left|H_\gamma\right| \beta\right\rangle$ occurs at one of the $\beta$ values where $x=\cos (2 \beta)$ is a root of $p(x)$. Therefore:
\begin{equation}
        \beta_\gamma^*= \underset{\beta \in \mathcal{B}_\gamma}{\text{argmin} }\left\langle\beta|H_\gamma| \beta\right\rangle,
    \end{equation}
    where
    \begin{equation}
    \mathcal{B}_\gamma=\left\{\left.\beta= \pm \frac{1}{2} \arccos \left(x_i\right) \right\rvert\, i=1,2,3,4\right\}.
    \end{equation}
    Finally, the global optimal angles $\left(\gamma^*, \beta^*\right)$ are given by:
    \begin{equation}
    \left(\gamma^*, \beta^*\right)=\underset{\gamma \in \mathbb{R}}{\text{argmin} }\left\langle\gamma, \beta_\gamma^*|H_P| \gamma, \beta_\gamma^*\right\rangle.
    \end{equation}
\end{proof}

\section{Proof of \texorpdfstring{\Cref{edge_only_is_bad_cor}}{edge only is bad cor}} \label{edge_only_is_bad_proof}
We present the proof of \cref{edge_only_is_bad_cor}, demonstrating that transforming an Ising model with fields into one without them can actually increase the complexity of finding optimal parameters, particularly in bounded-degree graphs where the maximum degree is much smaller than the number of vertices.

\begin{proof}[Proof of \cref{edge_only_is_bad_cor}]
From \cref{qaoa_2_field_max_freq_thm}, the maximum angular frequency for an Ising model with external fields is given by: 
\begin{equation}
    \omega_{\max} \left[ \langle H_P \rangle_\gamma \right]\! = \! \max \left\{ \omega_{\max} \left[ \langle C_i \rangle_\gamma \right], \; \omega_{\max} \left[ \langle C_{uv} \rangle_\gamma \right] \! \Big| \forall i \in V, \; \forall \{u,v\} \in E \right\}.
\end{equation}
For a triangle-free graph with bounded maximum degree $D_{\max}$, it follows that: 
\begin{equation} 
\omega_{\max} \left[ \langle C_i \rangle_\gamma \right] = \omega_{\max} \left[ \langle C_{uv} \rangle_\gamma \right] = 2(D_{\max}+1). 
\end{equation}
Applying \cref{qaoa_min_samples_thm}, we derive the upper bound on the permissible spacing between consecutive samples: 
\begin{equation} 
\Delta_{\gamma} < \frac{\pi}{2(D_{\max}+1)}. 
\end{equation}
When converting an Ising model with external fields of size $n$ (i.e., $n$ spins or vertices) to an equivalent Ising model without external fields, the maximum degree of the resulting graph increases from $D_{\max}$ to $n$. According to \cref{qaoa_2_max_freq_col}, the maximum angular frequency for an Ising model without external fields can be computed using \cref{qaoa_2_max_freq_col_eq2}. Although the resulting graph may contain triangles, the maximum contribution to the angular frequency arises from the non-triangle terms. Therefore, the maximum angular frequency simplifies to: 
\begin{align} 
\omega_{\max} \left[ \langle C_{uv} \rangle_\gamma \right] &= 2 \left( |J_{uv}| + \max \left\{ \sum_{w \in e} |J_{wv}|, \sum_{w \in d} |J_{uw}| \right\} \right) \\
&= 2 \left( 1 + \sum_{w \in d} |J_{uw}| \right) \\
&= 2 \left( 1 + (n - 1) \right) \\
&= 2n. 
\end{align}
Applying \cref{qaoa_min_samples_thm} again, we obtain the bound on the maximum permissible spacing between consecutive samples for the transformed model: 
\begin{equation} 
\Delta_{\gamma}^{\prime} < \frac{\pi}{2n}. 
\end{equation}
Taking the ratio of $\Delta_{\gamma}$ to $\Delta_{\gamma}^{\prime}$, we have: 
\begin{equation} 
\frac{\Delta_{\gamma}}{\Delta_{\gamma}^{\prime}} < \frac{n}{D_{\max} + 1}. 
\end{equation}
Thus, when an Ising model with external fields, whose underlying graph is triangle-free with maximum degree $D_{\max}$, is converted to an equivalent model without external fields, the maximum permissible spacing between consecutive samples decreases by a factor of $\frac{n}{D_{\max} + 1}$.
\end{proof}

\section{Steckin's Lemma for the Subdivision Algorithm} \label{steck_sec_app}

In this section, we introduce Steckin's Lemma and demonstrate its application as a rejection criterion in the subdivision algorithm for estimating the maximum modulus of trigonometric polynomials and identifying their optima. We begin by presenting Steckin's Lemma for univariate real-valued trigonometric polynomials:
\begin{lemma}[Steckin's Lemma]\label{steck_uni_lemma}
    Suppose that the polynomial
    \begin{equation}
        f(t)=\sum_{|n| \leq N} \alpha_n e^{i n t} \quad\left(\alpha_n \in \mathbf{C}\right)
    \end{equation}
    is real-valued and that $t_0 \in [0, 2\pi]$ satisfies $f(t_0) = \|f\|_{\infty}$. Then 
    \begin{equation}
        f\left(t_0+s\right) \geq\|f\|_{\infty} \cos N s \quad(|s| \leq \pi / N).
    \end{equation}
\end{lemma}
The proof of \cref{steck_uni_lemma} is outlined in Exercise 1.8 of~\cite{Edwards1979}. Stecklin's Lemma offers valuable insight into the behaviour of real-valued trigonometric polynomials at their maximum points. Specifically, it states that if a polynomial $f(t)$ attains its maximum value at $t_0$, then within a neighbourhood around $t_0$ (where the distance $s$ satisfies $|s| \leq \pi / N$), the polynomial does not decrease below $\|f\|_{\infty} \cos (N s)$. This ensures that $f(t)$ retains a significant portion of its maximum value near $t_0$. Such a property is crucial in applications like signal processing and approximation theory, where understanding the local behaviour of trigonometric polynomials is essential for accurate analysis and optimisation.

\vspace{0.25cm}
\noindent Now, given an analytic polynomial $p$, we can define its maximum modulus on the unit disk as
\begin{equation}
\|p\|_{\infty}:=\sup \{|p(z)|: z \in \mathbb{C},|z| \leq 1\}
\end{equation}
Since $\|p\|_{\infty}$ is attained by $|p|$ on the boundary of the unit disc, our objective is to find bounds on
\begin{equation}
\sup \left\{\left|p\left(e^{i t}\right)\right|: t \in[0,2 \pi]\right\},
\end{equation}
and since our calculations are simplified if we work with the square of $|p|$, we write
\begin{equation}
    q(t):=\left|p\left(e^{i t}\right)\right|^2=p\left(e^{i t}\right) \overline{p\left(e^{i t}\right)}.
\end{equation}
Suppose $q$ has been evaluated at $t_k$. If $h<\pi / N$ and the global maximiser $t_0$ lies within the interval $\left[t_k-h, t_k+h\right]$, then
\begin{equation}
q\left(t_k\right) \geq\|q\|_{\infty} \cos N\left(t_k-t_0\right) \geq\|q\|_{\infty} \cos N h
\label{uni_steck_bound}
\end{equation}
as shown by \cref{steck_uni_lemma}. Consequently, if the global maximiser is known to lie within one of the intervals $I_1, \ldots, I_n$(each of size $h$), and $q$ has been evaluated at their midpoints $t_1, \ldots, t_n$, then there exists some $k$ such that $\|q\|_{\infty} \leq q\left(t_k\right) \sec N h$. Define $\tilde{q}:=\max \left\{q\left(t_i\right): i=1, \ldots, n\right\}$. Then, the following bounds hold:
\begin{equation}
\tilde{q} \leq\|q\|_{\infty} \leq \tilde{q} \sec N h.
\end{equation}
Moreover, if $q\left(t_i\right)<\tilde{q} \cos N h$, the global maximiser cannot lie within the interval $I_i$, allowing us to discard it. These observations lead to a straightforward algorithm for calculating $\|p\|_{\infty}$ and determining the location of the global optimum, as outlined in algorithm \ref{subdivision_algorithm}.

\noindent Chevrotiere~\cite{de2009finding} extended Green's~\cite{green1999calculating} work by generalising \cref{steck_uni_lemma} to multivariate polynomials and expanding the subdivision algorithm to estimate the maximum modulus of multivariate trigonometric polynomials on polydisks. The Multivariate Stecklin's Lemma is formalised as follows:
\begin{lemma}[Multivariate Steckin's Lemma]\label{steck_mult_lemma} Let $g(t)$ be a real-valued multivariate trigonometric polynomial defined by
\begin{equation}
    g(t)=\sum_{\substack{\left|n_j\right| \leq d_j \\ j=1, \ldots, k}} c_{n_1, \ldots, n_k} e^{i\left(n_1 t_1+\cdots+n_k t_k\right)},
\end{equation}
where $t=\left(t_1, \ldots, t_k\right) \in[0,2 \pi]^k$ and $c_{n_1, \ldots, n_k} \in \mathbb{C}$. Suppose there exists a point $t_0 \in[0,2 \pi]^k$ such that $g\left(t_0\right)=\|g\|_{\infty}$. Then, for any $s=\left(s_1, \ldots, s_k\right) \in \mathbb{R}^k$ satisfying $d_1\left|s_1\right|+\cdots+d_k\left|s_k\right| \leq \pi$, the polynomial $g$ satisfies
\begin{equation}
    g\left(t_0+s\right) \geq\|g\|_{\infty} \cos \left(d_1\left|s_1\right|+\cdots+d_k\left|s_k\right|\right).
\end{equation}
\end{lemma}
The proof of \cref{steck_mult_lemma} is provided in Chevrotiere~\cite{de2009finding}. Chevrotiere demonstrates that by using \cref{steck_mult_lemma} as a rejection criterion, Green's~\cite{green1999calculating} subdivision algorithm can be effectively extended to arbitrary dimensions.

\noindent Let us now address the optimisation of the univariate functions defined in \cref{qaoa_2_local_opt_params_thm} and \cref{qaoa_2_local_fields_opt_params_thm}. Suppose we optimise over both parameters $(\gamma, \beta)$ using the expressions from \cref{qaoa_ozaeta_thm} and \cref{qaoa_ozaeta_col}. In this scenario, \cref{steck_mult_lemma} provides the following bound:
\begin{equation}
    \left\langle C(\gamma^* + s_1, \beta^* + s_2) \right\rangle^2 \geq \left\langle C(\gamma^*, \beta^*) \right\rangle^2 \cos \left(\nu^{\gamma}_{max}|s_1| + \nu^{\beta}_{max}|s_2| \right),
\end{equation}
where $\nu_{\max }^\gamma$ and $\nu_{\max }^\beta$ are the maximum frequencies of the $\gamma$ and $\beta$ terms, respectively, and $\left|s_1\right|$ and $\left|s_2\right|$ denote the distances from the global maximum satisfying $\nu_{\max }^\gamma\left|s_1\right|+\nu_{\max }^\beta\left|s_2\right|<\pi$.

However, if we restrict our search to $\bigl(\gamma, \beta^*\bigr)$ as in \cref{qaoa_2_local_opt_params_thm}, then $\lvert s_2\rvert=0$, and the bound simplifies to the same form as in \cref{uni_steck_bound}. In the context of \cref{qaoa_2_local_fields_opt_params_thm}, we only know a locally optimal $\beta_\gamma^*$ for each $\gamma$; if $\beta_\gamma^* = \beta^*$ for any $\gamma \in [0,\pi]$, then $\lvert s_2\rvert=0$ again, and the bound remains unchanged. If instead $\lvert s_2\rvert \neq 0$ (i.e., $\beta_\gamma^* \neq \beta^*$), then using the fact that $\nu_{\max}^\gamma \lvert s_1 \rvert + \nu_{\max}^\beta \lvert s_2 \rvert < \pi$, it follows that
\begin{equation}
    \cos\bigl(\nu_{\max}^\gamma\lvert s_1 \rvert + \nu_{\max}^\beta\lvert s_2 \rvert \bigr) 
< 
\cos\bigl(\nu_{\max}^\gamma\lvert s_1 \rvert\bigr),
\end{equation}
which leads to the rejection of that particular interval. Therefore, in optimising the univariate function of $\gamma$, it suffices to treat both expressions in \cref{qaoa_2_local_opt_params_thm} and \cref{qaoa_2_local_fields_opt_params_thm} as univariate trigonometric polynomials. Consequently, we can apply Green’s~\cite{green1999calculating} subdivision algorithm to efficiently estimate the optimal value $\gamma^*$.

\begin{theorem} \label{subdiv_linear_thm}
    Let $p:[0,2 \pi] \rightarrow \mathbb{R}$ be a real trigonometric polynomial of degree $N$ such that $\|p\|_{\infty} = 1$. The subdivision algorithm requires $\mathcal{O}(N/\sqrt{\varepsilon})$ function evaluations to approximate $\operatorname{max} \, \, |p(x)|^2$ to within an additive tolerance $\varepsilon>0$.
\end{theorem}
\begin{proof}
    The algorithm begins by subdividing the interval $[0, 2\pi]$ into $M$ subintervals, each of length $\Delta_{\gamma} = \pi/M$, where $M > 2N$. Consequently, the initial number of subintervals satisfies $M = \mathcal{O}(N)$. Although the algorithm prunes many subintervals in each iteration---specifically those whose midpoint values are below a certain threshold---it is possible, in the worst case, to retain up to $\mathcal{O}(N)$ subintervals.

    With each subdivision step, the number of subintervals doubles. Therefore, after $i$ iterations, the number of subintervals can grow to $\mathcal{O}\left(N \cdot 2^i\right)$. Since the algorithm only evaluates the midpoint of each subinterval during subdivision, the total number of function evaluations after $i$ iterations remains bounded by $\mathcal{O}\left(N \cdot 2^i\right)$.
    
    The stopping criterion for Green's algorithm is that the quantity $\tilde{q} \left(\sec(N \Delta_{\gamma}) - 1\right)$ becomes sufficiently small. Here, $\tilde{q}$ represents the current estimate of the polynomial's maximum modulus, and $\Delta_{\gamma}$ is the current width of the subintervals. Given that $\|p\|_{\infty} = 1$, we approximate $\tilde{q} \approx 1$. For small $\Delta{\gamma}$, we can approximate:
    \begin{equation}
    \sec \left(N \Delta_\gamma\right)-1 \approx \frac{\left(N \Delta_\gamma\right)^2}{2}.
    \end{equation}
    To ensure that the algorithm's error is below a tolerance $\varepsilon$, we require:
    \begin{equation}
    \tilde{q}\left(\sec \left(N \Delta_\gamma\right)-1\right) \lesssim \varepsilon \quad \Rightarrow \quad\left(N \Delta_\gamma\right)^2 \lesssim \varepsilon \quad \Rightarrow \Delta_{\gamma} \lesssim \frac{\sqrt{\varepsilon}}{N}.
    \end{equation}
    Since each iteration halves $\Delta_{\gamma}$, starting from $\Delta_{\gamma}^{(0)} = \mathcal{O}\left(\frac{1}{N}\right)$, after $i$ halvings we have:
    \begin{equation}
        \Delta_{\gamma} \approx \frac{1}{N} \frac{1}{2^i}.
    \end{equation}
    To satisfy the error tolerance, we set:
    \begin{equation}
    \frac{1}{N \cdot 2^i} \lesssim \frac{\sqrt{\varepsilon}}{N} \quad \Rightarrow \quad 2^{-i} \lesssim \sqrt{\varepsilon} \quad \Rightarrow \quad 2^i \approx \frac{1}{\sqrt{\varepsilon}} .
    \end{equation}
    Substituting $2^i \approx \frac{1}{\sqrt{\varepsilon}}$ into the total number of function evaluations, we obtain $\mathcal{O}(N/\sqrt{\varepsilon})$.
\end{proof}

\section{Globally Optimal Parameters for Regular Graphs}

\subsection{Mathematical Preliminaries}
\subsubsection{Probability Density Functions}
In this subsection, we introduce the necessary definitions and notations related to probability density functions (PDFs) that are essential for our proof. We begin by defining the expectation operator for functions of random variables and extend these definitions to multiple variables and arbitrary index sets.

\vspace{0.25cm}
\noindent \underline{\textbf{Single-Variable Expectation}}\\
\newline
Let $f: \mathbb{R} \rightarrow \mathbb{R}$ be a probability density function (PDF). For a measurable function $g: \mathbb{R} \rightarrow \mathbb{R}$, the expectation of $g(x)$ with respect to $f$ is denoted and defined as
\begin{equation}
\underset{x \sim f}{\mathbb{E}}[g(x)]=\int_{\mathbb{R}} f(x) g(x) d x .  
\end{equation}
This integral represents the average value of $g(x)$ when $x$ is distributed according to the PDF $f$.

\vspace{0.25cm}
\noindent \underline{\textbf{Multi-Variable Expectation for $n$ Independent Variables}}\\
\newline
Consider $n$ independent random variables $x_1, x_2, \ldots, x_n$, each distributed according to the same PDF $f$. We denote the joint expectation over these $n$ variables as $\mathbb{E}_{x \sim f^n}$. Formally, for a function $g$ : $\mathbb{R}^n \rightarrow \mathbb{R}$, the expectation is defined by
\begin{align}
    \underset{x \sim f^n}{\mathbb{E}} \left[ g(x) \right] &= \underset{x_1 \sim f}{\mathbb{E}} \cdots \underset{x_n \sim f}{\mathbb{E}} \left[ g(x_1, \dots, x_m) \right] \\
    &= \int_{\mathbb{R}^n} f\left(x_1\right) f\left(x_2\right) \cdots f\left(x_n\right) g\left(x_1, \ldots, x_n\right) d x_1 d x_2 \cdots d x_n.
\end{align}
This expression computes the expectation of $g$ over the product measure $f^n$, reflecting the independence of each $x_i$.

\vspace{0.25cm}
\noindent \underline{\textbf{Expectation over an Arbitrary Index Set}}\\
\newline
Let $E$ be an arbitrary index set, and for each $e \in E$, let $x_e$ be a random variable distributed according to the PDF $f$. We define the expectation operator over the collection of variables indexed by $E$ as
\begin{equation}
        \underset{x \sim f^E}{\mathbb{E}}=\prod_{e \in E} \underset{x_e \sim f}{\mathbb{E}}.
\end{equation}
For a function $g: \mathbb{R}^E \rightarrow \mathbb{R}$, the expectation is given by
\begin{align}
        \underset{x \sim f^E}{\mathbb{E}} \left[ g(x) \right] &=  \prod_{e \in E} \underset{x_e \sim f}{\mathbb{E}} \left[g\left((e)_{e \in E}\right)\right] \\
        &= \mathbb{E}_{x \sim f^E}[g(x)]=\int_{\mathbb{R}^E} \prod_{e \in E} f\left(x_e\right) g\left(\left(x_e\right)_{e \in E}\right) \prod_{e \in E} d x_e.
\end{align}
This generalises the expectation to any index set $E$, allowing for flexibility in handling functions of numerous variables.

\subsubsection{Interchanging Summation and Integration in Expectations}
Throughout this section, we assume that the probability density functions $p$ are chosen such that the infinite sum arising from the Taylor series expansion and the integrals involved in taking expectations can be interchanged. This assumption is crucial for the validity of the subsequent manipulations.

\noindent Consider a function $f$ expressed as its Taylor series expansion:
\begin{equation}
    f(x)=\sum_{i=0}^{\infty} a_i x^i
\end{equation}
Under our assumption, we can compute the expectation of $f(x)$ with respect to the probability distribution $p$ by interchanging the summation and the integral. The steps are as follows:
\begin{align}
\underset{x \sim p}{\mathbb{E}} \left[f(x)\right] & =\underset{x \sim p}{\mathbb{E}} \left[\sum_{i=0}^{\infty} a_i x^i \right] \\
& =\int d x \, p(x) \sum_{i=0}^{\infty} a_i x^i \\
& =\underbrace{\int d x \sum_{i=0}^{\infty}}_{=\sum_{i=0}^{\infty} \int d x} p(x) a_i x^i \\
& = \sum_{i=0}^{\infty} a_i \int d x \, p(x) x^i \\
& =\sum_{i=0}^{\infty} a_i \underset{x \sim p}{\mathbb{E}} \left[x^i \right].
\end{align}

\subsection{Some Useful Lemmas}

\begin{lemma} \label{sin_taylor_lemma}
Let $x$ be a random variable distributed according to $p(x)$, and let $D$ be a large parameter. Then the expectation of $x \sin \left(\frac{a x}{\sqrt{D}}\right)$ can be approximated by:
\begin{equation}
    \begin{aligned}
        \mathbb{E}_{x \sim p} \left[ x \sin \left( \frac{a x}{\sqrt{D}} \right) \right] 
        &= \frac{a}{\sqrt{D}} \mathbb{E}\left[ h^2 \right] - \frac{a^3}{6 D^{3/2}} \mathbb{E}\left[ h^4 \right] + \mathcal{O}\left( D^{-\frac{5}{2}} \right), \\
        &= \frac{a}{\sqrt{D}} \mathbb{E}\left[ h^2 \right] + \mathcal{O}\left( D^{-\frac{3}{2}} \right).
    \end{aligned}
\end{equation}
\end{lemma}
\begin{proof}
    \begin{align}
         \underset{x \sim p}{\mathbb{E}} \left[ x \sin \frac{ax}{\sqrt{D}}\right] &= \underset{x \sim p}{\mathbb{E}} \left[ x \sum_{\substack{k=1 \\ \text { odd }}}^{\infty} \frac{(-1)^{\frac{k-1}{2}}}{k!}\left(\frac{a x}{\sqrt{D}}\right)^k \right] \qquad \left(\text{using } \sin x=\sum_{\substack{k=1 \\ \text { odd }}}^{\infty} \frac{(-1)^{\frac{k-1}{2}}}{k!} x^k\right)  \\
        & =\underset{x \sim p}{\mathbb{E}} \left[ \sum_{\substack{k=1 \\ \text{odd}}}^{\infty} \frac{(-1)^{\frac{k-1}{2}}}{k!} \frac{a^k}{D^{\frac{k}{2}}} \, x^{k+1} \right] \\
        & = \sum_{\substack{k=1 \\ \text{odd}}}^{\infty} \underset{x \sim p}{\mathbb{E}} \left[ \frac{(-1)^{\frac{k-1}{2}}}{k!} \frac{a^k}{D^{\frac{k}{2}}} \, x^{k+1} \right] \\
        &= \sum_{\substack{k=1 \\ \text{odd}}}^{\infty} \frac{(-1)^{\frac{k-1}{2}}}{k!} \frac{a^k}{D^{\frac{k}{2}}} \, \underset{x \sim p}{\mathbb{E}} \left[x^{k+1} \right] \\
        &= \frac{a}{\sqrt{D}} \underset{x \sim p}{\mathbb{E}} \left[h^2\right] - \frac{1}{3!} \frac{a^3}{D^{\frac{3}{2}}} \underset{x \sim p}{\mathbb{E}} \left[h^4\right] + \mathcal{O}\left(\frac{1}{D^{\frac{5}{2}}} \right) \\
        &= \frac{a}{\sqrt{D}} \underset{x \sim p}{\mathbb{E}} \left[h^2\right]+ \mathcal{O}\left(\frac{1}{D^{\frac{3}{2}}}\right).
    \end{align}
\end{proof}

\begin{lemma} \label{taylor_expansion_power}
Let $a$ and $b$ be parameters independent of $D$, and let $D$ be a large parameter. Then the following approximations hold:
\begin{equation}
    \left[1+\frac{a}{D}+\frac{b}{D^2} + \mathcal{O}\left(\frac{1}{D^3}\right)\right]^D = e^{a} \left(1 + \left(b - \frac{a^2}{2}\right)\frac{1}{D} + \mathcal{O}\left(\frac{1}{D^2}\right)\right), \label{eq:power_D}
\end{equation}
and:
\begin{equation}
    \left[1+\frac{a}{D}+\frac{b}{D^2} + \mathcal{O}\left(\frac{1}{D^3}\right)\right]^{D+1} = e^{a} \left(1 + \frac{a + b - \frac{a^2}{2}}{D} + \mathcal{O}\left(\frac{1}{D^2}\right)\right). \label{eq:power_D_plus_1}
\end{equation}
\end{lemma}

\begin{proof}
    Using the Taylor expansion:
    \begin{equation}
        \ln (1 + x)  = x - \frac{x^2}{2} + \mathcal{O}(x^3),
    \end{equation}
    we expand:
    \begin{align}
        \ln \left(1+\frac{a}{D}+\frac{b}{D^2}+ \mathcal{O}\left(\frac{1}{D^3}\right)\!\right)\! &= \frac{a}{D}+\frac{b}{D^2}+ \mathcal{O}\!\left(\frac{1}{D^3}\right)\! -\frac{1}{2}\left(\frac{a}{D}+\frac{b}{D^2}+\mathcal{O}\!\left(\frac{1}{D^3}\right)\!\right)^2 
        \hspace{-0.25cm}+\mathcal{O}\!\left(\frac{1}{D^3}\right)\\
        & =\frac{a}{D}+\frac{b}{D^2}-\frac{a^2}{2 D^2}+\mathcal{O}\left(\frac{1}{D^3}\right) \\
        & =\frac{a}{D}+\left(b-\frac{a^2}{2}\right) \frac{1}{D^2}+\mathcal{O}\left(\frac{1}{D^3}\right).
    \end{align}  
   Multiplying by $D$, we find:
    \begin{equation}
        D\ln \left(1+\frac{a}{D}+\frac{b}{D^2}+ \mathcal{O}\left(\frac{1}{D^3}\right)\right) = a+\left(b-\frac{a^2}{2}\right) \frac{1}{D}+\mathcal{O}\left(\frac{1}{D^2}\right).
    \end{equation}
    \begin{enumerate}
        \item For the first result, we compute:
        \begin{align}
            \left[1+\frac{a}{D}+\frac{b}{D^2} + \mathcal{O}\left(\frac{1}{D^3}\right)\right]^D  &=e^{D \ln \left(1+\frac{a}{D}+\frac{b}{D^2}+\mathcal{O}\left(\frac{1}{D^3}\right)\right)} \\
            & =e^{a+\left(b-\frac{a^2}{2}\right) \frac{1}{D}+\mathcal{O}\left(\frac{1}{D^2}\right)} \\
            & =e^a \cdot e^{\left(b-\frac{a^2}{2}\right) \frac{1}{D}+\mathcal{O}\left(\frac{1}{D^2}\right)} \\
            & =e^a\left(1+\left(b-\frac{a^2}{2}\right) \frac{1}{D}+\mathcal{O}\left(\frac{1} {D^2}\right)\right) .
        \end{align}
        where the Taylor expansion:
        \[
        e^x = 1 + x + \frac{x^2}{2} + \mathcal{O}(x^3)
        \]
        is applied in the last step.
        \item For the second result:
        \begin{align}
            \left[1+\frac{a}{D}+\frac{b}{D^2}+ \mathcal{O}\left(\frac{1}{D^3}\right)\right]^{D+1}  & =e^a\left(1+\left(b-\frac{a^2}{2}\right) \frac{1}{D}+\mathcal{O}\left(\frac{1}{D^2}\right)\right)\left[1+\frac{a}{D}+\mathcal{O}\left(\frac{1}{D^2}\right)\right] \\
            & =e^a\left(1+\left(b-\frac{a^2}{2}\right) \frac{1}{D}+\frac{a}{D}+\mathcal{O}\left(\frac{1}{D^2}\right)\right) \\
            & =e^a\left(1+\frac{a+b-\frac{a^2}{2}}{D}+\mathcal{O}\left(\frac{1}{D^2}\right)\right).
        \end{align}
    \end{enumerate}
\end{proof}

\begin{lemma} \label{cos_exp_lemma}
Let $x$ be a random variable distributed according to $p(x)$, and let $D$ be a large parameter. Then the following approximations hold:
    \begin{enumerate}
        \item The expectation of $\cos \left(\frac{a x}{\sqrt{D}}\right)$ is given by:
        \begin{equation}
            \underset{x \sim p}{\mathbb{E}} \left[ \cos \frac{a x}{\sqrt{D}} \right] =1-\frac{1}{2} \frac{a^2}{D} \underset{x \sim p}{\mathbb{E}} \left[x^2 \right]+\frac{1}{4!} \frac{a^4}{D^2} \underset{x \sim p}{\mathbb{E}} \left[x^4 \right]+ \mathcal{O}\left(\frac{1}{D^3}\right).
            \label{eq:cos_exp_lemma1}
        \end{equation}
        \item For the $D$-th power of expectation:
        \begin{equation}
            \begin{aligned}
            \left( \underset{x \sim p}{\mathbb{E}} \left[\cos \frac{a x}{\sqrt{D}}\right] \right)^D  &= e^{-\frac{1}{2}a^2\underset{x \sim p}{\mathbb{E}}\left[x^2\right]} \!\left[ 1 + \frac{1}{8} a^4 \left( \frac{1}{3} \underset{x \sim p}{\mathbb{E}} \left[x^4\right] \! - \! \left(\underset{x \sim p}{\mathbb{E}} \left[x^2\right] \right)^2 \right)\frac{1}{D} + \mathcal{O}\left(\frac{1}{D^2}\right) \right] \\
            &= e^{-\frac{1}{2}a^2 \underset{x \sim p}{\mathbb{E}} \left[x^2\right]} \left( 1 + \mathcal{O}\left(\frac{1}{D}\right)\right).
            \end{aligned}
            \label{eq:cos_exp_lemma2}
        \end{equation}
        \item For the $(D+1)$-th power of expectation:
        \begin{equation}
            \begin{aligned}
            \left( \underset{x \sim p}{\mathbb{E}} \left[\cos \frac{a x}{\sqrt{D}} \right]\right)^{D+1} &= e^{-\frac{1}{2}a^2\underset{x \sim p}{\mathbb{E}}\left[x^2\right]} \Bigg[ 1 +\! \Bigg(\!-\frac{1}{2} a^2 \underset{x \sim p}{\mathbb{E}} \left[x^2\right]\!+\frac{1}{4!} a^4 \underset{x \sim p}{\mathbb{E}} \left[x^4\right] \\
            &\quad -\frac{1}{8} a^4\left(\underset{x \sim p}{\mathbb{E}} \left[x^2\right]\right)^2\Bigg)\frac{1}{D} + \mathcal{O}\left(\frac{1}{D^2}\right) \Bigg] \\
            &= e^{-\frac{1}{2}a^2 \underset{x \sim p}{\mathbb{E}} \left[x^2\right]} \left( 1 + \mathcal{O}\left(\frac{1}{D}\right)\right).
        \end{aligned}
        \label{eq:cos_exp_lemma3}
        \end{equation}
    \end{enumerate}
\end{lemma}
\begin{proof}
    \begin{enumerate}
        \item For the first result, we compute:
    \begin{align}
        \underset{x \sim p}{\mathbb{E}} \left[ \cos \frac{a x}{\sqrt{D}} \right] & =\underset{x \sim p}{\mathbb{E}} \left[ \sum_{\substack{k=0 \\ 
        k \text { even }}}^{\infty} \frac{(-1)^{k / 2}}{k!}\left(\frac{a x}{\sqrt{D}}\right)^k \right] \quad \left( \text{using } \cos x=\sum_{\substack{k=0 \\ k \text { even }}}^{\infty} \frac{(-1)^{k / 2}}{k!} x^k \right) \\
        & = \sum_{\substack{k=0 \\
        k \text { even }}}^{\infty} \underset{x \sim p}{\mathbb{E}} \left[\frac{(-1)^{k / 2}}{k!}\left(\frac{a x}{\sqrt{D}}\right)^k \right] \\
        & =\sum_{\substack{k=0 \\
        k \text { even }}}^{\infty} \frac{(-1)^{k / 2}}{k!} \frac{a^k}{D^{k / 2}} \underset{x \sim p}{\mathbb{E}} \left[x^k\right] \\
        &= 1-\frac{1}{2} \frac{a^2}{D} \underset{x \sim p}{\mathbb{E}} \left[x^2\right]+\frac{1}{4!} \frac{a^4}{D^2} \underset{x \sim p}{\mathbb{E}} \left[x^4\right]+ \mathcal{O}\left(\frac{1}{D^3}\right).
    \end{align}
    \item Using the \cref{eq:cos_exp_lemma1}, we expand its $D$-th power:
    \begin{equation}
        \left( \underset{x \sim p}{\mathbb{E}} \left[ \cos \frac{a x}{\sqrt{D}} \right] \right)^D = \left[1-\frac{1}{2} \frac{a^2}{D} \underset{x \sim p}{\mathbb{E}} \left[x^2\right]+\frac{1}{4!} \frac{a^4}{D^2} \underset{x \sim p}{\mathbb{E}} \left[x^4\right]+ \mathcal{O}\left(\frac{1}{D^3}\right)\right]^D
    \end{equation}
    Let:
    \begin{align}
        s &=-\frac{1}{2} a^2 \underset{x \sim p}{\mathbb{E}} \left[x^2\right], \label{s_def}\\
        t &=\frac{1}{4!} a^4 \underset{x \sim p}{\mathbb{E}} \left[x^4\right]. \label{t_def}
    \end{align}
    Then:
    \begin{align}
    t-\frac{s^2}{2} & =\frac{1}{4!} a^4 \underset{x \sim p}{\mathbb{E}} \left[x^4\right]-\frac{1}{2}\left(-\frac{1}{2} a^2 \underset{x \sim p}{\mathbb{E}} \left[x^2\right]\right)^2 \\
    & =\frac{1}{4!} a^4 \underset{x \sim p}{\mathbb{E}} \left[x^4\right]-\frac{1}{8} a^4\left(\underset{x \sim p}{\mathbb{E}} \left[x^2\right]\right)^2 \\
    & =\frac{1}{8} a^4\left[\frac{1}{3} \underset{x \sim p}{\mathbb{E}} \left[x^4\right]-\left(\underset{x \sim p}{\mathbb{E}} \left[x^2\right]\right)^2\right].
    \end{align}
    Using \cref{eq:power_D} of lemma \ref{taylor_expansion_power}, we find:
    \begin{equation}
        \left( \underset{x \sim p}{\mathbb{E}} \left[\cos \frac{a x}{\sqrt{D}} \right] \right)^D  = e^{-\frac{1}{2}a^2 \underset{x \sim p}{\mathbb{E}} \left[x^2\right]} \! \left[1 + \frac{1}{8} a^4 \! \left(\frac{1}{3} \underset{x \sim p}{\mathbb{E}} \left[x^4\right] \! - \! \left(\underset{x \sim p}{\mathbb{E}} \left[x^2\right]\right)^2\right)\frac{1}{D} + \mathcal{O}\left(\frac{1}{D^2}\right) \right].
    \end{equation}
    \item Similarly, for the $(D+1)$-th power:
    \begin{equation}
        \left( \underset{x \sim p}{\mathbb{E}} \left[\cos \frac{a x}{\sqrt{D}} \right] \right)^{D+1} = \left[1-\frac{1}{2} \frac{a^2}{D} \underset{x \sim p}{\mathbb{E}} \left[x^2\right]+\frac{1}{4!} \frac{a^4}{D^2} \underset{x \sim p}{\mathbb{E}} \left[x^4\right]+\mathcal{O}\left(\frac{1}{D^3}\right)\right]^{D+1}.
    \end{equation}
    Using the definitions of $s$ and $t$ from \cref{s_def} and \cref{t_def}, we compute:
    \begin{equation}
        s+t-\frac{s^2}{2}=-\frac{1}{2} a^2 \underset{x \sim p}{\mathbb{E}} \left[x^2\right]+\frac{1}{8} a^4\left[\frac{1}{3} \underset{x \sim p}{\mathbb{E}} \left[x^4\right]-\left(\underset{x \sim p}{\mathbb{E}} \left[x^2\right]\right)^2\right].
    \end{equation}
    Using the \cref{eq:power_D_plus_1} of lemma \ref{taylor_expansion_power}, we get:
    \begin{equation}
            \begin{aligned}
            \left( \underset{x \sim p}{\mathbb{E}} \left[\cos \frac{a x}{\sqrt{D}} \right]\right)^{D+1} &= e^{-\frac{1}{2}a^2\underset{x \sim p}{\mathbb{E}}\left[x^2\right]} \Bigg[ \! 1 \! +\! \Bigg(\!-\frac{1}{2} a^2 \underset{x \sim p}{\mathbb{E}} \! \left[x^2\right]\!+\!\frac{1}{4!} a^4 \underset{x \sim p}{\mathbb{E}} \! \left[x^4\right] \\
            &\quad -\frac{1}{8} a^4\left(\underset{x \sim p}{\mathbb{E}} \left[x^2\right]\right)^2\Bigg)\frac{1}{D} + \mathcal{O}\left(\frac{1}{D^2}\right) \Bigg].
        \end{aligned}
        \end{equation}
    \end{enumerate}
\end{proof}

\begin{lemma}\label{cos_sum_diff_lemma}
Let $x$ and $x^{\prime}$ be independent random variables distributed according to $p(x)$, and let $D$ be a large parameter. Then the following approximation holds:
\begin{equation} \label{cos_sum_diff_eqn}
\underset{x, x^{\prime} \sim p}{\mathbb{E}} \left[ \cos \left(\left(x \pm x^{\prime}\right) \frac{a}{\sqrt{D}}\right) \right] =1-\frac{a^2}{D}\left[\underset{x \sim p}{\mathbb{E}} \left[x^2\right] \pm(\underset{x \sim p}{\mathbb{E}} \left[x\right])^2\right]+ \mathcal{O}\left(\frac{1}{D^2}\right).
\end{equation}

\end{lemma}
\begin{proof}
    Using \cref{eq:cos_exp_lemma1} of lemma \ref{cos_exp_lemma}, we have:
    \begin{equation}
        \underset{x, x^{\prime} \sim p}{\mathbb{E}} \left[\cos \left(\left(x \pm x^{\prime}\right) \frac{a}{\sqrt{D}}\right) \right]= 1-\frac{1}{2} \frac{a^2}{D} \underbrace{\underset{x, x^{\prime} \sim p}{\mathbb{E}} \left[\left(x \pm x^{\prime}\right)^2 \right]}_{\circled{1}}+ \mathcal{O}\left(\frac{1}{D^2}\right).
    \end{equation}
    Expanding $\circled{1}$, we get:
    \begin{align}
        \underset{x, x^{\prime} \sim p}{\mathbb{E}} \left[\left(x \pm x^{\prime}\right)^2 \right]&=\underset{x, x^{\prime} \sim p}{\mathbb{E}}\left[x^2+x^{\prime 2} \pm 2 x x^{\prime}\right] \\
        & =\underset{x \sim p}{\mathbb{E}} \left[x^2\right]+ \underset{x^{\prime} \sim p}{\mathbb{E}} \left[x^{\prime 2}\right] \pm 2\left(\underset{x \sim p}{\mathbb{E}} \left[x\right] \right)\left(\underset{x^{\prime} \sim p}{\mathbb{E}} \left[x^{\prime}\right]\right) \\
        & =2 \underset{x \sim p}{\mathbb{E}} \left[x^2\right] \pm 2 \left(\underset{x \sim p}{\mathbb{E}} \left[x\right] \right)^2.
    \end{align}
    Substituting $\circled{1}$ back, we obtain:
    \begin{align}
        \underset{x, x^{\prime} \sim p}{\mathbb{E}} \left[\cos \left(\left(x \pm x^{\prime}\right) \frac{a}{\sqrt{D}}\right) \right] &= 1-\frac{1}{2} \frac{a^2}{D}\left[2 \underset{x \sim p}{\mathbb{E}} \left[x^2\right] \pm 2 \left(\underset{x \sim p}{\mathbb{E}} \left[x\right]\right)^2\right]+ \mathcal{O}\left(\frac{1}{D^2}\right)\\
        &= 1-\frac{a^2}{D}\left[\underset{x \sim p}{\mathbb{E}} \left[x^2\right] \pm \left(\underset{x \sim p}{\mathbb{E}} \left[x\right] \right)^2\right]+ \mathcal{O}\left(\frac{1}{D^2}\right).
    \end{align}
\end{proof}

\begin{lemma} \label{taylor_exp_CDlamda_lemma} Let $a, b, c \in \mathbb{R}$, and let $\lambda \in[0,1]$. For a large parameter $D$, the following approximation holds: 
     \begin{equation} \label{taylor_exp_CDlamda_eqn}
     \left[1 + \frac{a}{D} + \frac{b}{D^2} + \mathcal{O}\left(\frac{1}{D^3}\right)\right]^{c D^\lambda} = 
     \begin{cases} 
     e^{a c} \left(1 + c \left(b - \frac{a^2}{2}\right) D^{-1} + \mathcal{O}\left(D^{-2}\right)\right) & \text{if } \lambda = 1, \\[0.5cm] 
     1 + a c D^{\lambda-1} + \mathcal{O}\left(D^{2\lambda-2}\right) & \text{if } 0 \leq \lambda < 1. \end{cases} 
     \end{equation}
\end{lemma}
\begin{proof}
    Recall from the proof of lemma \ref{taylor_expansion_power}, we know that:
    \begin{equation}
        \ln \left(1+\frac{a}{D}+\frac{b}{D^2}+ \mathcal{O}\left(\frac{1}{D^3}\right)\right) = \frac{a}{D}+\left(b-\frac{a^2}{2}\right) \frac{1}{D^2}+\mathcal{O}\left(\frac{1}{D^3}\right).
    \end{equation}
    Multiplying by $c D^\lambda$, we obtain:
    \begin{equation}
    c D^\lambda \ln \left(1+\frac{a}{D}+\frac{b}{D^2}+\mathcal{O}\left(\frac{1}{D^2}\right)\right)=c\left[a D^{\lambda-1}+\left(b-\frac{a^2}{2}\right) D^{\lambda-2}+\mathcal{O}\left(D^{\lambda-3}\right)\right].
    \end{equation}
    Expanding the LHS of \cref{taylor_exp_CDlamda_eqn}, we get:
    \begin{align}
        \left[1+\frac{a}{D}+\frac{b}{D^2}+ \mathcal{O}\left(\frac{1}{D^3}\right)\right]^{c D^\lambda} & =\left[1+\frac{a}{D}+\frac{b}{D^2}+\mathcal{O}\left(\frac{1}{D^3}\right)\right]^{c D^\lambda} \\
        & =e^{c D^\lambda \ln \left(1+\frac{a}{D}+\frac{b}{D^2}+\mathcal{O}\left(\frac{1}{D^3}\right)\right)} \\
        & =e^{c\left[a D^{\lambda-1}+\left(b-\frac{a^2}{2}\right) D^{\lambda-2}+\mathcal{O}\left(D^{\lambda-3}\right)\right]}.
    \end{align}
    We now handle two cases based on the value of $\lambda$.\\
    \newline
    \noindent \underline{\textbf{Case 1: $\lambda = 1$}}
    \begin{align}
        \left[1+\frac{a}{D}+\frac{b}{D^2}+ \mathcal{O}\left(\frac{1}{D^3}\right)\right]^{c D^\lambda} & =e^{c\left[a+\left(b-\frac{a^2}{2}\right) D^{-1}+\mathcal{O}\left(D^{-2}\right)\right]} \\
        & =e^{a c} e^{c\left(b-\frac{a^2}{2}\right) D^{-1}+\mathcal{O}\left(D^{-2}\right)} \\
        & =e^{a c}\left(1+c\left(b-\frac{a^2}{2}\right) D^{-1}+\mathcal{O}\left(D^{-2}\right)\right).
    \end{align}
    \noindent \underline{\textbf{Case 2: $\lambda < 1$}}
    \begin{align}
        \left[1+\frac{a}{D}+\frac{b}{D^2}+ \mathcal{O}\left(\frac{1}{D^3}\right)\right]^{c D^\lambda} \hspace{-0.5cm} &=e^{c\left[a D^{\lambda-1}+\left(b-\frac{a^2}{2}\right) D^{\lambda-2}+\mathcal{O}\left(D^{\lambda-3}\right)\right]} \\
        & =1+c\left[a D^{\lambda-1}+\left(b-\frac{a^2}{2}\right) D^{\lambda-2}+O\left(D^{\lambda-3}\right)\right]\\
        &\quad +c^2\left[a D^{\lambda-1}+\left(b-\frac{a^2}{2}\right) D^{\lambda-2}+O\left(D^{\lambda-3}\right)\right]^2 \! + \mathcal{O}\left( D^{3(\lambda-1)}\right) \\
        & =1+a c D^{\lambda-1}+a^2 c^2 D^{2 \lambda-2}+\mathcal{O}\left(D^{\lambda-2}\right) \\
        & =1+a c D^{\lambda-1}+\mathcal{O}\left(D^{2 \lambda-2}\right).
    \end{align}
\end{proof}

\begin{lemma} \label{cos_cd_lemma} Let $x$ be a random variable distributed according to $p(x)$, with $c \in \mathbb{R}, a$ independent of $D$, and $\lambda \in[0,1]$. For large $D$, the following approximation holds:
    \begin{equation}
    \left(\underset{x \sim p}{\mathbb{E}} \left[\cos \frac{a x}{\sqrt{D}}\right] \! \right)^{c D^\lambda} \hspace{-0.6cm}= \begin{cases}
    e^{-\frac{c}{2} a^2 \underset{x \sim p}{\mathbb{E}}  \left[x^2\right]} \hspace{-0.1cm} \left[ \! 1 \! + \! \frac{1}{8} a^4 c\left( \! \frac{1}{3} \underset{x \sim p}{\mathbb{E}} \! \left[x^4\right] \! - \! \left(\underset{x \sim p}{\mathbb{E}} \! \left[x^2\right]\right)^2\right) \! \frac{1}{D} \! + \! \mathcal{O} \! \left(\frac{1}{D^2}\right) \! \right]  &  \lambda=1, \\[0.5cm]
    1-\frac{c}{2} a^2 \underset{x \sim p}{\mathbb{E}}  \left[x^2\right] D^{\lambda-1}+\mathcal{O}\left(D^{2 \lambda-2}\right) & 0 \leqslant \lambda<1.\end{cases}
    \end{equation}
\end{lemma}
\begin{proof}
    Raising \cref{eq:cos_exp_lemma1} of lemma \ref{cos_exp_lemma} to the power $cD^\lambda$, we have:
    \begin{equation}
            \left(\underset{x \sim p}{\mathbb{E}} \left[ \cos \frac{a x}{\sqrt{D}} \right]\right)^{cD^\lambda}= \left[1-\frac{1}{2} \frac{a^2}{D} \underset{x \sim p}{\mathbb{E}} \left[x^2 \right]+\frac{1}{4!} \frac{a^4}{D^2} \underset{x \sim p}{\mathbb{E}} \left[x^4 \right]+ \mathcal{O}\left(\frac{1}{D^3}\right)\right]^{cD^\lambda}.
    \end{equation}
    Using lemma \ref{taylor_exp_CDlamda_lemma}, we evaluate the expression for the two cases:

    \noindent \underline{\textbf{Case 1: $\lambda = 1$}}
    \begin{equation}
        \left(\underset{x \sim p}{\mathbb{E}} \left[\cos \frac{a x}{\sqrt{D}}\right]\right)^{cD^\lambda}= e^{-\frac{c}{2}a^2 \underset{x \sim p}{\mathbb{E}} \left[x^2\right]}\left[1 + \frac{1}{8} a^4c\left(\frac{1}{3} \underset{x \sim p}{\mathbb{E}} \left[x^4\right]-\left(\underset{x \sim p}{\mathbb{E}} \left[x^2\right]\right)^2\right)\frac{1}{D} + \mathcal{O}\left(\frac{1}{D^2}\right) \right].
    \end{equation}
    \noindent \underline{\textbf{Case 2: $0 \geq \lambda < 1$}}
    \begin{equation}
        \left(\underset{x \sim p}{\mathbb{E}} \left[\cos \frac{a x}{\sqrt{D}}\right]\right)^{cD^\lambda}= 1-\frac{c}{2} a^2 \underset{x \sim p}{\mathbb{E}} \left[x^2\right] D^{\lambda-1}+ \mathcal{O}\left(D^{2 \lambda-2}\right).
    \end{equation}
\end{proof}

\begin{lemma} \label{cos_sum_diff_exp_lemma} Let $x$ and $x^{\prime}$ be independent random variables distributed according to $p(x)$, with $c \in \mathbb{R}, a$ independent of $D$, and $\lambda \in[0,1]$. For large $D$, the following approximation holds:
    \begin{equation}
    \left[\underset{x, x^{\prime} \sim p}{\mathbb{E}} \! \left[ \cos \! \left( \! \left(x \pm x^{\prime}\right) \frac{a}{\sqrt{D}}\right) \! \right] \! \right]^{cD^{\lambda}} \hspace{-0.5cm}= \hspace{-0.1cm}
    \begin{cases}
    e^{-a^2 c\left[\underset{x \sim p}{\mathbb{E}} \left[x^2\right] \pm(\underset{x \sim p}{\mathbb{E}} \left[x\right])^2\right]}\!\left(1+\mathcal{O}\left(D^{-1}\right)\right) & \text{if }  \lambda=1, \\[0.5cm]
    1 \! - \! a^2 c \! \left[\underset{x \sim p}{\mathbb{E}} \left[x^2\right] \! \pm \! (\underset{x \sim p }{\mathbb{E}} \left[x\right])^2\right] \! D^{\lambda-1} \! +\mathcal{O} \! \left(D^{2 \lambda-2}\right) & \text{if } 0 \leqslant \lambda<1.
    \end{cases}
    \end{equation}
\end{lemma}
\begin{proof}
    Raising \cref{cos_sum_diff_eqn} of lemma \ref{cos_sum_diff_lemma} to the power $cD^\lambda$, we get:
    \begin{equation}
    \left[\underset{x, x^{\prime} \sim p}{\mathbb{E}} \left[ \cos \left(\left(x \pm x^{\prime}\right) \frac{a}{\sqrt{D}}\right) \right] \right]^{cD^{\lambda}}\hspace{-0.5cm} =\left[1-\frac{a^2}{D}\left[\underset{x \sim p}{\mathbb{E}} \left[ x^2 \right] \pm(\underset{x \sim p}{\mathbb{E}} \left[ x \right])^2\right]+ \mathcal{O}\left(\frac{1}{D^2}\right)\right]^{cD^{\lambda}}.
    \end{equation}
    Using lemma \ref{taylor_exp_CDlamda_lemma}, we obtain the stated result.
\end{proof}

\subsection{Proof of \texorpdfstring{\Cref{fst_local_opt_thm}}{fst local opt thm}} \label{fst_local_opt_proof}

Before proceeding with the proof of \cref{fst_local_opt_thm}, we redefine some notations introduced earlier in the manuscript for improved clarity and compactness. These adjustments facilitate a more streamlined presentation of the proof. Specifically, for an edge $\{u, v\} \in E$ in a graph $G = (V, E)$, we define:
\begin{itemize}\setlength\itemsep{0.1em}
    \item $\mathcal{N}_{v \setminus u} = \mathcal{N}(v) \setminus \{u\}$: the set of vertices connected to $v$ excluding $u$.
    \item $\mathcal{N}_{u \setminus v} = \mathcal{N}(u) \setminus \{v\}$: the set of vertices connected to $u$ excluding $v$.
    \item $\mathcal{N}_{v \doublesetminus u} = \{w \in \mathcal{N}(v) \mid w \neq u \text{ and } w \notin \mathcal{N}(u)\}$: the neighbours of $v$ that are not neighbours of $u$.
    \item $\mathcal{N}_{u \doublesetminus v} = \{w \in \mathcal{N}(u) \mid w \neq v \text{ and } w \notin \mathcal{N}(v)\}$: the neighbours of $u$ that are not neighbours of $v$.
    \item $\mathcal{N}_{uv} = \mathcal{N}_{u \setminus v} \cap \mathcal{N}_{v \setminus u} = \mathcal{N}(u) \cap \mathcal{N}(v)$: the set of vertices that form a triangle with the edge $\{u, v\}$.
\end{itemize}
Here, $\mathcal{N}(w) = \{x \in V \mid \{x, w\} \in E\}$ denotes the set of neighbours of a vertex $w$.

\begin{proof}
Recall that the QAOA$_1$ cost function for an Ising model defined on a graph $G=(V, E)$ with respect to parameters $\beta, \gamma \in \mathbb{R}$ is
\begin{equation}
\xi_G(\gamma, \beta)=\sum_{i \in V} C_i(\gamma, \beta; G)+\sum_{\{u, v\} \in E} C_{u v}^{(1)}(\gamma, \beta ; G)+ \sum_{\{u, v\} \in E}C_{u v}^{(2)}(\gamma, \beta ; G).
\end{equation}
where the functions $C_i, C_{u v}^{(1)}$, and $C_{u v}^{(2)}$ are defined as follows:
\begin{align}
    C_i(\gamma, \beta: G) &=h_i \sin (2 \beta) \sin \left(2 h_i \gamma\right) \prod_{k \in \mathcal{N}(i)} \cos \left(2 J_{i k} \gamma\right), \\
    C_{u v}^{(1)}(\gamma, \beta ; G) &=\frac{J_{u v}}{2} \sin (4 \beta) \sin \left(2 J_{u v} \gamma\right) \hspace{-0.1cm} \left\{\hspace{-0.1cm} \cos \left( 2 h_v \gamma\right) \hspace{-0.35cm} \prod_{w \in \mathcal{N}_{v \setminus u}} \hspace{-0.35cm} \cos \left( 2 J_{w v} \gamma\right) +\cos \left(2 h_u \gamma\right) \hspace{-0.35cm} \prod_{w \in \mathcal{N}_{u \setminus v}} \hspace{-0.35cm} \cos \left( 2 J_{u w} \gamma\right) \hspace{-0.1cm} \right\}, \\
C_{u v}^{(2)}(\gamma, \beta ; G) &=-\frac{J_{u v}}{2} \sin ^2 (2 \beta) \hspace{-0.3cm} \prod_{w \in N_{v \doublesetminus u}} \hspace{-0.3cm} \cos \left( 2 J_{w v} \gamma\right) \hspace{-0.25cm} \prod_{w \in N_{u \doublesetminus v}} \hspace{-0.25cm} \cos \left(2 J_{u w} \gamma\right) \\
& \hspace{-2.1cm} \times \left\{\cos \left(2 \left(h_u+h_v\right) \gamma \right) \hspace{-0.2cm} \prod_{f \in N_{u v}} \hspace{-0.2cm} \cos \left(2 \left(J_{u f}+J_{v f}\right) \gamma \right)  -\cos \left(2 \left(h_u-h_v\right) \gamma\right) \hspace{-0.2cm} \prod_{f \in N_{u v}} \hspace{-0.2cm} \cos \left(2 \left(J_{u f}-J_{v f}\right) \gamma\right)\right\}. \nonumber
\end{align}
Let $f: \mathbb{R} \rightarrow \mathbb{R}$ and $g: \mathbb{R} \rightarrow \mathbb{R}$ be probability density functions. Consider a graph $G=(V, E)$ where the edge weights $J_{u v}$ are independently drawn from $f$ and the vertex weights $h_i$ are independently drawn from $g$. The expected QAOA$_1$ cost function for $G$ with parameters $\beta$ and $\gamma$ is then
\begin{equation}
    \underset{\substack{J \sim f^E \\ h \sim g^V}}{\mathbb{E}} \left[\xi_G(\gamma, \beta) \right] = \sum_{i \in V} \underset{\substack{J \sim f^E \\ h \sim g^V}}{\mathbb{E}} \left[C_i(\gamma, \beta: G)\right] + \hspace{-0.25cm} \sum_{\{u, v\} \in E} \underset{\substack{J \sim f^E \\ h \sim g^V}}{\mathbb{E}} \left\{C_{u v}^{(1)}(\gamma, \beta ; G)+ \hspace{-0.25cm} \sum_{\{u, v\} \in E}C_{u v}^{(2)}(\gamma, \beta ; G)\right\}.
\end{equation}
To facilitate the analysis, we decompose the expectation into separate terms:
\begin{align}
    T_1 &= \underset{\substack{J \sim f^E \\ h \sim g^V}}{\mathbb{E}} \left[C_i(\gamma, \beta: G)\right], \\
    T_2 &= \underset{\substack{J \sim f^E \\ h \sim g^V}}{\mathbb{E}} \left[ C_{u v}^{(1)}(\gamma, \beta ; G) \right], \\
    T_3 &= \underset{\substack{J \sim f^E \\ h \sim g^V}}{\mathbb{E}} \left[ C_{u v}^{(2)}(\gamma, \beta ; G) \right].
\end{align}
We will compute each of these terms individually to understand their contributions to the overall expected cost function.

\vspace{0.25cm}
\noindent \underline{\textbf{Step 1: Computing the Term $T_1$}}\\
\newline
We begin by evaluating the expectation $T_1$ of the term $C_i(\gamma, \beta ; G)$ over the distributions of the edge weights $J \sim f^E$ and vertex weights $h \sim g^V$:
\begin{align}
    T_1 &= \underset{\substack{J \sim f^E \\ h \sim g^V}}{\mathbb{E}} \left[C_i(\gamma, \beta: G)\right]\\
    &= \underset{\substack{J \sim f^E \\ h \sim g^V}}{\mathbb{E}} \left[ h_i \sin (2\beta) \sin(2h_i\gamma) \prod_{k \in \mathcal{N}(i)} \cos (2J_{ik}\gamma) \right] \\
    &= \sin(4\beta) \underset{h \sim g}{\mathbb{E}} \left[h \sin(2h\gamma)\right] \prod_{k \in \mathcal{N}(i)} \underset{J \sim f}{\mathbb{E}} \left[\cos (2J\gamma)\right]\\
    &= \sin(4\beta) \underset{h \sim g}{\mathbb{E}} \left[h \sin(2h\gamma)\right] \left( \underset{J \sim f}{\mathbb{E}} \left[\cos (2J\gamma)\right] \right)^{|\mathcal{N}(i)|}.
\end{align}
Since the graph $G$ is $(D+1)$-regular, each vertex $i$ has $|\mathcal{N}(i)|=D+1$ neighbors. By setting $\gamma=\frac{\alpha}{\sqrt{D}}$, the expression for $T_1$ simplifies to:
\begin{equation}
    T_1 = \sin(4\beta) \underbrace{\underset{h \sim g}{\mathbb{E}} \left[h \sin(2h\gamma)\right]}_{\circled{1}} \underbrace{\left( \underset{J \sim f}{\mathbb{E}} \left[\cos (2J\gamma)\right] \right)^{D+1}}_{\circled{2}}.
\end{equation}
Next, we apply \cref{sin_taylor_lemma} and \cref{cos_exp_lemma} to approximate the terms $\circled{1}$ and $\circled{2}$ as follows:
\begin{align}
    \circled{1} &= \underset{h \sim g}{\mathbb{E}} \left[h \sin \frac{2h \alpha}{\sqrt{D}}\right]=\frac{2\alpha}{\sqrt{D}} \underset{h \sim g}{\mathbb{E}} \left[h^2\right]+ \mathcal{O}\left(\frac{1}{D^{3 / 2}}\right),\\
    \circled{2} &= \left( \underset{J \sim f}{\mathbb{E}} \left[\cos (2J\gamma)\right] \right)^{D+1}=e^{-2 \alpha^2 \underset{J \sim f}{\mathbb{E}} \left[J^2\right]}\left(1+\mathcal{O}\left(\frac{1}{D}\right)\right).
\end{align}
Substituting $\circled{1}$ and $\circled{2}$ back into the expression for $T_1$, we obtain:
\begin{align}
    T_1 &= \sin(4\beta) \left[ \frac{2\alpha}{\sqrt{D}} \underset{h \sim g}{\mathbb{E}} \left[h^2\right]+ \mathcal{O}\left(\frac{1}{D^{3 / 2}}\right) \right]e^{-2 \alpha^2 \underset{J \sim f}{\mathbb{E}} \left[J^2\right]}\left(1+\mathcal{O}\left(\frac{1}{D}\right)\right) \\
    &= \sin(4\beta) \underset{h \sim g}{\mathbb{E}} \left[h^2\right] \frac{2\alpha}{\sqrt{D}}e^{-2 \alpha^2 \underset{J \sim f}{\mathbb{E}} \left[J^2\right]}+\mathcal{O}\left(\frac{1}{D}\right).
\end{align}
\vspace{0.25cm}
\noindent \underline{\textbf{Step 2: Computing the Term $T_2$}}\\
\newline
Next, we evaluate the expectation $T_2$ of the term $C_{u v}^{(1)}(\gamma, \beta ; G)$ over the distributions of the edge weights $J \sim f^E$ and vertex weights $h \sim g^V$:
\begin{align}
    T_2 &= \underset{\substack{J \sim f^E \\ h \sim g^V}}{\mathbb{E}} \left[C_{u v}^{(1)}(\gamma, \beta ; G) \right]\\
    &= \underset{\substack{J \sim f^E \\ h \sim g^V}}{\mathbb{E}} \left[\frac{J_{u v}}{2} \sin (4 \beta) \sin \left(2 J_{u v} \gamma\right) \hspace{-0.1cm} \left\{\cos \left( 2 h_v \gamma\right) \hspace{-0.35cm} \prod_{w \in \mathcal{N}_{v \setminus u}} \hspace{-0.35cm} \cos \left( 2 J_{w v} \gamma\right) +\cos \left(2 h_u \gamma\right) \hspace{-0.35cm} \prod_{w \in \mathcal{N}_{u \setminus v}} \hspace{-0.35cm} \cos \left( 2 J_{u w} \gamma\right)\right\} \right] \\
    &= \frac{1}{2} \sin(4\beta) \underset{J \sim f}{\mathbb{E}}\left[ J \sin(2J \gamma) \right] \hspace{-0.1cm}  \Bigg\{\underset{h \sim g}{\mathbb{E}}\left[\cos \left( 2 h \gamma\right)\right] \hspace{-0.4cm} \prod_{w \in \mathcal{N}_{v \setminus u}} \hspace{-0.25cm} \underset{J \sim f}{\mathbb{E}}\left[\cos \left( 2 J \gamma\right)\right] \hspace{-0.1cm}  + \hspace{-0.1cm} \underset{h \sim g}{\mathbb{E}}\left[\cos \left( 2 h \gamma\right)\right] \hspace{-0.35cm} \prod_{w \in \mathcal{N}_{u \setminus v}} \hspace{-0.25cm} \underset{J \sim f}{\mathbb{E}}\left[\cos \left( 2 J \gamma\right)\right]\hspace{-0.1cm} \Bigg\} \\
    &= \frac{1}{2} \sin(4\beta) \underset{J \sim f}{\mathbb{E}}\left[ J \sin(2J \gamma) \right] \underset{h \sim g}{\mathbb{E}}\left[\cos \left( 2 h \gamma\right)\right] \hspace{-0.1cm} \left\{ \hspace{-0.15cm} \left(\underset{J \sim f}{\mathbb{E}}\left[\cos \left( 2 J \gamma\right)\right] \hspace{-0.1cm}\right)^{|\mathcal{N}_{v\setminus u}|} \hspace{-0.35cm} + \hspace{-0.1cm} \left(\underset{J \sim f}{\mathbb{E}}\left[\cos \left( 2 J \gamma\right)\right] \hspace{-0.1cm}\right)^{|\mathcal{N}_{u \setminus v}|} \hspace{-0.05cm} \right\}.
\end{align}
Given that the graph $G$ is $(D+1)$-regular, each edge $\{u, v\}$ has neighbourhoods satisfying:
\begin{equation}
    |\mathcal{N}_{u \setminus v}| = |\mathcal{N}_{v \setminus u}| = (D+1)-1 = D.
\end{equation}
Substituting this into the expression and setting $\gamma=\frac{\alpha}{\sqrt{D}}$, we further simplify $T_2$ as:
\begin{equation}
    T_2 = \sin(4\beta) \underbrace{\underset{J \sim f}{\mathbb{E}}\left[ J \sin(2J \frac{\alpha}{\sqrt{D}}) \right]}_{\circled{1}} \underbrace{\underset{h \sim g}{\mathbb{E}}\left[\cos \left( 2 h \frac{\alpha}{\sqrt{D}}\right)\right]}_{\circled{2}} \underbrace{\left(\underset{J \sim f}{\mathbb{E}}\left[\cos \left( 2 J \frac{\alpha}{\sqrt{D}}\right)\right]\right)^D}_{\circled{3}}.
\end{equation}
Applying \cref{sin_taylor_lemma} and \cref{cos_exp_lemma}, we approximate the terms $\circled{1}$, $\circled{2}$, and $\circled{3}$ as follows:
\begin{align}
    \circled{1} &= \underset{J \sim f}{\mathbb{E}}\left[ J \sin(2J \frac{\alpha}{\sqrt{D}}) \right] = \frac{2\alpha}{\sqrt{D}} \underset{J \sim f}{\mathbb{E}}\left[J^2\right] \mathcal{O}\left( \frac{1}{D^{\frac{3}{2}}} \right), \\
    \circled{2} &= \underset{h \sim g}{\mathbb{E}}\left[\cos \left( 2 h \frac{\alpha}{\sqrt{D}}\right)\right] = 1 + \mathcal{O}\left( \frac{1}{D} \right),\\
    \circled{3} &= \left(\underset{J \sim f}{\mathbb{E}}\left[\cos \left( 2 J \frac{\alpha}{\sqrt{D}}\right)\right]\right)^D = e^{-2\alpha^2 \underset{J \sim f}{\mathbb{E}}\left[ J^2\right]}\left(1 + \mathcal{O}\left( \frac{1}{D} \right) \right).
\end{align}
Substituting $\circled{1}$, $\circled{2}$, and $\circled{3}$ back into the expression for $T_2$, we obtain:
\begin{align}
    T_2 &= \sin(4\beta) \left[ \frac{2\alpha}{\sqrt{D}} \underset{J \sim f}{\mathbb{E}} \left[J^2\right]+ \mathcal{O}\left(\frac{1}{D^{3 / 2}}\right) \right]e^{-2 \alpha^2 \underset{J \sim f}{\mathbb{E}} \left[J^2\right]}\left(1+\mathcal{O}\left(\frac{1}{D}\right)\right) \\
    &= \sin(4\beta) \underset{J \sim f}{\mathbb{E}} \left[J^2\right] \frac{2\alpha}{\sqrt{D}}e^{-2 \alpha^2 \underset{J \sim f}{\mathbb{E}} \left[J^2\right]}+\mathcal{O}\left(\frac{1}{D}\right).
\end{align}

\vspace{0.25cm}
\noindent \underline{\textbf{Step 3: Computing the Term $T_3$}}\\
\newline
We now evaluate the expectation $T_3$ of the term $C_{u v}^{(2)}(\gamma, \beta ; G)$ over the distributions of the edge weights $J \sim f^E$ and vertex weights $h \sim g^V$ as follows:
{
  \allowdisplaybreaks
\begin{align}
    T_3 &= \underset{\substack{J \sim f^E \\ h \sim g^V}}{\mathbb{E}} \left[ C_{u v}^{(2)}(\gamma, \beta ; G) \right] \\
    &= \underset{\substack{J \sim f^E \\ h \sim g^V}}{\mathbb{E}} \left\{ -\frac{J_{u v}}{2} \sin ^2 (2 \beta) \prod_{w \in N_{v \doublesetminus u}} \cos \left( 2 J_{w v} \gamma\right) \prod_{w \in N_{u \doublesetminus v}} \cos \left(2 J_{u w} \gamma\right) \right. \\
    & \quad \times \left\{\cos \left(2 \left(h_u+h_v\right) \gamma \right) \hspace{-0.2cm} \prod_{f \in N_{u v}} \hspace{-0.2cm} \cos \left(2 \left(J_{u f}+J_{v f}\right) \gamma \right)  \left. -\cos \left(2 \left(h_u-h_v\right) \gamma\right) \hspace{-0.2cm} \prod_{f \in N_{u v}} \hspace{-0.2cm} \cos \left(2 \left(J_{u f}-J_{v f}\right) \gamma\right)\right\} \right\} \nonumber \\
    &= -\frac{1}{2} \sin^2(2\beta) \underset{J \sim f}{\mathbb{E}}\left[ J \right] \prod_{w \in N_{v \doublesetminus u}} \underset{J \sim f}{\mathbb{E}}\left[\cos \left( 2 J \gamma\right) \right] \prod_{w \in N_{u \doublesetminus v}} \underset{J \sim f}{\mathbb{E}}\left[\cos \left( 2 J \gamma\right) \right] \\
    & \quad \times \left\{ \underset{h,h^{\prime} \sim g}{\mathbb{E}} \left[ \cos \left(2(h+h^{\prime})\gamma \right) \right] \prod_{f \in N_{u v}} \underset{J,J^{\prime} \sim f}{\mathbb{E}} \left[ \cos \left(2(J+J^{\prime})\gamma \right) \right] \right. \nonumber \\
    & \quad \left. - \underset{h,h^{\prime} \sim g}{\mathbb{E}} \left[ \cos \left(2(h-h^{\prime})\gamma \right) \right] \prod_{f \in N_{u v}} \underset{J,J^{\prime} \sim f}{\mathbb{E}} \left[ \cos \left(2(J-J^{\prime})\gamma \right) \right] \right\} \nonumber \\
    &= -\frac{1}{2} \sin^2(2\beta) \underset{J \sim f}{\mathbb{E}}\left[ J \right] \left(\underset{J \sim f}{\mathbb{E}}\left[\cos \left( 2 J \gamma\right) \right]\right)^{|N_{v \doublesetminus u}|+|N_{u \doublesetminus v}|} \\
    & \quad \times \left\{ \underset{h,h^{\prime} \sim g}{\mathbb{E}} \left[ \cos \left(2(h+h^{\prime})\gamma \right) \right] \left( \underset{J,J^{\prime} \sim f}{\mathbb{E}} \left[ \cos \left(2(J+J^{\prime})\gamma \right) \right]\right)^{|\mathcal{N}_{uv}|} \right. \nonumber \\
    & \quad \left.  - \underset{h,h^{\prime} \sim g}{\mathbb{E}} \left[ \cos \left(2(h-h^{\prime})\gamma \right) \right] \left( \underset{J,J^{\prime} \sim f}{\mathbb{E}} \left[ \cos \left(2(J-J^{\prime})\gamma \right) \right]\right)^{|\mathcal{N}_{uv}|} \right\} . \nonumber
\end{align}}
Since $G$ is $(D+1)$-regular graph, we have:
\begin{align}
    |\mathcal{N}(u)|&=|\mathcal{N}(v)|=D+1\\
    \mathcal{N}(u) &= \mathcal{N}_{u \doublesetminus v} \cup \mathcal{N}_{uv} \cup \{v\},\\
    \mathcal{N}(v) &= \mathcal{N}_{v \doublesetminus u} \cup \mathcal{N}_{uv} \cup \{u\},
\end{align}
which implies:
\begin{align}
D+1&=|\mathcal{N}(u)|=\left|\mathcal{N}_{u \doublesetminus v}\right|+\left|\mathcal{N}_{u v}\right|+1, \\
D+1&=|\mathcal{N}(v)|=\left|\mathcal{N}_{v \doublesetminus u}\right|+\left|\mathcal{N}_{u v}\right|+1.
\end{align}
Hence, we get:
\begin{align}
& \left|\mathcal{N}_{u \doublesetminus v}\right|+\left|\mathcal{N}_{uv}\right|=\left|\mathcal{N}_{v \doublesetminus u}\right|+\left|\mathcal{N}_{uv}\right|=D, \\
\Rightarrow \quad & \left|\mathcal{N}_{u \doublesetminus v}\right|=\left|\mathcal{N}_{v \doublesetminus u}\right|,\left|\mathcal{N}_{uv}\right| \leq D.
\end{align}
Setting $\gamma = \frac{\alpha}{\sqrt{D}}$ and using $\left|\mathcal{N}_{u \doublesetminus v}\right|=\left|\mathcal{N}_{v \doublesetminus u}\right|$, $T_3$ can be written as follows:
\begin{equation}
    \begin{aligned}
    T_3 &= -\frac{1}{2} \sin^2(2\beta) \underset{J \sim f}{\mathbb{E}}\left[ J \right] \left(\underset{J \sim f}{\mathbb{E}}\left[\cos \left( 2 J \frac{\alpha}{\sqrt{D}}\right) \right]\right)^{2|N_{u \doublesetminus v}|} \\
    & \quad \times \Bigg\{ \underset{h,h^{\prime} \sim g}{\mathbb{E}} \left[ \cos \left(2(h+h^{\prime})\frac{\alpha}{\sqrt{D}} \right) \right] \left( \underset{J,J^{\prime} \sim f}{\mathbb{E}} \left[ \cos \left(2(J+J^{\prime})\frac{\alpha}{\sqrt{D}} \right) \right]\right)^{|\mathcal{N}_{uv}|}   \\
    & \quad -  \underset{h,h^{\prime} \sim g}{\mathbb{E}} \left[ \cos \left(2(h-h^{\prime})\frac{\alpha}{\sqrt{D}} \right) \right] \left( \underset{J,J^{\prime} \sim f}{\mathbb{E}} \left[ \cos \left(2(J-J^{\prime})\frac{\alpha}{\sqrt{D}} \right) \right]\right)^{|\mathcal{N}_{uv}|} \Bigg\}.
\end{aligned}
\end{equation}
For some $a,b \geq 0$ and $0 \leq \lambda, \mu \leq 1$, we let
\begin{align}
    |\mathcal{N}_{u \doublesetminus v}| &= aD^{\lambda}, \\
    |\mathcal{N}_{uv}| &= bD^{\mu}.
\end{align}
such that
\begin{equation}
\left|\mathcal{N}_{u \doublesetminus v}\right|+\left|\mathcal{N}_{uv}\right|=D \Rightarrow a D^\lambda+b D^\mu=D.
\end{equation}
Using the above, $T_3$ can be rewritten as follows:
\begin{equation}
    \begin{aligned}
    T_3 &= -\frac{1}{2} \sin^2(2\beta) \underset{J \sim f}{\mathbb{E}}\left[ J \right] \underbrace{\left(\underset{J \sim f}{\mathbb{E}}\left[\cos \left( 2 J \frac{\alpha}{\sqrt{D}}\right) \right]\right)^{2aD^{\lambda}}}_{\circled{1}} \\
    & \quad \times \Bigg\{ \underbrace{\underset{h,h^{\prime} \sim g}{\mathbb{E}} \left[ \cos \left(2(h+h^{\prime})\frac{\alpha}{\sqrt{D}} \right) \right]}_{\circled{2}} \underbrace{\left( \underset{J,J^{\prime} \sim f}{\mathbb{E}} \left[ \cos \left(2(J+J^{\prime})\frac{\alpha}{\sqrt{D}} \right) \right]\right)^{bD^{\mu}}}_{\circled{3}}   \\
    & \quad -  \underbrace{\underset{h,h^{\prime} \sim g}{\mathbb{E}} \left[ \cos \left(2(h-h^{\prime})\frac{\alpha}{\sqrt{D}} \right) \right]}_{\circled{4}} \underbrace{\left( \underset{J,J^{\prime} \sim f}{\mathbb{E}} \left[ \cos \left(2(J-J^{\prime})\frac{\alpha}{\sqrt{D}} \right) \right]\right)^{bD^{\mu}}}_{\circled{5}} \Bigg\}.
\end{aligned}
\end{equation}
We now compute the individual terms of $T_3$ below.\\

\begin{enumerate}
    \item \underline{Computation of $\circled{1}$}
        \begin{itemize}
        \item \textbf{Case 1: }$a = 0$ (i.e., $|\mathcal{N}_{u \doublesetminus v}| = 0$)
            \begin{equation}
                \circled{1} = \left(\underset{J \sim f}{\mathbb{E}}\left[\cos \left( 2 J \frac{\alpha}{\sqrt{D}}\right) \right]\right)^0 = 1.
            \end{equation}
        \item \textbf{Case 2:} $a > 0$ (i.e., $|\mathcal{N}_{u \doublesetminus v}| \neq 0$)

        \noindent Using \cref{cos_cd_lemma} with substitutions \( x \to J \), \( p \to f \), \( a \to 2\alpha \), and \( c \to 2a \), we obtain:
        \begin{equation}
            \circled{1} = \begin{cases}
                e^{-4a\alpha^2 \underset{J \sim f}{\mathbb{E}} \left[ J^2 \right]} \left(1 + \mathcal{O}\left( D^{-1}\right) \right) & \text{if } \lambda = 1,\\[0.25cm]
                1- 4a\alpha^2 \underset{J \sim f}{\mathbb{E}} \left[ J^2 \right] D^{\lambda -1} \mathcal{O}\left(D^{2\lambda - 2} \right) &  \text{if } 0 \geq \lambda < 1.
            \end{cases}
        \end{equation}
    \end{itemize}
    Hence, $\circled{1}$ can be written as follows:
    \begin{equation}
        \circled{1} = \begin{cases}
                1 &  \text{if } a = 0, \\[0.25cm]
                e^{-4a\alpha^2 \underset{J \sim f}{\mathbb{E}} \left[ J^2 \right]} \left(1 + \mathcal{O}\left( D^{-1}\right) \right) & \text{if } a > 0, \, \lambda = 1, \\[0.25cm]
                1- 4a\alpha^2 \underset{J \sim f}{\mathbb{E}} \left[ J^2 \right] D^{\lambda -1} \mathcal{O}\left(D^{2\lambda - 2} \right) & \text{if } a >0, \, 0 \geq \lambda < 1.
            \end{cases}
    \end{equation}
    \item \underline{Computation of $\circled{2}$ and $\circled{4}$}
    \begin{align}
            \circled{2,4} &= \underset{h,h^{\prime} \sim g}{\mathbb{E}} \left[ \cos \left(2(h \pm h^{\prime})\frac{\alpha}{\sqrt{D}} \right) \right]\\
            &= 1-\frac{4\alpha^2}{D}\left[\underset{h \sim g}{\mathbb{E}} \left[h^2\right] \pm(\underset{h \sim g}{\mathbb{E}} [h])^2\right]+\mathcal{O}\left(\frac{1}{D^2}\right) \\ 
            & =1+\mathcal{O}\left(\frac{1}{D}\right)
    \end{align}
    where the second line was obtained using lemma \ref{cos_sum_diff_lemma}.
    \item \underline{Computation of $\circled{3}$ and $\circled{5}$}
\begin{itemize}
        \item \underline{\textbf{Case 1:} $b = 0$ ($|\mathcal{N}_{uv}| = 0$)}
        \begin{equation}
            \circled{3,5} = \left( \underset{J,J^{\prime} \sim f}{\mathbb{E}} \left[ \cos \left(2(J \pm J^{\prime})\frac{\alpha}{\sqrt{D}} \right) \right]\right)^{0} = 1.
        \end{equation}
        \item \underline{\textbf{Case 2:} $b > 0$ ($|\mathcal{N}_{uv}| > 0$)}\\
    
        \noindent Using lemma \ref{cos_sum_diff_exp_lemma}, and setting $x \to J$, $p \to f$, $a \to 2\alpha$, $c \to b$, and $\lambda \to \mu$, we get:
        \begin{align}
            \circled{3,5} &= \left( \underset{J,J^{\prime} \sim f}{\mathbb{E}} \left[ \cos \left(2(J \pm J^{\prime})\frac{\alpha}{\sqrt{D}} \right) \right]\right)^{bD^{\mu}} \\
            &= \begin{cases}
                e^{-4\alpha^2b \left( \underset{J \sim f}{\mathbb{E}}\left[J^2 \right] \pm \left( \underset{J \sim f}{\mathbb{E}}\left[J \right]\right)^2 \right)} \left(1 + \mathcal{O}\left(D^{-1}\right) \right) & \text{if } \mu = 1,\\[0.25cm]
                1 - 4\alpha^2b \left( \underset{J \sim f}{\mathbb{E}}\left[J^2 \right] \pm \left( \underset{J \sim f}{\mathbb{E}}\left[J \right]\right)^2 \right)D^{\mu - 1} + \mathcal{O}\left( D^{2\mu - 2}\right) &\text{if } 0 \leq \mu < 1.
            \end{cases}
        \end{align}
    \end{itemize}
    Hence, $\circled{3,5}$ can be written as follows:
    \begin{equation}
        \circled{3,5} = \begin{cases}
                1 & \text{if } b = 0, \\[0.25cm]
                e^{-4\alpha^2b \left( \underset{J \sim f}{\mathbb{E}}\left[J^2 \right] \pm \left( \underset{J \sim f}{\mathbb{E}}\left[J \right]\right)^2 \right)} \left(1 + \mathcal{O}\left(D^{-1}\right) \right) &\text{if } b > 0, \, \mu = 1,\\[0.25cm]
                1 - 4\alpha^2b \left( \underset{J \sim f}{\mathbb{E}}\left[J^2 \right] \pm \left( \underset{J \sim f}{\mathbb{E}}\left[J \right]\right)^2 \right)D^{\mu - 1} + \mathcal{O}\left( D^{2\mu - 2}\right) & \text{if } b > 0, \, 0 \leq \mu < 1.
            \end{cases}
    \end{equation}

    \item \underline{Computation of $\circled{6} = \circled{2}\circled{3} - \circled{4}\circled{5}$}
    \begin{itemize}
        \item \underline{\textbf{Case 1:} $b = 0$ (i.e., $|\mathcal{N}_{uv}| = 0$)}
        \begin{align}
            \circled{6} &= \left(1-\frac{4\alpha^2}{D}\left[\underset{h \sim g}{\mathbb{E}} \left[h^2\right] +(\underset{h \sim g}{\mathbb{E}} [h])^2\right]+\mathcal{O}\left(\frac{1}{D^2}\right) \right) \cdot 1 \\
            & \quad - \left(1-\frac{4\alpha^2}{D}\left[\underset{h \sim g}{\mathbb{E}} \left[h^2\right] -(\underset{h \sim g}{\mathbb{E}} [h])^2\right]+\mathcal{O}\left(\frac{1}{D^2}\right) \right) \cdot 1 \nonumber \\
            &= -\frac{8 \alpha^2}{D} \left( \underset{h \sim g}{\mathbb{E}} [h] \right)^2 + \mathcal{O}\left( \frac{1}{D^2}\right).
        \end{align}
        \item \underline{\textbf{Case 2:} $b > 0$ (i.e., $|\mathcal{N}_{uv}| > 0$)}
        \begin{align}
            \circled{6} &= \underset{h,h^{\prime} \sim g}{\mathbb{E}} \left[ \cos \left(2(h+h^{\prime})\frac{\alpha}{\sqrt{D}} \right) \right] \left( \underset{J,J^{\prime} \sim f}{\mathbb{E}} \left[ \cos \left(2(J+J^{\prime})\frac{\alpha}{\sqrt{D}} \right) \right]\right)^{bD^{\mu}}   \\
            & \quad -  \underset{h,h^{\prime} \sim g}{\mathbb{E}} \left[ \cos \left(2(h-h^{\prime})\frac{\alpha}{\sqrt{D}} \right) \right] \left( \underset{J,J^{\prime} \sim f}{\mathbb{E}} \left[ \cos \left(2(J-J^{\prime})\frac{\alpha}{\sqrt{D}} \right) \right]\right)^{bD^{\mu}}\\
            &=  \left[1 + \mathcal{O}\left( D^{-1}\right) \right] \times \begin{cases}
                e^{-4\alpha^2b \left( \underset{J \sim f}{\mathbb{E}}\left[J^2 \right] + \left( \underset{J \sim f}{\mathbb{E}}\left[J \right]\right)^2 \right)} \left(1 + \mathcal{O}\left(D^{-1}\right) \right) & \text{if } \mu = 1\\
                1 - 4\alpha^2b \left( \underset{J \sim f}{\mathbb{E}}\left[J^2 \right] + \left( \underset{J \sim f}{\mathbb{E}}\left[J \right]\right)^2 \right)D^{\mu - 1} + \mathcal{O}\left( D^{2\mu - 2}\right) &\text{if } 0 \leq \mu < 1
            \end{cases}\\
            &- \left[1 + \mathcal{O}\left( D^{-1}\right) \right] \times \begin{cases}
                e^{-4\alpha^2b \left( \underset{J \sim f}{\mathbb{E}}\left[J^2 \right] - \left( \underset{J \sim f}{\mathbb{E}}\left[J \right]\right)^2 \right)} \left(1 + \mathcal{O}\left(D^{-1}\right) \right) &\text{if } \mu = 1\\
                1 - 4\alpha^2b \left( \underset{J \sim f}{\mathbb{E}}\left[J^2 \right] - \left( \underset{J \sim f}{\mathbb{E}}\left[J \right]\right)^2 \right)D^{\mu - 1} + \mathcal{O}\left( D^{2\mu - 2}\right) &\text{if } 0 \leq \mu < 1.
            \end{cases}
        \end{align}
        When $\mu = 1$, $\circled{6}$ can be further simplified as follows:
        \begin{align}
            \circled{6} &= e^{-4\alpha^2b \left( \underset{J \sim f}{\mathbb{E}}\left[J^2 \right] + \left( \underset{J \sim f}{\mathbb{E}}\left[J \right]\right)^2 \right)} \hspace{-0.2cm} \left(1 + \mathcal{O}\left(D^{-1}\right) \right)^2 \hspace{-0.2cm} - e^{-4\alpha^2b \left( \underset{J \sim f}{\mathbb{E}}\left[J^2 \right] - \left( \underset{J \sim f}{\mathbb{E}}\left[J \right]\right)^2 \right)} \hspace{-0.2cm} \left(1 + \mathcal{O}\left(D^{-1}\right) \right)^2\\
            &= e^{-4\alpha^2b \left( \underset{J \sim f}{\mathbb{E}}\left[J^2 \right] + \left( \underset{J \sim f}{\mathbb{E}}\left[J \right]\right)^2 \right)} - e^{-4\alpha^2b \left( \underset{J \sim f}{\mathbb{E}}\left[J^2 \right] - \left( \underset{J \sim f}{\mathbb{E}}\left[J \right]\right)^2 \right)} + \mathcal{O}\left(D^{-1}\right) \\
            &= e^{-4\alpha^2b \underset{J \sim f}{\mathbb{E}}\left[ J^2 \right]}  \left[ e^{-4\alpha^2b \left(\underset{J \sim f}{\mathbb{E}}\left[ J \right]\right)^2} - e^{4\alpha^2b \left(\underset{J \sim f}{\mathbb{E}}\left[ J \right]\right)^2} \right] + \mathcal{O}\left(D^{-1}\right) \\
            &= -2 e^{-4\alpha^2b \underset{J \sim f}{\mathbb{E}}\left[ J^2 \right]}\sinh\left( 4\alpha^2b \left(\underset{J \sim f}{\mathbb{E}}\left[ J \right]\right)^2 \right) + \mathcal{O}\left(D^{-1}\right),
        \end{align}
        where in the last line, we have used the identity $\sinh x = \frac{1}{2}\left(e^x - e^{-x} \right)$.

        \noindent When $0 \leq \mu < 1$, $\circled{6}$ can be simplified as follows:
        \begin{align}
            \circled{6} &= \left[1 + \mathcal{O}\left( D^{-1}\right) \right] \left[ 1 - 4\alpha^2b \left( \underset{J \sim f}{\mathbb{E}}\left[J^2 \right] + \left( \underset{J \sim f}{\mathbb{E}}\left[J \right]\right)^2 \right)D^{\mu - 1} + \mathcal{O}\left( D^{2\mu - 2}\right) \right] \\
            & \quad - \left[1 + \mathcal{O}\left( D^{-1}\right) \right] \left[ 1 - 4\alpha^2b \left( \underset{J \sim f}{\mathbb{E}}\left[J^2 \right] - \left( \underset{J \sim f}{\mathbb{E}}\left[J \right]\right)^2 \right)D^{\mu - 1} + \mathcal{O}\left( D^{2\mu - 2}\right) \right] \\
            &= 1 - 4\alpha^2b \left( \underset{J \sim f}{\mathbb{E}}\left[J^2 \right] + \left( \underset{J \sim f}{\mathbb{E}}\left[J \right]\right)^2 \right)D^{\mu - 1} + \mathcal{O}\left( D^{-1}\right) + \mathcal{O}\left( D^{2\mu - 2}\right) \\
            & \quad - 1 - 4\alpha^2b \left( \underset{J \sim f}{\mathbb{E}}\left[J^2 \right] - \left( \underset{J \sim f}{\mathbb{E}}\left[J \right]\right)^2 \right)D^{\mu - 1} + \mathcal{O}\left( D^{-1}\right) + \mathcal{O}\left( D^{2\mu - 2}\right) \\
            &= -8 \alpha^2b \left(\underset{J \sim f}{\mathbb{E}}\left[J \right]\right)^2 + \mathcal{O}\left( D^{\max(-1, 2\mu-2)}\right).
        \end{align}
    \end{itemize}
    In summary, $\circled{6}$ can be expressed as follows:
    \begin{equation}
        \circled{6} = \begin{cases}
            -\frac{8 \alpha^2}{D} \left( \underset{h \sim g}{\mathbb{E}} [h] \right)^2 + \mathcal{O}\left( \frac{1}{D^2}\right) &\text{if } b = 0,\\
            -2 e^{-4\alpha^2b \underset{J \sim f}{\mathbb{E}}\left[ J^2 \right]}\sinh\left( 4\alpha^2b \left(\underset{J \sim f}{\mathbb{E}}\left[ J \right]\right)^2 \right) + \mathcal{O}\left(D^{-1}\right) &\text{if } b > 0, \, \mu = 1,\\
            -8 \alpha^2b \left(\underset{J \sim f}{\mathbb{E}}\left[J \right]\right)^2 + \mathcal{O}\left( D^{\max(-1, 2\mu-2)}\right) &\text{if } b > 0, \, 0 \leq \mu < 1.
        \end{cases}
    \end{equation}
\end{enumerate}

Having computed all the terms in $T_3$, we now begin to evaluate the possible expressions for $T_3$ for different possible values of $a$, $b$, $\lambda$, and $\mu$.
\begin{equation}
    \begin{aligned}
    T_3 &= -\frac{1}{2} \sin^2(2\beta) \underset{J \sim f}{\mathbb{E}}[J] \times \begin{cases}
            1 &\text{if } a = 0,\\[0.25cm]
            e^{-4a\alpha^2 \underset{J \sim f}{\mathbb{E}} \left[ J^2 \right]} \left(1 + \mathcal{O}\left( D^{-1}\right) \right) &\text{if } a > 0, \, \lambda = 1,\\[0.25cm]
            1- 4a\alpha^2 \underset{J \sim f}{\mathbb{E}} \left[ J^2 \right] D^{\lambda -1} \mathcal{O}\left(D^{2\lambda - 2} \right) &\text{if } a >0, \, 0 \geq \lambda < 1,
        \end{cases}\\
        & \quad \times \begin{cases}
        -\frac{8 \alpha^2}{D} \left( \underset{h \sim g}{\mathbb{E}} [h] \right)^2 + \mathcal{O}\left( \frac{1}{D^2}\right) &\text{if } b = 0,\\
        -2 e^{-4\alpha^2b \underset{J \sim f}{\mathbb{E}}\left[ J^2 \right]}\sinh\left( 4\alpha^2b \left(\underset{J \sim f}{\mathbb{E}}\left[ J \right]\right)^2 \right) + \mathcal{O}\left(D^{-1}\right) &\text{if } b > 0, \, \mu = 1,\\
        -8 \alpha^2b \left(\underset{J \sim f}{\mathbb{E}}\left[J \right]\right)^2 + \mathcal{O}\left( D^{\max(-1, 2\mu-2)}\right) &\text{if } b > 0, \, 0 \leq \mu < 1.
    \end{cases}
\end{aligned}
\end{equation}

\noindent Recall that $aD^{\lambda} + b D^{\mu} = D$. We can divide into four cases, the possible values of $a$, $b$, $\lambda$, and $\mu$.
\begin{enumerate}
    \item[(C1)] $a=0 \quad \Rightarrow \quad D=b D^\mu \quad \Rightarrow \quad b=\mu=1$
    \item[(C2)] $b=0 \quad \Rightarrow \quad D=a D^\lambda \quad \Rightarrow \quad a=\lambda=1$
    \item[(C3)] $\lambda<1 \quad \Rightarrow \quad \mu=1$ (otherwise $aD^{\lambda}+bD^{\mu} < D$ for some $D$)
    \item[(C4)] $\mu = 1 \quad \Rightarrow \quad \lambda=1$ (otherwise $aD^{\lambda}+bD^{\mu} < D$ for some $D$)
\end{enumerate}
The above gives rise to the following cases:
\begin{table}[H]
\centering
\renewcommand{\arraystretch}{1.2}
\setlength{\tabcolsep}{8pt}

\begin{tabular}{lccc}
\toprule
\diagbox[width=8em]{$b$}{$a$} &
  $a=0$ &
  $a>0,\:\lambda=1$ &
  $a>0,\:0 \leq \lambda < 1$ \\ 
\midrule
\midrule
$b = 0$ &
  \begin{tabular}[c]{@{}c@{}}%
    Contradicts (C1)\\
    $\because a=0 \Rightarrow b=1$%
  \end{tabular} &
  \begin{tabular}[c]{@{}c@{}}%
    Case 1\\
    $b=0,\;a=\lambda=1$%
  \end{tabular} &
  \begin{tabular}[c]{@{}c@{}}%
    Contradicts (C2)\\
    $\because b=0 \Rightarrow \lambda=1$%
  \end{tabular} \\
\midrule
$b>0,\;\mu=1$ &
  \begin{tabular}[c]{@{}c@{}}%
    Case 2\\
    $a=0,\;b=\mu=1$%
  \end{tabular} &
  \begin{tabular}[c]{@{}c@{}}%
    Case 3\\
    $a,b>0,\;a+b=1,\;\lambda=\mu=1$%
  \end{tabular} &
  \begin{tabular}[c]{@{}c@{}}%
    Case 4\\
    $a,b>0,\;\mu=1,\;0 \leq \lambda < 1$%
  \end{tabular} \\
\midrule
$b>0,\;0 \leq \mu < 1$ &
  \begin{tabular}[c]{@{}c@{}}%
    Contradicts (C1)\\
    $\because a=0 \Rightarrow \mu=1$%
  \end{tabular} &
  \begin{tabular}[c]{@{}c@{}}%
    Case 5\\
    $a,b>0,\;\lambda=1,\;0 \leq \mu < 1$%
  \end{tabular} &
  \begin{tabular}[c]{@{}c@{}}%
    Contradicts (C3)\\
    $\because \lambda<1 \Rightarrow \mu=1$%
  \end{tabular} \\
\bottomrule
\end{tabular}
\end{table}
As indicated above, four of the above cases result in contradictions. We will consider the remaining 5 cases separately.

\begin{itemize}
    \item \underline{\textbf{Case 1:} $b = 0, \quad a = \lambda = 1$ ($|\mathcal{N}_{u \doublesetminus v}|  = D$, $|\mathcal{N}_{uv}| = 0$, i.e. no triangles)}
    \begin{align}
        T_3 &= - \frac{1}{2} \sin^2(2 \beta) \underset{J \sim f}{\mathbb{E}}\left[J\right] e^{-4\alpha^2 \underset{J \sim f}{\mathbb{E}}\left[J^2\right]} \left( 1 + \mathcal{O}\left(D^{-1}\right)\right)\left[-\frac{8\alpha^2}{D} \left(\underset{h \sim g}{\mathbb{E}}\left[h\right]\right)^2 + \mathcal{O}\left(D^{-2}\right) \right] \\
        &= 4\alpha^2 \sin^2(2 \beta) \underset{J \sim f}{\mathbb{E}}\left[J\right] e^{-4\alpha^2 \underset{J \sim f}{\mathbb{E}}\left[J^2\right]}\left(\underset{h \sim g}{\mathbb{E}}\left[h\right]\right)^2 \frac{1}{D} + \mathcal{O}\left(D^{-2}\right).
    \end{align}
    \item \underline{\textbf{Case 2:} $a = 0, \quad b = \mu = 1$ ($|\mathcal{N}_{u \doublesetminus v}|  = 0$, $|\mathcal{N}_{uv}| = D$, i.e. Complete Graph)}
    \begin{align}
        T_3 &= - \frac{1}{2} \sin^2(2 \beta) \underset{J \sim f}{\mathbb{E}}\left[J\right] \left[ -2 e^{-4\alpha^2 \underset{J \sim f}{\mathbb{E}}\left[J^2\right]} \sinh\left(4 \alpha^2 \left(\underset{J \sim f}{\mathbb{E}}\left[J\right]\right)^2  \right) + \mathcal{O}\left(D^{-1}\right)\right] \\
        &= \sin^2(2 \beta) \underset{J \sim f}{\mathbb{E}} \left[J\right] e^{-4\alpha^2 \underset{J \sim f}{\mathbb{E}}\left[J^2\right]} \sinh\left(4 \alpha^2 \left(\underset{J \sim f}{\mathbb{E}}\left[J\right]\right)^2  \right) + \mathcal{O}\left(D^{-1}\right).
    \end{align}
    \item \underline{\textbf{Case 3:} $a, b > 0$, $\lambda = \mu = 1$, $a+b = 1$}
    \begin{align}
        T_3 &= - \frac{1}{2} \sin^2(2 \beta) \underset{J \sim f}{\mathbb{E}}\left[J\right] e^{-4a\alpha^2 \underset{J \sim f}{\mathbb{E}}\left[J^2\right]} \left( 1 + \mathcal{O}\left(D^{-1}\right)\right) \\
        & \quad \times \left[ -2 e^{-4b\alpha^2 \underset{J \sim f}{\mathbb{E}}\left[J^2\right]} \sinh\left(4b \alpha^2 \left(\underset{J \sim f}{\mathbb{E}}\left[J\right]\right)^2  \right) + \mathcal{O}\left(D^{-1}\right)\right] \nonumber \\
        &= \sin^2(2 \beta) \underset{J \sim f}{\mathbb{E}}\left[J\right] e^{-4a\alpha^2 \underset{J \sim f}{\mathbb{E}}\left[J^2\right]} e^{-4b\alpha^2 \underset{J \sim f}{\mathbb{E}}\left[J^2\right]}\sinh\left(4b \alpha^2 \left(\underset{J \sim f}{\mathbb{E}}\left[J\right]\right)^2  \right) + \mathcal{O}\left(D^{-1}\right) \\
        &= \sin^2(2 \beta) \underset{J \sim f}{\mathbb{E}}\left[J\right] e^{-4(a+b)\alpha^2 \underset{J \sim f}{\mathbb{E}}\left[J^2\right]} \sinh\left(4b \alpha^2 \left(\underset{J \sim f}{\mathbb{E}}\left[J\right]\right)^2  \right) + \mathcal{O}\left(D^{-1}\right) \\
        &= \sin^2(2 \beta) \underset{J \sim f}{\mathbb{E}}\left[J\right] e^{-4\alpha^2 \underset{J \sim f}{\mathbb{E}}\left[J^2\right]} \sinh\left(4b \alpha^2 \left(\underset{J \sim f}{\mathbb{E}}\left[J\right]\right)^2  \right) + \mathcal{O}\left(D^{-1}\right).
    \end{align}
    \item \underline{\textbf{Case 4:} $a, b > 0, \quad \mu = 1$, $0 \leq \lambda < 1$}
    \begin{align}
        T_3 &= - \frac{1}{2} \sin^2(2 \beta) \underset{J \sim f}{\mathbb{E}}\left[J\right] \hspace{-0.1cm} \left( \hspace{-0.05cm} 1 - 4a\alpha^2 \hspace{-0.1cm} \underset{J \sim f}{\mathbb{E}} \hspace{-0.1cm}\left[J^2\right] \hspace{-0.1cm} D^{\lambda-1} \hspace{-0.1cm} +   \mathcal{O} \hspace{-0.1cm} \left(D^{2\lambda-2}\right)\hspace{-0.15cm} \right) \\
        & \quad \times \left[ -2 e^{-4b\alpha^2 \underset{J \sim f}{\mathbb{E}}\left[J^2\right]} \hspace{-0.1cm} \sinh \hspace{-0.1cm} \left( \hspace{-0.1cm}4b \alpha^2 \hspace{-0.1cm} \left(\underset{J \sim f}{\mathbb{E}}\left[J\right] \hspace{-0.1cm} \right)^2  \right) \hspace{-0.1cm} + \hspace{-0.1cm} \mathcal{O} \hspace{-0.1cm} \left(D^{-1}\right)\hspace{-0.1cm} \right] \nonumber \\
        &= \sin^2(2 \beta) \underset{J \sim f}{\mathbb{E}}\left[J\right]e^{-4b\alpha^2 \underset{J \sim f}{\mathbb{E}}\left[J^2\right]} \sinh\left(4b \alpha^2 \left(\underset{J \sim f}{\mathbb{E}}\left[J\right]\right)^2  \right)\\
        &\quad- 4a\alpha^2 \sin^2(2 \beta) \underset{J \sim f}{\mathbb{E}}\left[J\right] \underset{J \sim f}{\mathbb{E}}\left[J^2\right]e^{-4b\alpha^2 \underset{J \sim f}{\mathbb{E}}\left[J^2\right]} \sinh\left(4b \alpha^2 \left(\underset{J \sim f}{\mathbb{E}}\left[J\right]\right)^2  \right)D^{\lambda-1} \nonumber \\
        & \quad + \mathcal{O}\left(D^{\max(2\lambda-2, -1)}\right). \nonumber
    \end{align}
    \item \underline{\textbf{Case 5:} $a, b > 0$, $\lambda = 1$, $0 \leq \mu < 1$}
    \begin{align}
        T_3 &= - \frac{1}{2} \sin^2(2 \beta) \underset{J \sim f}{\mathbb{E}}\left[J\right] e^{-4a\alpha^2 \underset{J \sim f}{\mathbb{E}}\left[J^2\right]} \left( 1 + \mathcal{O}\left(D^{-1}\right)\right) \\
        & \quad \times \left[-8b \alpha^2 \left( \underset{J \sim f}{\mathbb{E}}\left[J\right] \right)^2 D^{\mu - 1} + \mathcal{O}\left(D^{\max(-1, 2\mu - 2)} \right)\right] \nonumber \\
        &= 4b \alpha^2 \sin^2(2\beta) \left( \underset{J \sim f}{\mathbb{E}}\left[J\right] \right)^3 e^{-4a\alpha^2 \underset{J \sim f}{\mathbb{E}}\left[J^2\right]} D^{\mu - 1} + \mathcal{O}\left(D^{\max(-1, 2\mu - 2)} \right).
    \end{align}    
\end{itemize}
\vspace{0.25cm}
\noindent \underline{\textbf{Step 4: Computing the Sum of Terms $T_2$ and $T_3$}}\\
\newline
Recall that $T_2$ is computed as follows:
\begin{equation}
    T_2 = \sin(4\beta) \underset{J \sim f}{\mathbb{E}} \left[J^2\right] \frac{2\alpha}{\sqrt{D}}e^{-2 \alpha^2 \underset{J \sim f}{\mathbb{E}} \left[J^2\right]}+\mathcal{O}\left(D^{-1}\right).
\end{equation}
\begin{itemize}
    \item \underline{\textbf{Case 1:} $b = 0, \quad a = \lambda = 1$}
    \begin{align}
        T_2 + T_3 &= \sin(4\beta) \underset{J \sim f}{\mathbb{E}} \left[J^2\right] \frac{2\alpha}{\sqrt{D}}e^{-2 \alpha^2 \underset{J \sim f}{\mathbb{E}} \left[J^2\right]}+\mathcal{O}\left(D^{-1}\right)\\
        & \quad + 4\alpha^2 \sin^2(2 \beta) \underset{J \sim f}{\mathbb{E}}\left[J\right] e^{-4\alpha^2 \underset{J \sim f}{\mathbb{E}}\left[J^2\right]}\left(\underset{h \sim g}{\mathbb{E}}\left[h\right]\right)^2 \frac{1}{D} + \mathcal{O}\left(D^{-2}\right) \nonumber \\
        &= \sin(4\beta) \underset{J \sim f}{\mathbb{E}} \left[J^2\right] \frac{2\alpha}{\sqrt{D}}e^{-2 \alpha^2 \underset{J \sim f}{\mathbb{E}} \left[J^2\right]}+\mathcal{O}\left(D^{-1}\right).
    \end{align}
    \item \underline{\textbf{Case 2:} $a = 0, \quad b = \mu = 1$}
    \begin{equation}
        \begin{aligned}
        T_2 + T_3 &= \sin(4\beta) \underset{J \sim f}{\mathbb{E}} \left[J^2\right] \frac{2\alpha}{\sqrt{D}}e^{-2 \alpha^2 \underset{J \sim f}{\mathbb{E}} \left[J^2\right]}\\
        & \quad + \sin^2(2 \beta) \underset{J \sim f}{\mathbb{E}}\left[J\right] e^{-4\alpha^2 \underset{J \sim f}{\mathbb{E}}\left[J^2\right]} \sinh\left(4 \alpha^2 \left(\underset{J \sim f}{\mathbb{E}}\left[J\right]\right)^2  \right) + \mathcal{O}\left(D^{-1}\right).
    \end{aligned}
    \end{equation}
    \item \underline{\textbf{Case 3:} $a, b > 0, \quad \lambda = \mu = 1$, i.e. $a+b = 1$}
    \begin{equation}
        \begin{aligned}
        T_2 + T_3 &= \sin(4\beta) \underset{J \sim f}{\mathbb{E}} \left[J^2\right] \frac{2\alpha}{\sqrt{D}}e^{-2 \alpha^2 \underset{J \sim f}{\mathbb{E}} \left[J^2\right]}\\
        & \quad + \sin^2(2 \beta) \underset{J \sim f}{\mathbb{E}}\left[J\right] e^{-4\alpha^2 \underset{J \sim f}{\mathbb{E}}\left[J^2\right]} \sinh\left(4b \alpha^2 \left(\underset{J \sim f}{\mathbb{E}}\left[J\right]\right)^2  \right) + \mathcal{O}\left(D^{-1}\right).
    \end{aligned}
    \end{equation}
    \item \underline{\textbf{Case 4:} $a, b > 0, \quad \mu = 1, \quad 0 \leq \lambda < 1$}
    \begin{equation}
        \begin{aligned}
        T_2 + T_3 &= \sin(4\beta) \underset{J \sim f}{\mathbb{E}} \left[J^2\right] \frac{2\alpha}{\sqrt{D}}e^{-2 \alpha^2 \underset{J \sim f}{\mathbb{E}} \left[J^2\right]}+ \sin^2(2 \beta) \underset{J \sim f}{\mathbb{E}}\left[J\right]e^{-4b\alpha^2 \underset{J \sim f}{\mathbb{E}}\left[J^2\right]} \sinh\left(4b \alpha^2 \left(\underset{J \sim f}{\mathbb{E}}\left[J\right]\right)^2  \right)\\
        &\quad- 4a\alpha^2 \sin^2(2 \beta) \underset{J \sim f}{\mathbb{E}}\left[J\right] \underset{J \sim f}{\mathbb{E}}\left[J^2\right]e^{-4b\alpha^2 \underset{J \sim f}{\mathbb{E}}\left[J^2\right]} \sinh\left(4b \alpha^2 \left(\underset{J \sim f}{\mathbb{E}}\left[J\right]\right)^2  \right)D^{\lambda-1} \\
        & \quad + \mathcal{O}\left(D^{\max(2\lambda-2, -1)}\right).
    \end{aligned}
    \end{equation}
    \underline{\textbf{Case 4a:} $\lambda<\frac{1}{2}$}\\
    \newline
    \noindent Note that $\lambda<\frac{1}{2} \Leftrightarrow -\frac{1}{2} > \lambda -1$. Hence, $T_2 + T_3$ can be further simplified as follows:
    \begin{equation}
        \begin{aligned}
        T_2 + T_3 &= \sin(4\beta) \underset{J \sim f}{\mathbb{E}} \left[J^2\right] \frac{2\alpha}{\sqrt{D}}e^{-2 \alpha^2 \underset{J \sim f}{\mathbb{E}} \left[J^2\right]}\\
        & \quad + \sin^2(2 \beta) \underset{J \sim f}{\mathbb{E}}\left[J\right]e^{-4b\alpha^2 \underset{J \sim f}{\mathbb{E}}\left[J^2\right]} \sinh\left(4b \alpha^2 \left(\underset{J \sim f}{\mathbb{E}}\left[J\right]\right)^2  \right) + \mathcal{O}\left(D^{\lambda-1}\right).
    \end{aligned}
    \end{equation}
    \underline{\textbf{Case 4b:} $\lambda=\frac{1}{2}$}\\
    \newline
    \noindent Note that $\lambda=\frac{1}{2} \Leftrightarrow -\frac{1}{2} = \lambda -1$. Hence, $T_2 + T_3$ can be further simplified as follows:
    \begin{equation}
        \begin{aligned}
        T_2 + T_3 &= \sin^2(2 \beta) \underset{J \sim f}{\mathbb{E}}\left[J\right]e^{-4b\alpha^2 \underset{J \sim f}{\mathbb{E}}\left[J^2\right]} \sinh\left(4b \alpha^2 \left(\underset{J \sim f}{\mathbb{E}}\left[J\right]\right)^2  \right) \\
        &\quad + \underset{J \sim f}{\mathbb{E}} \left[J^2\right] \Bigg\{ 2\alpha \sin(4\beta) e^{-2 \alpha^2 \underset{J \sim f}{\mathbb{E}} \left[J^2\right]}\\
        &\quad - 4a\alpha^2 \sin^2(2 \beta) \underset{J \sim f}{\mathbb{E}}\left[J\right]e^{-4b\alpha^2 \underset{J \sim f}{\mathbb{E}}\left[J^2\right]} \sinh\left(4b \alpha^2 \left(\underset{J \sim f}{\mathbb{E}}\left[J\right]\right)^2  \right) \Bigg\} \frac{1}{\sqrt{D}}\\
        &\quad + \mathcal{O}\left(D^{\max(2\lambda-2, -1)}\right).
    \end{aligned}
    \end{equation}
    \underline{\textbf{Case 4c:} $\lambda>\frac{1}{2}$}\\
    \newline
    \noindent Note that $\lambda>\frac{1}{2} \Leftrightarrow -\frac{1}{2} < \lambda -1$. Hence, $T_2 + T_3$ can be further simplified as follows:
    \begin{equation}
        \begin{aligned}
        T_2 + T_3 &= \sin^2(2 \beta) \underset{J \sim f}{\mathbb{E}}\left[J\right]e^{-4b\alpha^2 \underset{J \sim f}{\mathbb{E}}\left[J^2\right]} \sinh\left(4b \alpha^2 \left(\underset{J \sim f}{\mathbb{E}}\left[J\right]\right)^2  \right)\\
        &\quad- 4a\alpha^2 \sin^2(2 \beta) \underset{J \sim f}{\mathbb{E}}\left[J\right] \underset{J \sim f}{\mathbb{E}}\left[J^2\right]e^{-4b\alpha^2 \underset{J \sim f}{\mathbb{E}}\left[J^2\right]} \sinh\left(4b \alpha^2 \left(\underset{J \sim f}{\mathbb{E}}\left[J\right]\right)^2  \right)D^{\lambda-1} \\
        &\quad + \mathcal{O}\left(\frac{1}{\sqrt{D}}\right).
    \end{aligned}
    \end{equation}
    \item \underline{\textbf{Case 5:} $a, b > 0, \quad \lambda = 1, \quad 0 \leq \mu < 1$}
    \begin{equation}
        \begin{aligned}
        T_2 + T_3 &= \sin(4\beta) \underset{J \sim f}{\mathbb{E}} \left[J^2\right] \frac{2\alpha}{\sqrt{D}}e^{-2 \alpha^2 \underset{J \sim f}{\mathbb{E}} \left[J^2\right]} + \mathcal{O}\left(D^{-1} \right)\\
        & \quad + 4b \alpha^2 \sin^2(2\beta) \left( \underset{J \sim f}{\mathbb{E}}\left[J\right] \right)^3 e^{-4a\alpha^2 \underset{J \sim f}{\mathbb{E}}\left[J^2\right]} D^{\mu - 1} + \mathcal{O}\left(D^{\max(-1, 2\mu - 2)} \right).
    \end{aligned}
    \end{equation}
    \underline{\textbf{Case 5a:} $\mu > \frac{1}{2}$}\\
    \newline
    \noindent Note that $\mu > \frac{1}{2} \Rightarrow \mu - 1> -\frac{1}{2}$. Now, $-\frac{1}{2} > 2\mu-2 \Leftrightarrow \mu < \frac{3}{4}$. Hence, $T_2 + T_3$ can be further simplified as follows:
    \begin{equation}
        T_2 + T_3 = 4b \alpha^2 \sin^2(2\beta) \left( \underset{J \sim f}{\mathbb{E}}\left[J\right] \right)^3 e^{-4a\alpha^2 \underset{J \sim f}{\mathbb{E}}\left[J^2\right]} D^{\mu - 1} + \mathcal{O}\left(D^{\max(-\frac{1}{2}, 2\mu - 2)} \right).
    \end{equation}
    \underline{\textbf{Case 5b:} $\mu = \frac{1}{2}$}
    \begin{equation}
        \begin{aligned}
        T_2 + T_3 &= \Bigg( 2\alpha\sin(4\beta) \underset{J \sim f}{\mathbb{E}} \left[J^2\right] e^{-2 \alpha^2 \underset{J \sim f}{\mathbb{E}} \left[J^2\right]}  + 4b \alpha^2 \sin^2(2\beta) \left( \underset{J \sim f}{\mathbb{E}}\left[J\right] \right)^3 e^{-4a\alpha^2 \underset{J \sim f}{\mathbb{E}}\left[J^2\right]}\Bigg)\frac{1}{\sqrt{D}}  + \mathcal{O}\left(D^{-1} \right).
    \end{aligned}
    \end{equation}
    \underline{\textbf{Case 5c:} $\mu < \frac{1}{2}$}\\
    \newline
    \noindent Note that $\mu < \frac{1}{2} \Rightarrow -1 \leq \mu - 1 < -\frac{1}{2}$. Hence, $T_2 + T_3$ can be further simplified as follows:
    \begin{equation}
    T_2 + T_3 = \sin(4\beta) \underset{J \sim f}{\mathbb{E}} \left[J^2\right] \frac{2\alpha}{\sqrt{D}}e^{-2 \alpha^2 \underset{J \sim f}{\mathbb{E}} \left[J^2\right]}+\mathcal{O}\left(D^{\mu-1}\right).
\end{equation}
\end{itemize}

\vspace{0.25cm}
\noindent \underline{\textbf{Step 5: Computing the Scaled Expected QAOA$_1$ Cost Function}}\\
\newline
Recall that the expected QAOA$_1$ cost function for an Ising model on a graph $G=(V, E)$ with parameters $\beta, \gamma \in \mathbb{R}$ is given by:
\begin{align}
    \underset{\substack{J \sim f^E \\ h \sim g^V}}{\mathbb{E}} \left[\xi_G(\gamma, \beta) \right] &= \sum_{i \in V} \underset{\substack{J \sim f^E \\ h \sim g^V}}{\mathbb{E}} \left[C_i(\gamma, \beta; G)\right] + \hspace{-0.4cm}\sum_{\{u, v\} \in E} \underset{\substack{J \sim f^E \\ h \sim g^V}}{\mathbb{E}} \left\{C_{u v}^{(1)}(\gamma, \beta ; G)+ \hspace{-0.45cm} \sum_{\{u, v\} \in E}C_{u v}^{(2)}(\gamma, \beta ; G)\right\} \\
    &= \sum_{i \in V} T_1 +\sum_{\{u, v\} \in E} \left( T_2 + T_3 \right) .
\end{align}
For a $(D+1)$-regular graph, the expectation terms $T_1, T_2$, and $T_3$ are identical for all vertices $i$ and edges $\{u, v\}$. This allows us to simplify the expression as follows:
\begin{equation}
    \underset{\substack{J \sim f^E \\ h \sim g^V}}{\mathbb{E}} \left[\xi_G(\gamma, \beta) \right] = T_1 |V| + (T_2+T_3)|E|,
\end{equation}
where $n=|V|$ is the number of vertices and $|E|$ is the number of edges in the graph $G$. Since the graph is $(D+1)$-regular, the number of edges satisfies $|E|=\frac{n}{2}(D+1)$. Substituting this into the equation, we obtain:
\begin{equation}
    \underset{\substack{J \sim f^E \\ h \sim g^V}}{\mathbb{E}} \left[\xi_G(\gamma, \beta) \right] = nT_1  + \frac{n}{2}(D+1)(T_2+T_3).
\end{equation}
To facilitate analysis, we introduce a scaled version of the cost function by dividing the expected cost by the number of edges $|E|$:
\begin{equation}
    \frac{1}{|E|}\underset{\substack{J \sim f^E \\ h \sim g^V}}{\mathbb{E}} \left[\xi_G(\gamma, \beta) \right] = \frac{|V|}{|E|} T_1 + (T_2+T_3).
\end{equation}
Calculating the ratio $|V|/|E|$ for a $(D+1)$-regular graph, we have:
\begin{equation}
    \frac{|V|}{|E|} = \frac{n}{\frac{n}{2}(D+1)} = \frac{2}{D+1}.
\end{equation}
Next, we evaluate the term $(|V|/|E|)T_1$ to get the following:
\begin{align}
    \frac{|V|}{|E|}T_1 &= \frac{2}{D+1}\left[ \sin(4\beta) \underset{h \sim g}{\mathbb{E}} \left[h^2\right] \frac{2\alpha}{\sqrt{D}}e^{-2 \alpha^2 \underset{J \sim f}{\mathbb{E}} \left[J^2\right]}+\mathcal{O}\left(\frac{1}{D}\right)\right] \\
    &= \mathcal{O}\left( D^{-\frac{3}{2}}\right).
\end{align}
Finally, substituting this result into the scaled cost function yields:
\begin{equation}
    \frac{1}{|E|}\underset{\substack{J \sim f^E \\ h \sim g^V}}{\mathbb{E}} \left[\xi_G(\gamma, \beta) \right] = T_2 + T_3 + \mathcal{O}\left( D^{-\frac{3}{2}}\right).
\end{equation}
This scaling demonstrates that as the degree $D$ of the graph increases, the contribution from $T_1$ becomes negligible, and the expected scaled cost function is primarily determined by the terms $T_2$ and $T_3$, with corrections of order $D^{-\frac{3}{2}}$. Therefore, the expressions derived in step 4 accurately characterise the scaled expected QAOA$_1$ cost function.

\vspace{0.25cm}
\noindent \underline{\textbf{Step 6: Optimal $\gamma^*$ Values for Leading-Order Terms}}\\
\newline
Before proceeding, we note that the subscripts for the expectation value are omitted, as the expressions in the remainder of this proof involve only the terms where $J \sim f$.

\vspace{0.25cm}
\noindent Not all expressions derived for the different cases yield closed-form solutions for the optimal $\alpha^*$ value. To address this, we focus on the leading-order terms, which simplify the expressions and allow for closed-form solutions for $\alpha^*$. The leading-order expressions for the scaled QAOA$_1$ objective are given by:
\begin{equation}
    \frac{1}{|E|}\mathbb{E} \left[\xi_G(\alpha, \beta) \right] = 
        \begin{cases}
            \mathcal{C}_1(\alpha, \beta) + \mathcal{O}\left(D^{-1}\right), & b = 0, \, a = \lambda = 1, \\
            \mathcal{C}_2(\alpha, \beta, 1, 1) + \mathcal{O}\left(1\right), & a = 0, \, b = \mu = 1, \\
            \mathcal{C}_2(\alpha, \beta, 1, b) + \mathcal{O}\left(1\right), & a, b > 0, \, \lambda = \mu = 1, \\
            \mathcal{C}_2(\alpha, \beta, b, b) + \mathcal{O}\left(1\right), & a, b > 0, \, \mu = 1, \, \lambda < 1,\\
            \mathcal{C}_3(\alpha, \beta) + \mathcal{O}\left(D^{\mu-1}\right), & a, b > 0, \, \lambda = 1, \, \mu > \frac{1}{2}, \\
            \mathcal{C}_1(\alpha, \beta) + \mathcal{O}\left(D^{\mu-1}\right), & a, b > 0, \, \lambda = 1, \, \mu < \frac{1}{2},
        \end{cases} 
\end{equation}
where the functions $\mathcal{C}_1$, $\mathcal{C}_2$, and $\mathcal{C}_3$ are defined as:
\begin{align}
    \mathcal{C}_1(\alpha, \beta) &= 2\alpha\sin(4\beta) \mathbb{E} \left[J^2\right] e^{-2 \alpha^2 \mathbb{E} \left[J^2\right]}\frac{1}{\sqrt{D}}, \\
    \mathcal{C}_2(\alpha, \beta, \theta_1, \theta_2) &= \sin^2(2 \beta) \mathbb{E}\left[J\right] e^{-4 \theta_1 \alpha^2 \mathbb{E}\left[J^2\right]} \sinh\left(4 \theta_2 \alpha^2 \mathbb{E}\left[J\right]^2  \right),\\
    \mathcal{C}_3(\alpha, \beta) &= 4b\alpha^2 \sin^2(2\beta) \mathbb{E}[J]^3 e^{-4a\alpha^2 \mathbb{E}[J^2]} D^{\mu-1}.
\end{align}
For the case $(a, b > 0, \, \lambda = 1, \, \mu = \frac{1}{2})$, the leading-order terms coincide with the original expression, but they do not admit a closed-form solution for $\alpha^*$. In contrast, for the cases $(b = 0, \, a = \lambda = 1)$, $(a, b > 0, \, \lambda = 1, \, \mu > \frac{1}{2})$, and $(a, b > 0, \, \lambda = 1, \, \mu < \frac{1}{2})$, the leading-order terms also coincide with the original expressions, but they do allow for closed-form solutions for $\alpha^*$.

\vspace{0.25cm}
\noindent We now proceed to compute the optimal closed-form values of the three functions $\mathcal{C}_1$, $\mathcal{C}_2$, and $\mathcal{C}_3$.

\vspace{0.25cm}
\noindent \underline{Optimising $\mathcal{C}_1$}\\
\newline
To determine the optimal $\alpha$ that minimises $\mathcal{C}_1$, we fix $\beta = 3\pi/8$, simplifying $\mathcal{C}_1$ to:
\begin{equation}
    \mathcal{C}_1(\alpha) = -2\alpha \mathbb{E} \left[J^2\right] e^{-2 \alpha^2 \mathbb{E} \left[J^2\right]}\frac{1}{\sqrt{D}}.
\end{equation}
The optimal $\alpha$ is found by differentiating $\mathcal{C}_1$ with respect to $\alpha$ and setting the derivative to zero:
\begin{align}
    \frac{d \, \mathcal{C}_1(\alpha)}{d \alpha} &= -2\mathbb{E} \left[J^2\right] \left( e^{-2 \alpha^2 \mathbb{E} \left[J^2\right]} - 4 \alpha^2 \mathbb{E} \left[J^2\right] e^{-2 \alpha^2 \mathbb{E} \left[J^2\right]} \right) \frac{1}{\sqrt{D}}\\
    &=-2\mathbb{E} \left[J^2\right] \left( 1 - 4 \alpha^2 \mathbb{E} \left[J^2\right]  \right) e^{-2 \alpha^2 \mathbb{E} \left[J^2\right]} \frac{1}{\sqrt{D}}.
\end{align}
Setting the derivative to zero yields:
\begin{equation}
1-4 \alpha^2 \mathbb{E}\left[J^2\right]=0 \quad \Rightarrow \quad \alpha^2=\frac{1}{4 \mathbb{E}\left[J^2\right]} \quad \Rightarrow \quad \alpha= \pm \frac{1}{2 \sqrt{\mathbb{E}\left[J^2\right]}}.
\end{equation}
Since $\alpha \geq 0$, the optimal value is:
\begin{equation}
    \alpha^* =  \frac{1}{2 \sqrt{\mathbb{E} \left[J^2\right]}}.
    \label{global_opt_eqn1}
\end{equation}

\vspace{0.25cm}
\noindent \underline{Optimising $\mathcal{C}_2$}\\
\newline
For $\mathcal{C}_2$, we set $\beta = \pi/4$, reducing the expression to:
\begin{equation}
    \mathcal{C}_2(\alpha, \theta_1, \theta_2) = \mathbb{E}\left[J\right] e^{-4 \theta_1 \alpha^2 \mathbb{E}\left[J^2\right]} \sinh\left(4 \theta_2 \alpha^2 \mathbb{E}\left[J\right]^2  \right). 
\end{equation}
To facilitate differentiation, we take the natural logarithm:
\begin{equation}
    \ln \left[\mathcal{C}_2(\alpha, \theta_1, \theta_2)\right] = \ln \left[\mathbb{E}\left[J\right]\right] -4 \theta_1 \alpha^2 \mathbb{E}\left[J^2\right] + \ln \left[ \sinh\left(4 \theta_2 \alpha^2 \mathbb{E}\left[J\right]^2  \right)\right]. 
\end{equation}
Differentiating with respect to $\alpha$ gives:
\begin{align}
    \frac{d \, \ln \left[\mathcal{C}_2(\alpha, \theta_1, \theta_2)\right]}{d\alpha} &= -8 \theta_1 \alpha \mathbb{E}\left[J^2\right] + 8 \theta_2 \alpha \mathbb{E}\left[J\right]^2 \coth\left(4 \theta_2 \alpha^2 \mathbb{E}\left[J\right]^2  \right)  .
\end{align}
Setting this derivative to zero results in:
\begin{equation}
    -8 \theta_1 \alpha \mathbb{E}\left[J^2\right] + 8 \theta_2 \alpha \mathbb{E}\left[J\right]^2 \coth\left(4 \theta_2 \alpha^2 \mathbb{E}\left[J\right]^2  \right)  = 0.
\end{equation}
It is evident that $\alpha = 0$ is a critical point. To identify other critical points, we divide through by $8 \alpha$, assuming $\alpha \neq 0$, and solve for $\alpha^2$:
\begin{align}
     - \theta_1 \mathbb{E}\left[J^2\right] +  \theta_2 \mathbb{E}\left[J\right]^2 \coth\left(4 \theta_2 \alpha^2 \mathbb{E}\left[J\right]^2  \right)  &= 0\\
     \coth\left(4 \theta_2 \alpha^2 \mathbb{E}\left[J\right]^2  \right) - \dfrac{\theta_1 \mathbb{E}\left[J^2\right]}{\theta_2 \mathbb{E}\left[J\right]^2} &= 0 \\
     4 \theta_2 \alpha^2 \mathbb{E}\left[J\right]^2  &= \coth^{-1} \left(\dfrac{\theta_1 \mathbb{E}\left[J^2\right]}{\theta_2 \mathbb{E}\left[J\right]^2}\right) \\
     4 \theta_2 \alpha^2 \mathbb{E}\left[J\right]^2  &= \frac{1}{2}\ln \left( \frac{\frac{\theta_1 \mathbb{E}\left[J^2\right]}{\theta_2 \mathbb{E}\left[J\right]^2} + 1}{\frac{\theta_1 \mathbb{E}\left[J^2\right]}{\theta_2 \mathbb{E}\left[J\right]^2} - 1} \right) \\
     4 \theta_2 \alpha^2 \mathbb{E}\left[J\right]^2  &= \frac{1}{2}\ln \left( \frac{\theta_1 \mathbb{E}\left[J^2\right] + \theta_2 \mathbb{E}\left[J\right]^2}{\theta_1 \mathbb{E}\left[J^2\right] - \theta_2 \mathbb{E}\left[J\right]^2} \right)\\
     \alpha^2 &= \frac{1}{8\theta_2 \mathbb{E}\left[J\right]^2} \ln \left( \frac{\theta_1 \mathbb{E}\left[J^2\right] + \theta_2 \mathbb{E}\left[J\right]^2}{\theta_1 \mathbb{E}\left[J^2\right] - \theta_2 \mathbb{E}\left[J\right]^2} \right).
\end{align}
As $\alpha \geq 0$, the optimal value for $\alpha > 0$ is:
\begin{equation}
    \alpha^* =  \sqrt{\frac{1}{8\theta_2 \mathbb{E}\left[J\right]^2} \ln \left( \frac{\theta_1 \mathbb{E}\left[J^2\right] + \theta_2 \mathbb{E}\left[J\right]^2}{\theta_1 \mathbb{E}\left[J^2\right] - \theta_2 \mathbb{E}\left[J\right]^2} \right)}.
    \label{global_opt_eqn2}
\end{equation}

\vspace{0.25cm}
\noindent \underline{Optimising $\mathcal{C}_3$}\\
\newline
For $\mathcal{C}_3$, we set $\beta = \pi/4$, reducing the expression to:
\begin{equation}
    \mathcal{C}_3(\alpha) = 4b\alpha^2 \mathbb{E}[J]^3 e^{-4a\alpha^2 \mathbb{E}[J^2]} D^{\mu-1}.
\end{equation}
Differentiating with respect to $\alpha$ gives:
\begin{align}
    \frac{d \, \mathcal{C}_3(\alpha)}{d\alpha} &= 4b \mathbb{E}[J]^3 \left[ 2\alpha - 4a\alpha^3 \mathbb{E}[J^2] \right] e^{-4a\alpha^2 \mathbb{E}[J^2]} D^{\mu-1}.
\end{align}
Setting the derivative to zero yields:
\begin{equation}
    2\alpha - 4a\alpha^3 \mathbb{E}[J^2] = 0.
\end{equation}
It is evident that $\alpha = 0$ is a critical point. To identify other critical points, we divide through by $2 \alpha$, assuming $\alpha \neq 0$, and solve for $\alpha$:
\begin{align}
    1 - 2a\alpha^2 \mathbb{E}[J^2] &= 0 \\
    \alpha &= \pm \frac{1}{\sqrt{2a\mathbb{E}[J^2]}}.
\end{align}
As $\alpha \geq 0$, the optimal value for $\alpha > 0$ is:
\begin{equation}
    \alpha = \frac{1}{\sqrt{2a\mathbb{E}[J^2]}}.
    \label{global_opt_eqn3}
\end{equation}

\vspace{0.25cm}
\noindent \underline{Optimal $\alpha^*$ Values}\\
\newline
The closed-form expressions for the optimal $\alpha^* \in \mathbb{R}^+$ values for the different cases are given by:
\begin{equation}
    \alpha^* = 
        \begin{cases}
            \frac{1}{2\sqrt{\mathbb{E}\left[J^2\right]}} , & b = 0, \, a = \lambda = 1, \\
            \sqrt{\frac{1}{8  \mathbb{E}\left[J\right]^2} \ln  \left( \frac{\mathbb{E}\left[J^2\right] + \mathbb{E}\left[J\right]^2}{\mathbb{E}\left[J^2\right] - \mathbb{E}\left[J\right]^2} \right)} , & a = 0, \, b = \mu = 1, \\
            \sqrt{\frac{1}{8 b  \mathbb{E}\left[J\right]^2} \ln  \left( \frac{ \mathbb{E}\left[J^2\right] + b \mathbb{E}\left[J\right]^2}{ \mathbb{E}\left[J^2\right] - b \mathbb{E}\left[J\right]^2} \right)} , & a, b > 0, \, \lambda = \mu = 1, \\
            \sqrt{\frac{1}{8b  \mathbb{E}\left[J\right]^2} \ln  \left( \frac{ \mathbb{E}\left[J^2\right] +  \mathbb{E}\left[J\right]^2}{ \mathbb{E}\left[J^2\right] -  \mathbb{E}\left[J\right]^2} \right)} , & a, b > 0, \, \mu = 1, \, \lambda < 1,\\
            \frac{1}{\sqrt{2a\mathbb{E}\left[J^2\right]}} , & a, b > 0, \, \lambda = 1, \, \mu > \frac{1}{2}, \\
            \frac{1}{2\sqrt{\mathbb{E}\left[J^2\right]}},  & a, b > 0, \, \lambda = 1, \, \mu < \frac{1}{2}.
        \end{cases} 
\end{equation}
Using the relation $\gamma = \alpha / \sqrt{D}$, the closed-form expressions for the optimal $\gamma^* \in \mathbb{R}^+$ values can be readily obtained.

\vspace{0.5cm}
\noindent \underline{\textbf{Step 7: Globally Optimal $\alpha^*$ Coincides with First Local Optimum}}\\
\newline
In this final part of the proof, we demonstrate that the derived globally optimal value $\alpha^*$ coincides with the first local optimum of the functions $\mathcal{C}_1$, $\mathcal{C}_2$, and $\mathcal{C}_3$.

\vspace{0.25cm}
\noindent \underline{Demonstrating Global Optimum Coincides with First Local Optimum for $\mathcal{C}_1$}\\
\newline
We want to show that the globally optimal parameter \cref{global_opt_eqn1} for $\mathcal{C}_1$ coincides with the first local optimum. Recall that the derivative of $\mathcal{C}_1$ is given by:
\begin{equation}
    \frac{d \, \mathcal{C}_1(\alpha)}{d \alpha}  = -2\mathbb{E} \left[J^2\right] \left( 1 - 4 \alpha^2 \mathbb{E} \left[J^2\right]  \right) e^{-2 \alpha^2 \mathbb{E} \left[J^2\right]} \frac{1}{\sqrt{D}}.
\end{equation}
We begin by analysing the behaviour of $\mathcal{C}_1^{\prime}(\alpha)$ at $\alpha = 0$. Substituting $\alpha = 0$ into the derivative, we find that
\begin{equation}
    \mathcal{C}_1^{\prime}(0)  = -2\mathbb{E} \left[J^2\right] \frac{1}{\sqrt{D}} < 0,
\end{equation}
which implies that $\mathcal{C}_1(\alpha)$ is decreasing at $\alpha = 0$. Next, we consider the interval $\alpha \in\left(0, \alpha^*\right)$. In this range, the term $1-4 \alpha^2 \mathbb{E}_{\mathrm{J}}\left[J^2\right]$ remains positive because $\alpha^2<\left(\alpha^*\right)^2$. Consequently, the derivative satisfies:
\begin{equation}
    -2\mathbb{E} \left[J^2\right] \left( 1 - 4 \alpha^2 \mathbb{E} \left[J^2\right]  \right) e^{-2 \alpha^2 \mathbb{E} \left[J^2\right]} \frac{1}{\sqrt{D}} < 0.
\end{equation}
This implies that $\mathcal{C}_1(\alpha)$  continues to decrease throughout the interval $\left(0, \alpha^*\right)$. At $\alpha=\alpha^*$, the term $1-4 \alpha^2 \mathbb{E}_{\mathbf{J}}\left[J^2\right]$ becomes zero and so $\mathcal{C}_1^{\prime}\left(\alpha^*\right)$ evaluates to $0$ thus confirming that $\alpha^*$ is a stationary point.

Finally, we examine the behaviour of $\mathcal{C}_1^{\prime}(\alpha)$ for $\alpha>\alpha^*$. In this region, the term $1-4 \alpha^2 \mathbb{E}_{\mathbf{J}}\left[J^2\right]$ becomes negative because $\alpha^2>\left(\alpha^*\right)^2$. As a result, the derivative satisfies:
\begin{equation}
-2\mathbb{E} \left[J^2\right] \left( 1 - 4 \alpha^2 \mathbb{E} \left[J^2\right]  \right) e^{-2 \alpha^2 \mathbb{E} \left[J^2\right]} \frac{1}{\sqrt{D}} > 0.
\end{equation}
This indicates that $\mathcal{C}_1(\alpha)$ is increasing for $\alpha>\alpha^*$.

From this analysis, we observe that $\mathcal{C}_1(\alpha)$ decreases for $\alpha \in\left[0, \alpha^*\right)$, reaches a stationary point at $\alpha=\alpha^*$, and increases for $\alpha>\alpha^*$. This behaviour establishes that $\alpha^*$ is the first local minimum of $\mathcal{C}_1(\alpha)$ for $\alpha \geq 0$. 

\vspace{0.25cm}
\noindent \underline{Demonstrating Global Optimum Coincides with First Local Optimum for $\mathcal{C}_2$}\\
\newline
We want to show that the globally optimal parameter \cref{global_opt_eqn2} for $\mathcal{C}_2$ coincides with the first local optimum. Recall that the derivative of $\mathcal{C}_2$ with respect to $\alpha$ is:
\begin{equation}
\mathcal{C}_2^{\prime}(\alpha, \theta_1, \theta_2) = 8 \alpha \mathbb{E}[J] e^{-4 \theta_1 \alpha^2 \mathbb{E}\left[J^2\right]}\left[\theta_2 \mathbb{E}[J]^2 \cosh \left(4 \theta_2 \alpha^2 \mathbb{E}[J]^2\right) -\theta_1 \mathbb{E}\left[J^2\right] \sinh \left(4 \theta_2 \alpha^2 \mathbb{E}[J]^2\right)\right],
\end{equation}
where $(\theta_1, \theta_2)$ can be $(1,1)$, $(1,b)$, or $(b,b)$ with $0 \leq b \leq 1$. Additionally, since the variance is non-negative, we have $\mathbb{E}[J^2] \geq \mathbb{E}[J]^2$.

\vspace{0.25cm}
\noindent The case $\mathbb{E}[J] > 0$ is trivial because the global optimum that minimises $\mathcal{C}_2$ is $\alpha^* = 0$ which clearly is the first local optimum. Thus, we focus on $\mathbb{E}[J] < 0$. Consider $\alpha \in (0, \alpha^*)$. Since $\alpha^2 < (\alpha^*)^2$, we can approximate the hyperbolic functions for small arguments:
\begin{equation}
\cosh (x) \approx 1+\frac{x^2}{2}, \quad \sinh (x) \approx x+\frac{x^3}{6}.
\end{equation}
Applying these approximations, the derivative becomes:
\begin{align}
     \mathcal{C}_2^{\prime} &\approx 8 \alpha \mathbb{E}[J] e^{-4 \theta_1 \alpha^2 \mathbb{E}\left[J^2\right]}\left[\theta_2 \mathbb{E}[J]^2\left(1+\frac{\left(4 \theta_2 \alpha^2 \mathbb{E}[J]^2\right)^2}{2}\right)-\theta_1 \mathbb{E}\left[J^2\right]\left(4 \theta_2 \alpha^2 \mathbb{E}[J]^2+\frac{\left(4 \theta_2 \alpha^2 \mathbb{E}[J]^2\right)^3}{6}\right)\right]\\
    &\approx 8 \alpha \mathbb{E}[J] e^{-4 \theta_1 \alpha^2 \mathbb{E}\left[J^2\right]}\left[\theta_2 \mathbb{E}[J]^2-4 \theta_1 \theta_2 \alpha^2 \mathbb{E}\left[J^2\right] \mathbb{E}[J]^2\right]\\
    &= 8 \alpha \theta_2 \mathbb{E}[J]^3 e^{-4 \theta_1 \alpha^2 \mathbb{E}\left[J^2\right]}\left[1 - 4\theta_2 \alpha^2 \mathbb{E}\left[J^2\right] \right].
\end{align}
Since $\alpha^2 < (\alpha^*)^2$, the term $1 - 4 \theta_2 \alpha^2 \mathbb{E}[J^2]$ remains positive. Given that $\mathbb{E}[J] < 0$ and $\theta_2 \geq 0$, the derivative $\mathcal{C}_2^{\prime} < 0$. This implies that $\mathcal{C}_2$ is decreasing on the interval $(0, \alpha^*)$.

\vspace{0.25cm}
\noindent For $\alpha > \alpha^*$, assuming $\alpha$ is still small enough for the above approximations to hold, the term $1 - 4 \theta_2 \alpha^2 \mathbb{E}[J^2]$ becomes negative because $\alpha^2 > (\alpha^*)^2$. Consequently, $\mathcal{C}_2^{\prime} > 0$, indicating that $\mathcal{C}_2$ is increasing for $\alpha > \alpha^*$.

\vspace{0.25cm}
\noindent For sufficiently large $\alpha \gg \alpha^*$, we approximate the hyperbolic functions as:
\begin{equation}
\cosh (x) \approx \sinh (x) \approx \frac{e^x}{2}.
\end{equation}
Substituting these into the derivative gives:
\begin{equation}
    \mathcal{C}_2^{\prime} \approx 4 \alpha \mathbb{E}[J] e^{-4 \theta_1 \alpha^2 \mathbb{E}\left[J^2\right]}e^{4 \theta_2 \alpha^2 \mathbb{E}[J]^2}\left(\theta_2 \mathbb{E}[J]^2-\theta_1 \mathbb{E}\left[J^2\right]\right).
\end{equation}
Given that $\theta_2 \leq \theta_1$ and $\mathbb{E}[J^2] \geq \mathbb{E}[J]^2$, we analyse the expression $\theta_2 \mathbb{E}[J]^2 - \theta_1 \mathbb{E}[J^2]$:
\begin{itemize}
    \item If $\theta_1 \mathbb{E}[J^2] > \theta_2 \mathbb{E}[J]^2$, then $\theta_2 \mathbb{E}[J]^2 - \theta_1 \mathbb{E}[J^2] < 0$. 
    \item If $\theta_1 \mathbb{E}[J^2] = \theta_2 \mathbb{E}[J]^2$, then $\theta_2 \mathbb{E}[J]^2 - \theta_1 \mathbb{E}[J^2] = 0$.
\end{itemize}
Thus, for $\alpha \gg \alpha^*$, the derivative $\mathcal{C}_2^{\prime} > 0$, confirming that $\mathcal{C}_2$ continues to increase.

\vspace{0.25cm}
\noindent The analysis shows that $\mathcal{C}_2(\alpha)$ is decreasing for $\alpha \in (0, \alpha^*)$, attains a stationary point at $\alpha = \alpha^*$, and increases for $\alpha > \alpha^*$. This behaviour confirms that $\alpha^*$ is the first local minimum of $\mathcal{C}_2(\alpha)$ for $\alpha \geq 0$. Therefore, the globally optimal parameter \cref{global_opt_eqn2} coincides with this first local optimum.

\vspace{0.25cm}
\noindent \underline{Demonstrating Global Optimum Coincides with First Local Optimum for $\mathcal{C}_3$}\\
\newline
Since the proof that the globally optimal value $\alpha^*$, defined in \cref{global_opt_eqn3}, coincides with the first local optimum closely parallels the argument presented for $\mathcal{C}_1$, we omit it here for the sake of brevity.
\end{proof}

\vfill

\end{document}